\documentclass[11pt]{article}

\newif\ifnotation

\usepackage[a4paper,top=20.5mm,bottom=32mm]{geometry} 
\usepackage{amsmath}
\usepackage{arydshln}
\usepackage{amsfonts}

\usepackage{amsthm}
\newtheorem{definition}{Definition}

\usepackage{thm-restate} 
\declaretheorem{theorem}
\declaretheorem{proposition}
\declaretheorem{lemma}

\usepackage[utf8]{inputenc}
\usepackage{graphicx}
\usepackage{paralist}
\usepackage[T1]{fontenc}
\usepackage{tikz}
\usetikzlibrary{positioning,arrows.meta}
\usetikzlibrary{fit}
\usepackage{hyperref}
\usepackage{xurl} 

\newcommand\eps{\varepsilon}
\newcommand\eeta{f} 
\newcommand\blank{\texttt{\char32}}

\begin{document}

\title{Probabilistic Finite Automaton Emptiness is Undecidable}
\author{Günter Rote\thanks{Freie Universität Berlin,
    Institut für Informatik, Takustr.~
    9, 14195 Berlin, \texttt{rote€inf.fu-berlin.de}}}
\maketitle
\vspace {-3ex}
\begin{abstract}
  It is undecidable whether the language {recognized} by a
  probabilistic finite automaton is empty.  Several other
  undecidability results, in particular regarding problems about
  matrix products, are based on this important theorem.
  We present three proofs of this theorem from the literature in a
  self-contained way, and we derive some strengthenings.  For example,
  we show that the problem remains undecidable for a fixed
  probabilistic finite automaton with 11~states, where only the
  starting distribution is given as input.
\end{abstract}

\thispagestyle{empty}

\tableofcontents

\ifnotation\else\hrule height 0pt\vspace{5.5ex}\hrule height 0pt \fi

\section{Probabilistic finite automata (PFA)}
\label{sec:PFA}

A probabilistic finite automaton (PFA)
combines characteristics of
a finite automaton and a Markov chain.
We give a formal definition
below. 
Informally, we can think of a PFA in terms of an algorithm that
reads a sequence of input symbols 
from left to right,
having only finite memory. That is, it can manipulate a
finite number of variables with bounded range, just like an ordinary
finite automaton. In addition, a
PFA 
can make coin flips.
As a consequence,
the question whether the PFA arrives in an accepting state and thus accepts a given input word 
is not a yes/no decision, but it happens with a certain {probability}.
The language \emph{recognized}
(or \emph{represented})
by a PFA is defined by specifying a
probability threshold or
\emph{cutpoint} $\lambda$. By convention, the language consists of all words for which the
probability of acceptance strictly exceeds~$\lambda$.

The \emph{PFA Emptiness Problem} is the problem of deciding whether
this language is empty.

This problem is undecidable.
There are three different 
proofs of this
theorem in the literature.
The first is by
Masakazu Nasu and Namio Honda~%
\cite{NasuHonda1969} from 1969.
The second proof, by Volker Claus \cite{claus81} is loosely related to
this proof.
A completely independent proof, which
 uses a very different approach,
was given by
Anne Condon and Richard J. Lipton~%
\cite{condon-lipton-1989} in 1989,
based on ideas of
R{\={u}}si{\c{n}}{\v{s}}
Freivalds~\cite{freivalds81} from 1981.
The somewhat intricate 
history is 
described
in Section~\ref{sec:history}.

We will present these three proofs
in Sections \ref {sec:nasu-honda-claus},
\ref{sec:thm-variable-matrices},
and \ref {sec:2c}, respectively.
The chains of reductions are shown in
Figure~\ref{fig:reductions} in Section~\ref{sec:shortcut}.
A self-contained proof of the basic undecidability result
(Proposition~\ref{=1/4})
takes about
3 pages, see Section \ref {sec:nasu-honda-claus}.
The rest of the paper is devoted to 
different
sharpenings of the undecidability statement,
where certain parameters of the PFA are restricted 
(Theorems~\ref{thm-condon-lipton}--\ref{thm-variable-matrices-binary}).

 \subsection{Formal problem definition}
\label{sec:formal}

 Formally, a PFA is given by a sequence of stochastic \emph{transition
 matrices} $M_\sigma$, 
 one for each letter $\sigma$ from the input alphabet $\Sigma$.
 The matrices are $d\times d$ matrices if the PFA has $d$~states.
 The start state is chosen according
 to a given
 probability distribution $\pi\in \mathbb R^d$.
 The set of accepting
 states is characterized by a 0-1-vector $\eeta \in \{0,1\}^{d}$.
%

In terms of these data, the PFA Emptiness Problem
with cutpoint $\lambda$, whose undecidability we will show, can be
formally described as follows.

\goodbreak

\begin{quote}
  \textsc{PFA Emptiness.}
  Given  a finite set of 
  stochastic matrices
 $\mathcal{M} \subset  \mathbb Q^{d\times d}$,
a probability distribution $\pi\in \mathbb Q^d$,
and a 0-1-vector $\eeta \in \{0,1\}^{d}$, 
is there a sequence
$M_1,M_2,\ldots, M_m$ with $M_j\in\mathcal{M}$ 
such that
\begin{equation}\label{eq:accept}
  \pi^T\!M_1M_2\ldots M_m\eeta > \lambda \ ?
\end{equation}
\goodbreak
\end{quote}
The most natural choice is $\lambda=\frac12$, but
the problem is undecidable for
any fixed (rational or irrational) cutpoint $\lambda$ with $0<
\lambda< 1$.  We can also ask $\ge\lambda$ instead of $>\lambda$.

Our results, which we discuss in the next section, show that the PFA
Emptiness Problem remains undecidable under
additional restrictions.
Table~\ref{tab:results}
gives an overview of the various assumptions and
constraints on the data. 
 Theorem~\ref{thm-condon-lipton} is essentially due to Condon and Lipton.
 Theorem~\ref{thm-claus-9-states} is due to Claus, and
Theorem~\ref{thm-variable-matrices-binary} is due to Mika Hirvensalo~\cite{hirvensalo07}.

\begin{table}[htb]
  \centering
  \newcommand\inp{\emph{input}}
  \begin{tabular}{|l|crlcl|}
    \hline
    Theorem&$\pi$&\!\!\!\!\!\!
                   $|\mathcal{M}|$&$M \in \mathcal{M}$&$\eeta$&acceptance criterion\\
    \hline
    \hline
    Thm.~\ref{thm-condon-lipton}
           & $\pi= e_2$ & 2 & \inp & $\eeta=e_1$ 
                                                       & any (Thm.~\ref{dichotomy})\\
    \hline
Thm.~\ref{thm-fixed-f}a
           & \inp & 52 & $18\times18$, positive 
                                               &$\!\!\eeta\in\{0,1\}^{18}$\! & $\ge1/2$\\
    Thm.~\ref{thm-fixed-f}b
           & \inp & 53 & $11\times11$ & $\eeta=e_1$ & $>1/4$\\
\hline    
Thm.~\ref{thm-fixed-f-2-matrices}
           & \inp & 2 & $572\times572$ & $\eeta=e_1$ & $>1/4$\\
\hline    
Thm.~\ref{thm-fixed-pi}a
           & $0<\pi_q<1$
                   & 52 & $9\times9$, positive & \inp & $\ge1/2$\\
Thm.~\ref{thm-fixed-pi}b
           & $0\le\pi_q\le1$ & 52 & $11\times11$ & \inp & $>1/4$\\
\hline    
    Thm.~\ref{thm-claus-9-states}
           & $\pi= e_2$ & 5 & \inp, $9\times 9$, posit\rlap{ive} & $\eeta=e_1$ 
                           & $>1/9$ 
    \\
Thm.~\ref{thm-variable-matrices-binary}
           & $\pi=e_2$ & 2 & \inp, $20\times 20$, posit\rlap{ive}
  & $\eeta=e_1$ &  $>1/20$ or $\ge1/20$ \\
\hline       
  \end{tabular}
  
  \caption{The main characteristics of the data $\pi$,
    $\mathcal{M}$, and $\eeta$ for 
 different
    undecidable versions of PFA Emptiness. The data that are not
    marked as \inp\ are fixed.
    The vectors
    $e_1$ and $e_2$ are two standard unit vectors of appropriate dimension.
  }
  \label{tab:results}
\end{table}

\section{Statement of results}
\label{sec:results}

 The PFA Emptiness Problem is undecidable even
if the start state is a fixed (deterministic) state,
and there is a single accepting state (different from the start state).
In this case,
$\pi$ is a standard unit vector, consisting of a single 1 and otherwise zeros,
and likewise, $\eeta $ is a standard unit vector.
The acceptance probability is found in a specific entry
(say, the upper right corner) of the product
$M_1M_2\ldots M_m$.


\begin{restatable} 
{theorem}
{CondonLiptonTwoMatrices}
  \label{thm-condon-lipton}
 For any fixed 
 $\lambda$ with $0<
\lambda< 1$, the
PFA Emptiness Problem \eqref{eq:accept} with
 cutpoint~$\lambda$
 is undecidable, even when restricted to instances
where $\mathcal{M}$ consists of
only two transition matrices,
all of whose entries are from the set $\{0,\frac12,1\}$,
and $\pi$ and $\eeta $ are standard unit vectors.
\end{restatable}

The proof is given in Section~\ref{sec:2c}.

We mention that we don't have to rely on a sharp
distinction between $\le\lambda$ and $>\lambda$, because
 the PFA that is constructed in the proof exhibits is a strong
 separation property
 (see Theorem~\ref{dichotomy} in Section~\ref{sec:boost},
 and Section~\ref{sec:further-boost}):
 Either there is a sequence of matrices 
 for which
the product
$  \pi^T\!M_1M_2\ldots M_m\eeta$ exceeds 
$1-\eps$, 
or, for every sequence, the value is below~$\eps$, 
where $\eps$ be chosen arbitrarily close to 0.

The remaining results deal with the case
where all matrices in $\mathcal M$
are fixed.

\begin{definition}
By a \emph{binary fraction}, we mean a
rational number whose denominator is a power of~$2$.
\end{definition}

\goodbreak
\begin{restatable} 
{theorem}
{NasuHondaEleven}
\label{thm-fixed-f}\
\begin{enumerate}[\rm (a)]

\item
  There is a fixed set $\mathcal{M}'$ of 52
stochastic matrices
of size $18\times 18$
with positive
entries that are multiples of
$1/2^{47}$,
and a fixed vector $f\in\{0,1\}^{18}$,
  for which the following question is undecidable\textup:

  Given a probability distribution $\pi \in \mathbb Q^{18}$ whose
  entries are positive binary fractions,
  is there a product
  $M_1M_2\ldots M_m$, with $M_j\in\mathcal{M}'$ for all $j=1,\ldots,m$,
with
$$\pi^T\!M_1M_2\ldots M_m\eeta \ge \tfrac 12\ ?$$

\item 
  There is a
  fixed set $\mathcal{M}$ of 53
stochastic matrices
of size $11\times 11$,
all of whose
entries are multiples of
$1/2^{47}$,
  for which the following question is undecidable\textup:

  Given a probability distribution $\pi\in \mathbb Q^{11}$ whose
  entries are 
  {binary fractions},
 is there a product
  $M_1M_2\ldots M_m$, with $M_j\in\mathcal{M}$ for all $j=1,\ldots,m$,
such that
$$\pi^T\!M_1M_2\ldots M_me_1 > \tfrac 14\ ?$$
In other words, is the language recognized by the PFA with starting
distribution $\pi$ and cutpoint $\lambda=\frac14$ nonempty? 

\end{enumerate}
\end{restatable}
In part (b),
$e_1$ denotes the first unit vector in $\mathbb R^{11}$, meaning
that there is a single accepting state.
The proof is given in section~\ref{eleven}.

Part (b) of the theorem has the acceptance criterion $>\frac14$, in
line with the conventions for a PFA. Part (a) deviates from this
convention by using a weak inequality $\ge\frac12$, but this is
rewarded
by allowing a stronger assumption:
All matrices in $\mathcal{M}$
are strictly positive.

The distinction between the
cutpoint values $\frac12$ and $\frac14$ in parts (a) and (b)
is inessential.
In fact,
for all of the Theorems~\ref{thm-fixed-f}--\ref{thm-fixed-pi},
 the cutpoint
 can be set to any fixed rational value within some range, but then the
 assumption that all entries are binary fractions must be given up,
 and the size of the matrices must sometimes be increased.

 An easier version of
 Theorem~\ref{thm-fixed-f}b,
 but with matrices of size
 $12\times 12$, is proved in Section~\ref{sec:use-UTM}
(Proposition~\ref{thm-fixed-f-weak}).

The input alphabet can be reduced to
two symbols at the expense of the number of states.
The proof will be given in section~\ref{sec:coding}.

\begin{restatable} 
{theorem}
{NasuHondaTwoMatrices}
\label{thm-fixed-f-2-matrices}
There is a PFA with 572 states, two input symbols
with fixed transition matrices,
all of whose
entries are multiples of
$1/2^{47}$,
and with a single accepting state,
  for which the following question is undecidable\textup:

  Given a probability distribution $\pi\in \mathbb Q^{572}$ whose
  entries are 
  {binary fractions},
is the language recognized by the PFA with starting
distribution $\pi$ and cutpoint $\lambda=\frac14$ nonempty? 
%
\end{restatable}

\paragraph{More general acceptance.}
If each state $q$ is allowed to have an arbitrary
probability
$\eeta 
_q 
$
as an ``acceptance degree''
instead of just 0 or 1, 
we can also turn things
around
and 
fix the
starting distribution $\pi$, but let 
the values $\eeta _q$
be part of the input.
The following theorem will be proved in Section~\ref{sec:vary-f}.

\goodbreak

\begin{restatable} 
{theorem}
{NasuHondaFixedMatricesReverse}
\label{thm-fixed-pi}\
\begin{enumerate}[\rm (a)]

\item
  There is a fixed set $\mathcal{M}'''$ of 52
positive
stochastic matrices
of size $9\times 9$
and a fixed 
starting distribution $\pi$,
all with positive
entries that are multiples of
$1/2^{44}$,
  for which the following question is undecidable\textup:

  Given a vector $\eeta \in \mathbb Q^{9}$ whose
  entries are binary fractions from the interval $[\frac14,\frac58]$,
  is there a product
  $M_1M_2\ldots M_m$, with $M_j\in\mathcal{M}'''$ for all $j=1,\ldots,m$,
with
$$\pi^T\!M_1M_2\ldots M_m\eeta \ge \tfrac 12\ ?$$

\item 
  There is a fixed set $\mathcal{M}''$ of 52
stochastic matrices
of size $11\times 11$
and a fixed starting distribution $\pi$,
all of whose
entries are multiples of
$1/2^{45}$,
  for which the following question is undecidable\textup:

  Given a vector $\eeta \in \mathbb Q^{11}$ whose
  entries are binary fractions from the interval $[0,1]$,
  is there a product
  $M_1M_2\ldots M_m$, with $M_j\in\mathcal{M}''$ for all $j=1,\ldots,m$,
such that
$$\pi^T\!M_1M_2\ldots M_m\eeta > \tfrac 14\ ?$$

\end{enumerate}
\end{restatable}

The distinction between parts (a) and (b) is analogous to
Theorem~\ref{thm-fixed-f}.  This time, part~(a) has an
additional advantage:
In addition to the  positivity of all data in 
$\mathcal{M}$,
$\pi$, 
and
~$\eeta$, 
the dimension is reduced from 11 to~9.

\paragraph{Uniqueness of solutions.}
We mention that 
Theorems~\ref{thm-fixed-f}--\ref{thm-fixed-pi},
can be modified such that the solutions of the constructed matrix
product problem instances
are unique if they exist,
see 
 Section~\ref{sec:unique}.
In other words, we are guaranteed that the
language recognized by the PFA contains at most one word.  This extension
requires a slightly larger number of matrices with larger denominators
in its entries.


\paragraph{Smaller transition matrices.}

In 1981, Volker Claus \cite{claus81} investigated
what he called the
$(n, k)$-bounded Emptiness Problem, for PFAs with at most $n$ states
and an alphabet of size at most $k$.
He derived the undecidability
for certain parameter pairs
$(n, k)$
from the PCP. In 1981, it was known that
the PCP is undecidable with as few as 9 word pairs.
Meanwhile, we know
by results of Neary~\cite{neary:PCP5:2015} from 2015
that 5 word pairs are sufficient.
With this improved bound, the result of Claus reads as follows.

\begin{restatable}[{Claus \cite[Theorem 6(iii), p.~151]{claus81}}]
  {theorem}
  {ClausSmallMatrices}
  \label{thm-claus-9-states}
The PFA Emptiness Problem \eqref{eq:accept}
is undecidable for PFAs with
a deterministic start state,
 a single accepting state, and $5$ positive transition matrices of size
$9\times 9$, and with cutpoint
$\lambda=1/9$.
\end{restatable}

The same approach
to find
 small matrices for which the
 PFA Emptiness Problem
 is unsolvable
was
used in 2003
by
Blondel and Canterini \cite{blondel-canterini-03}, who
concentrated on instances with \emph{two} matrices, or 
in other words, PFAs with a binary input alphabet.
These results were improved by Hirvensalo
\cite{hirvensalo07}\rlap,\footnote{see also
  the technical report \cite{tHirvensalo06a}}
who showed in 2007
that the PFA Emptiness Problem is undecidable
for two transition matrices of size $25\times 25$.
%
Substituting the improved bound of Neary~\cite{neary:PCP5:2015} on the number of word pairs for which the
PCP is undecidable, five states can be saved.

\begin{restatable}[
\cite{hirvensalo07}]
  {theorem}
  {HirvensaloAdaptedSmallMatrices}
  \label{thm-variable-matrices-binary}
The PFA Emptiness Problem \eqref{eq:accept}
is undecidable for PFAs with two positive transition matrices size
$20\times 20$, a single deterministic start state, a single accepting state, and with cutpoint
$\lambda=1/20$.

The same is true with weak inequality $(\ge1/20)$ as the  acceptance criterion.
\end{restatable}

We give the proofs of these theorems in Section~\ref{sec:thm-variable-matrices}.
We mention that the reduction to a binary input alphabet for the PFA
(i.e., two matrices)
was already considered by Claus~\cite
{claus81}, but
his results are superseded by Theorem~\ref{thm-variable-matrices-binary}.
Claus also has results for alphabets of size 3 and 4 
\cite[Theorem 6(iii), p.~151]{claus81}.
Conversely,
constructions of PFAs with few states regardless of the
number of matrices 
are implicit in the proofs of \cite{blondel-canterini-03} and
\cite{hirvensalo07}, but they are not as strong as Theorem~\ref{thm-claus-9-states}.

\section{Preface: history and matrix products}
\label{sec:history}

\subsection{Three proofs}
The study of probabilistic finite automata was initiated 
by
Michael Rabin in 1963~\cite{rabin1963}.
While this was 
an 
active research area in the 1960's,
it is less active today.
The first proof that
PFA Emptiness is undecidable is
due to Masakazu
Nasu and Namio
Honda from 1969
\cite[Theorem 21, p.~270]{NasuHonda1969}.
It
proceeds through a series of lemmas 
with 
tricky constructions,
showing
that more and more classes of 
languages, including certain types of context-free languages, can be
recognized by a PFA.
Eventually,
 the undecidability of the PFA Emptiness Problem is derived from
 Post's Correspondence Problem (PCP, see Section~\ref{sec:PCP}).
 The proof is reproduced in the final part of a monograph
by Azaria Paz from 1971
  \cite[Theorem 6.17 in Section IIIB, p.~190]
  {paz71}.
The presentation is  
quite close to the original, but very much condensed (and it never
cares to mention the PCP by name!).
I suppose, as the result was still recent 
when the book was written, 
it was the culmination point of the monograph. It appears as part of the last
theorem of the book, before a brief 
final chapter on applications and generalizations.
In the literature, as far as I have surveyed it, the result is almost universally
misattributed to Paz, although Paz gave credit
to Nasu and Honda
(not very specifically, however)
in the closing remarks of the chapter
\cite[Section IIIB.7, Bibliographical notes, p.~193]{paz71}.\footnote
{ ``The rest of that section beginning from Exercise 6.9 and on is based
 on the work of Nasu and Honda (1970).''
 The year 1970 is a mistake.}
%
%
%
%
A simplified version of
this proof appears 
in a
textbook of
Volker Claus from 1971 \cite[Satz 28, p.~157]{claus71}
in German, with proper attribution to
 Nasu and Honda.

 A second proof is due to Claus \cite{claus81} from 1981.
It has been practically forgotten until now.
 It also
 takes the PCP as the starting point, but it proceeds via products of
 integer matrices, which were first used in this context by Mike
 Paterson~\cite{paterson70} in 1970.
As mentioned earlier, Claus was interested in
constraints on the number of states and the size of the input alphabet
for which the Emptiness Problem for PFAs
 remains undecidable.
These results
were 
later rediscovered and improved by
 Blondel and Canterini \cite{blondel-canterini-03} in 2003
 and further improved by
Hirvensalo in 2007
\cite{hirvensalo07}.

A completely independent proof was sketched by Anne Condon and Richard Lipton in
1989~\cite{condon-lipton-1989}. It arose as an auxiliary 
result for their
investigation of space-bounded interactive proofs.  Condon and Lipton
based their reduction on the undecidability of the Halting Problem for 2-Counter
Machines (2CM),
see Section~\ref{sec:2c} below.

\subsection{Interlude: Other problems on matrix products}

As the formulation
\eqref{eq:accept}
shows,
the PFA Emptiness Problem is 
about products of matrices
that can be taken from a given set $\mathcal{M}$. There are other
problems
of this type, 
whose
undecidability comes down to PFA Emptiness: For example,
the \emph{joint spectral radius} of a set 
$\mathcal{M}$ 
of $d\times d$ matrices is
\begin{displaymath}
  \limsup_{m\to\infty}
    \max_{A_1,A_2,\ldots,A_m\in \mathcal{M}}
  \sqrt[m]{
      \| A_1A_2\ldots A_m\|},
  \end{displaymath}
  where $\|\cdot\|$ denotes an arbitrary 
  norm.
In 2000,
Blondel and Tsitsiklis~\cite{blondel2000} proved,
based on the PFA Emptiness Problem,
that it is
undecidable
whether the
joint spectral radius of a finite set of rational matrices exceeds~1.

This has recently been generalized in the analysis of the growth rate of
\emph{bilinear systems}, see
Matthieu Rosenfeld~\cite{Rosenfeld21}
and
Vuong Bui
\cite{bui2021growth
  ,bui23-phd}. 
The study of bilinear systems was initiated for a special case
of such a system in
Rote~\cite{rote2019maximum} in the context of a combinatorial
counting problem.
Corresponding decidability questions are discussed in
Rosenfeld~\cite{rosenfeld2022undecidable}
and
Bui~\cite[Chapter~6]{bui23-phd}.
These connections were my motivation for
starting the investigations about the PFA Emptiness Problem.

In fact,
Theorem~\ref{thm-fixed-pi}a,
which strengthens the
undecidability result of
PFA Emptiness to \emph{positive} transition matrices,
can be used to
resolve a conjecture of
Bui 
\cite[Conjecture 6.7]{bui23-phd}, by
adapting the reduction of
Blondel and Tsitsiklis~\cite{blondel2000}:
Already for 
two \emph{positive} matrices, 
it is undecidable to check
if their joint spectral radius 
is larger than~$1$.


\subsection{\dots\ back to the proofs of PFA Emptiness:}

In 2000,
Blondel and Tsitsiklis~\cite{blondel2000},
not being aware of the paper of Claus \cite{claus81} from 1981,
could arguably complain that
\emph{a complete proof that
PFA Emptiness is undecidable cannot be found in its
entirety in the published literature}.
Since then, Condon and Lipton's proof
has been published in  sufficient detail in other papers, for example by
Madani, Hanks, and Condon
\cite[Sec.~3.1 
and Appendix A
]{jair03} in 2003.
%
Moreover, in the publication list on Anne Condon's homepage,
the entry for the Condon--Lipton conference paper~\cite{condon-lipton-1989} from 1989
links to a
22-page manuscript, dated November 29, 2005.\footnote{\url{https://www.cs.ubc.ca/~condon/papers/condon-lipton89.pdf},
accessed 2024-05-01.}
According to the metadata, the file was generated on that date by
the dvips program from a file called
``journalsub.dvi''.
This manuscript also gives the proof in 
detail.
 Condon and Lipton's proof, which is based on ideas of Freivalds,
 is conceptually simple and illuminating.
 The current article 
 originated from lecture notes about this proof.

Meanwhile, I struggled with
Nasu and Honda's article 
and tried to penetrate through its rendition in Paz
\cite{paz71},
which proceeds through a cascade of definitions and lemmas that
stretch
over the whole book.
When I had already acquired a rough understanding of some crucial ideas,
I was lucky to hit upon 
the 
undecidability 
proof in the
textbook of
Claus 
\cite[Satz 28, p.~154--157]{claus71},
which is considerably simplified.
The result in
\cite
{claus71} is weaker, because the number of
input symbols is the number $k$ of word pairs of the PCP, whereas
Nasu and Honda establish undecidability already for an input alphabet
of size~2.
It is, however, easy to
reduce the input alphabet, see Lemma~\ref{coding}.
(Nasu and Honda's technique for achieving this reduction is
considerably more involved, see
Section~\ref{sec:reflection-binary} and
Appendix~\ref{sec:nhp}.)

After finishing a previous version of this note, I learned about
another undecidability proof of the
PFA Emptiness Problem
 by Gimbert and Oualhadj from 2010
\cite[Section~2.1]{gimbert-oualhadj-2010:PFA}.
Compared to the Nasu--Honda proof,
it treats the 
equality test $\phi(a)=\psi(a)$ (Section~\ref{check:equality})
differently, but otherwise it is similar to the proof
that I present
(Proposition~\ref{=1/4} in
Section \ref {sec:nasu-honda-claus}).
The authors
misattribute the main auxiliary result of their proof to
a paper of Alberto Bertoni from 1975
\cite{Bertoni1975}.
More details are given in Section~\ref{sec:bertoni}.

By tracing the literature back from
\cite{gimbert-oualhadj-2010:PFA},
I discovered the paper of Claus \cite{claus81} from 1981,
which is about the Emptiness Problem
expressly for PFAs with few states and with small input alphabet.
It takes the opposite route: Rather than using PFAs to prove results
about matrix products, it starts from problems involving products of
integer matrices and converts them to PFAs

\subsection{Overview}
In this note, 
I try to present the best parts of the three proofs in a
self-contained way.
I use slightly
different terminology, and
some details vary 
from constructions found elsewhere. 
I have preferred 
concrete formulations with
particular values of the parameters,
illustrating them with examples. Generalizations
to arbitrary parameters 
are treated as an afterthought.
I have made an effort to streamline the proofs.
In particular, the complete Nasu--Honda--Claus proof leading
to the main undecidability result
of Proposition~\ref{=1/4} takes only 3 pages
(Section \ref {sec:nasu-honda-claus}, pp.~\pageref{sec:binary}--\pageref{=1/4}), and I encourage the reader to jump
directly to this section.
In later parts, I will incrementally introduce new
ideas that decrease the number of states or deal with variants of
the problem, and the treatment
becomes more technical.
For reference, I review the original Nasu--Honda proof in
Appendix~\ref{sec:nhp}.

\subsection{Comparison of the proofs}
The proofs
use different ideas, and they
have different merits:
 Condon and Lipton's proof leads to an arbitrarily large gap between
 accepting and rejecting probabilities
 (Theorem~\ref{dichotomy}
 and Section~\ref{sec:further-boost})\footnote
 {However, there is a general technique by which the acceptance by strict
   inequality $>\lambda$ can be amplified to a gap between $\le 1/2$
   and arbitrarily close to 1,
   see
 Section~\ref{sec:amplification-study} and
   the remark at the end of Section~\ref{equality-testing}.}
 and it is easy to restrict the input alphabet to 2 symbols
 (Theorem~\ref{thm-condon-lipton}).
The number of states is beyond control. 
 While the constructions in
  Condon and Lipton's proof use
a high-level description of a PFA
as a randomized algorithm,
the proof of Nasu and Honda 
encourages to work with the transition matrices 
directly,
and consequently, allows a finer control over the number of states.

Moreover, by looking at the reductions in 
detail, one can even show undecidability of the
 Emptiness Problem for a \emph{fixed} PFA with 11 states and an input
 alphabet of size 53, where the only variable input is the starting distribution
(Theorem~\ref{thm-fixed-f}b). 
This and similar sharpenings of the undecidability statement 
are the 
 contributions of this paper in terms of new results,
and we hope they might find other applications.
A variation of the problem allows as few as 9 states
(Theorem~\ref{thm-fixed-pi}a). 

An alternative approach is somewhat similar in spirit, but it takes a detour via integer matrices,
which are only in the end converted to stochastic matrices
via Turakainen's
Theorem~\cite{turakainen69} (see the introduction of
Section~\ref{sec:fix-start}).
This also allows as few as 9 states
(Theorem~\ref{thm-claus-9-states}).

The distinction between the two main proof approaches is highlighted for
a particular example, the language
$\{\,\texttt{a}^i\texttt{b}^i\texttt{\#}\mid i\ge0\,\}$,
in Section~\ref{equality-testing}.

\section{The Condon--Lipton proof via  2-counter machines}
\label{sec:2c}

This section presents the proof of Condon and
Lipton~\cite{condon-lipton-1989} from 1989, leading to
Theorem~\ref{thm-condon-lipton}.

A 
\emph{counter machine} 
has
a finite control, 
represented by a state $q$ from a finite set~$Q$,
and a number of 
nonnegative counters. 
There is a designated
start state and a designated halting 
state.
Such a machine operates as follows. At each step,
it checks which counters are zero. Depending on the outcome of these tests and the
current state $q$, it may
increment or decrement each counter by~1, and it
enters a new state.

A counter machine with as few as two
counters
(a 2CM)
is as powerful as
a Turing machine.
This was first proved by Marvin Minsky~\cite{Minsky61} in 1961 and is by now
 textbook knowledge
\cite[Theorem 7.9]{hopcroft79}.\footnote
{\label{simulate-TM-2CM}%
  The usual way to simulate a Turing machine by a 2CM proceeds in
  three easy steps: (i) A two-sided infinite
tape can be simulated by two push-down stacks.
(ii) A push-down stack can be simulated by two counters, interpreting
the stack contents as digits for an appropriate radix that is large
enough to accommodate the stack alphabet;
two counters are necessary to perform multiplication and division by
the radix.
(iii) Any number of counters can be simulated by two counters,
representing the values $a,b,c,d,\ldots$ of the counters as
a product $2^a3^b5^c7^d\ldots$ of prime powers.
See \url{https://en.wikipedia.org/wiki/Counter_machine\#Two-counter_machines_are_Turing_equivalent_(with_a_caveat)},
accessed 2024-04-13.}
The question whether such a 2-counter machine halts
if it is started with both counter values at 0 is undecidable.


Denoting by $q_i,l_i,r_i$ the state and the values of the two counters after $i$ steps,
an accepting computation with $m$ steps can be written as follows:
\begin{displaymath}
  l_0,r_0,q_0,l_1,r_1,q_2,l_2,r_2,q_3, \ldots ,l_{m-1},r_{m-1},q_m
\end{displaymath}
To turn it into an input for a finite automaton, we encode it as a word $A$ over the alphabet 
$Q\cup \{\texttt{0},\texttt{1},\texttt{\#}\}$ with an end marker
\texttt{\#}:
\begin{equation}
  \label{eq:computation} A =
    \texttt{0}^{l_0}\texttt{1}^{r_0}
  q_0\texttt{0}^{l_1}\texttt{1}^{r_1}{q_1}\texttt{0}^{l_2}\texttt{1}^{r_2}{q_2}
  \dots\texttt{0}^{l_{m}}\texttt{1}^{r_{m}}q_m  \texttt{\#}
\end{equation}
There are some conditions
for an accepting computation
that 
a deterministic finite automaton can easily check:
Does the word
conform to this format?
Do the state transitions follow the rules?
Is $l_0={r_0}=0$?
Is
the initial and the final (halting) state correct?
We refer to these checks
as the \emph{formal checks}.

The only thing that a finite automaton cannot check
is the consistency of
the counters, for example, whether $l_{i+1}$ 
is equal to
$l_i$, or $l_i+1$, or $l_i-1$,
as appropriate.

For this task, we use the probabilistic capacities of the PFA.  If
there is an accepting computation $A$ of the
form \eqref{eq:computation} for the counter machine, we feed this computation as input to the
PFA again and again. In other words, we input the word $A^t$ for a
large enough~$t$. We will set up the PFA in such a way that
there is a strong separation of probabilities: It will
accept this input with probability at least $0.99$.  On the other
hand, if there is no accepting computation, then
every input will be rejected
 with probability at least
$0.99$.

\subsection{The Equality Checker}
\label{sec:eq}

As an auxiliary procedure, we study a PFA that reads words of the form
$\texttt{a}^i\texttt{b}^j\texttt{\#}$. The goal is to ``decide'' whether
$i=j$. We call this procedure the \emph{Equality Checker}.
There are three possible outcomes,
``Different'',
``Same'',
or ``Undecided''.

The PFA simulates a competition between two players $D$ and $S$
(``Different'' and
``Same'', or ``Double'' and ``Sum''), as shown in Figure~\ref
{example-equality-checker}.
There are four unbiased
coins of different colors.
\begin{itemize}
\item Player $D$ flips the red coin twice for each \texttt{a}
  and the orange coin twice for each \texttt{b}.
\item Player $S$ flips the blue coin and the green coin for each
  input symbol (\texttt{a} or 
   \texttt{b}).
\end{itemize}

\begin{figure}[htb]
  \centering
  \includegraphics[scale=1.1]{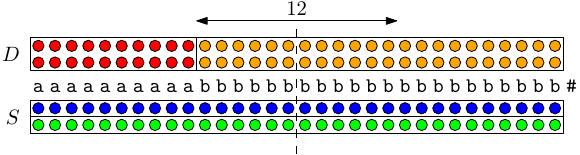}
  \caption{The coin flips for the input
    $\texttt{a}^{10}\texttt{b}^{22}\texttt{\#}$}
    \label{example-equality-checker}
\end{figure}

In addition,
the PFA keeps track of the difference $i-j$ modulo
$
12$.
If $i \not\equiv j \pmod{12}$,
the PFA declares the outcome to be ``Different''.

If $i \equiv j \pmod{12}$,
the outcome of the game is defined as follows.
We call a coin \emph{lucky} if it
always came up heads.
\begin{itemize}
\item If $D$ has a lucky coin and $S$ has no lucky coin, declare ``Different''.
\item If $S$ has a lucky coin and $D$ has no lucky coin, declare ``Same''.
\item Otherwise, declare ``Undecided''.
\end{itemize}
Since $i$ and $j$ are usually large, lucky 
means
\emph{extremely lucky}. Thus, the first two events are very rare, and
the outcome will almost always be ``Undecided''.  The outcome of the
Equality Checker is illustrated in Figure~\ref{fig:EC} and described
in the following lemma.

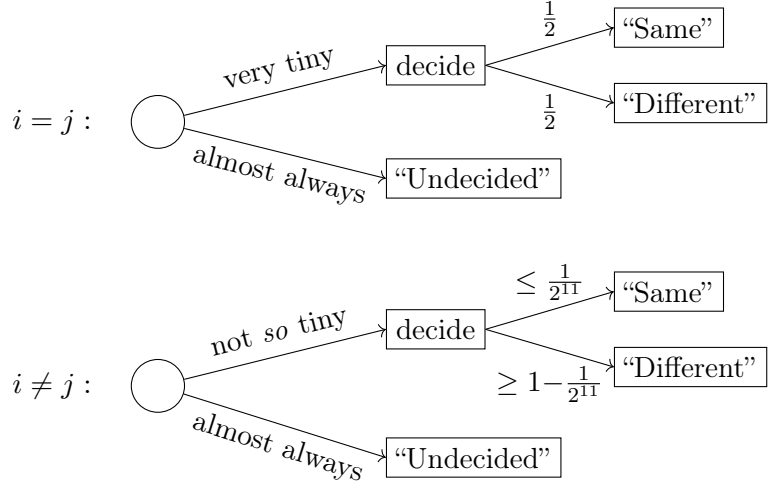
\begin{figure}[htb]
  \centering
\begin{tikzpicture}[
roundnode/.style={circle, draw=black, minimum size=7mm},
squarednode/.style={rectangle, draw=red!60, fill=red!5, very thick, minimum size=5mm},
boxnode/.style={rectangle, draw=black, minimum size=5mm},
]

\node[boxnode,anchor=west] at (1,2.5)     (undec1)             {``Undecided''};
\node[boxnode,anchor=west] at (4,4.5)     (eq1)             {``Same''};
\node[boxnode,anchor=west] at (4,3.5)     (diff1)             {``Different''};
\node[boxnode,anchor=west] at (1,4)        (dec1)       {decide};
\node[roundnode, label={left:$i=j:\quad $}] at (-2,3.25)        (start1)     {}  ;

\draw[->] (start1) -- (dec1.west) node[above,midway,sloped]{very tiny};
\draw[->] (start1) -- (undec1.west) node[below,midway,sloped]{almost always};
\draw[->] (dec1.east) -- (eq1.west) node[midway,above]{$\frac12$};
\draw[->] (dec1.east) -- (diff1.west) node[midway,below]{$\frac12$};

\node[boxnode,anchor=west] at (1,-1.2)     (undec2)             {``Undecided''};
\node[boxnode,anchor=west] at (4,1)     (eq2)             {``Same''};
\node[boxnode,anchor=west] at (4,0)     (diff2)             {``Different''};
\node[boxnode,anchor=west] at (1,0.5)        (dec2)       {decide};
\node[roundnode,label={left:$i\ne j:\quad $}] at (-2,-0.25)        (start2)     {}  ;

\draw[->] (start2) -- (dec2.west) node[above,midway,sloped]{not \emph{so} tiny};
\draw[->] (start2) -- (undec2.west) node[below,midway,sloped]{almost always};
\draw[->] (dec2.east) -- (eq2.west) node[midway,above]{$\le\frac1{2^{11}}$};
\draw[->] (dec2.east) -- (diff2.west) node[midway,below,yshift=-1mm]{$\ge1{-}\frac1{2^{11}}$};


\end{tikzpicture}

  \caption{The behavior of the Equality Checker, assuming
  $i \equiv j \pmod{12}$}
  \label{fig:EC}
\end{figure}

\begin{lemma}\ \label{EC-outcome}
  \begin{itemize}
  \item 
  If $i=j$,
  $\Pr[\textup{``Different''}] = \Pr[\textup{``Same''}]$.
  \item 
  If $i\ne j$,
  $\Pr[\textup{``Different''}] \ge 2^{11}\cdot \Pr[\textup{``Same''}]$.
  \end{itemize}
\end{lemma}


\begin{proof}
The first statement is clear, since each coin is flipped $2i$ times, and the situation between $D$ and $S$
is symmetric.

Assume that $i\ne j$.  
If $i \not\equiv j \pmod{12}$, then
$\Pr[\textup{``Same''}]=0$, and we are done.

Otherwise, $|i-j|\ge 12$, and the smaller of $i$ and $j$, say $i$, is
at most
$i\le \frac{i+j}2 - 6$. Then the red coin is flipped at most $2i\le
i+j-12$ times.
Thus,
\begin{align}
    \label{D-lucky}
  Pr[\text{$D$ has a lucky coin}]
 \ge Pr[\text{the red coin was lucky}]
 \ge 1/2^{i+j-12}
\end{align}
The blue and the green coin was each flipped $i+j$ times,
and hence
\begin{multline}
  \label{S-lucky}
  Pr[\text{$S$ has a lucky coin}]
  \le \\
  Pr[\text{the blue coin was lucky}]
  + Pr[\text{the green coin was lucky}]
 \le 2/2^{i+j}
\end{multline}
The ratio
$
{  Pr[\text{$D$ lucky}]} / {  Pr[\text{$S$ lucky}]}$
between \eqref{D-lucky} and
 \eqref{S-lucky} is at least $2^{11}$.
From each of these probabilities, we have to subtract the (small)
probability that both $S$ and $D$ have a lucky coin, but this tilts the ratio
between ``Different'' and ``Same'' even more in $D$'s favor.
 Formally:
\begin{displaymath}
  \frac{\Pr[\textup{``Different''}]}{ \Pr[\textup{``Same''}]}
=
  \frac{  \Pr[\text{$D$ lucky}] - \Pr[\text{$D$ lucky and $S$ lucky}]}
  {  \Pr[\text{$S$ lucky}]  - \Pr[\text{$D$ lucky and $S$ lucky}]}
  > 
  \frac{  \Pr[\text{$D$ lucky}]} {  \Pr[\text{$S$ lucky}]}
  \ge 2^{11}\qedhere
\end{displaymath}

\end{proof}

Since the algorithm only needs to count up to 11 and to maintain a few
flags,
it is clear that it
can be
carried out by a PFA.\footnote
{%
As an exercise, the reader may try to work out the required number of states.
The 
outcomes should be represented by a partition of
the states into four 
classes,
including a category ``Rejected'' for inputs that don't adhere to
the format $\texttt{a}^i\texttt{b}^j\texttt{\#}$.
A literal and naive implementation that simply keeps track of every lucky
and unlucky coin and sets a flag when a \texttt{b} is seen
(this is the only thing that needs to be remembered in order to check
the syntax, except for a final state change on reading~\texttt{\#})
would need $2^5\times 12 +4 =
388$ states.
%
%
By excluding impossible combinations of flags and
with some other tricks
like merging states whose distinction is irrelevant
(see Section~\ref{sec:2less}),
I managed to do it with 173 states.
If the PFA can trust that the input has the correct format, 108 states suffice.
}






\subsection{Correctness Test: checking a 2CM computation}
\label{sec:check-computation}

Recall that we wish to check a description of a computation of the
 form
\begin{equation}
  \nonumber
  A =
  \texttt{0}^{l_0}\texttt{1}^{r_0}
  q_0\texttt{0}^{l_1}\texttt{1}^{r_1}{q_1}\texttt{0}^{l_2}\texttt{1}^{r_2}{q_2}
  \dots\texttt{0}^{l_{n}}\texttt{1}^{r_{n}}q_n  \texttt{\#}\ .
\end{equation}
The Equality Checker can be adapted to look at, say, two consecutive
zero blocks
$\texttt{0}^{l_i}$ and
$\texttt{0}^{l_{i+1}}$ of a computation that represent the values of the counter $l$ and check
whether $l_{i+1}={l_{i}}$.
It can also be adapted to check
 $l_{i+1}=l_{i}+1$, or 
 $l_{i+1}=l_{i}-1$, as appropriate for the state
 $q_i$ and the results of the zero test of $l_i$ and $r_i$. The guarantees of Lemma~\ref{EC-outcome} about the
 outcome remain valid.

We run independent Equality Checkers for each relation
between two consecutive values 
 ${l_i}$ and ${l_{i+1}}$, as well as
 ${r_i}$ and ${r_{i+1}}$, of a computation $A$. In total, these are
 $2n$ Equality Checkers.
 In the  schematic drawing of
 Figure~\ref{fig:checker},
the outcomes
of the 
Equality Checkers
are shown as a row of boxes.
Typically, most of them will be
``Undecided'', with a few interspersed ``Same'' and ``Different'' results
(proportionally much fewer than shown in the first example row). We are interested
in the rare cases when all outcomes are
 ``Same'', or all ``Different''.

 \begin{figure}[htb]
   \centering
   \includegraphics[scale=0.95]{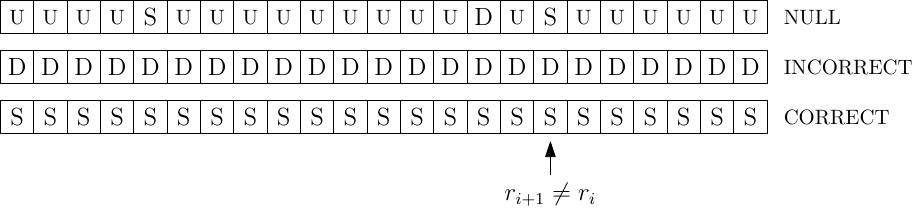}
   \caption{The Correctness Test for a computation $A$, and a
    hypothetical position where equality does not hold.}
   \label{fig:checker}
 \end{figure}
 
 The output of these Equality Checkers 
 is aggregated into a \emph{Correctness Test} as follows:
We report the output ``CORRECT'' 
if \emph{all} Equality Checkers report ``Same'', and
we report the output ``INCORRECT'' 
if \emph{all} Equality Checkers report ``Different''.
Otherwise, we report ``NULL''.

To compute this result,
only four independent Equality Checkers have to run simultaneously:
While reading the input,
the current block lengths $l_i$ and $r_i$ have to be compared with the preceding and
the next values.
Thus, the computation can be implemented by a PFA, with finitely many states.
(Looking more carefully, one sees that actually, only three Equality
Checkers are active at the same time: For example, when reading
$\texttt{1}^{r_i}$, the Equality Check between $\texttt{0}^{l_{i-1}}$
and $\texttt{0}^{l_{i}}$ has
already been completed.)

\begin{lemma}
  Suppose that a computation $A$ of the form \eqref{eq:computation} passes all formal checks.

  If $A$ represents an accepting computation,
    $$\Pr[\textup{``INCORRECT''}] = \Pr[\textup{``CORRECT''}].$$

  If $A$ does not represent an accepting computation,
    $$\Pr[\textup{``INCORRECT''}] \ge 2^{11}\cdot
    \Pr[\textup{``CORRECT''}].$$

  \end{lemma}
  \begin{proof}
    The probability for ``CORRECT'' is the product of the
    probabilities that each Equality Checker results in ``Same'', and
    analogously, for ``INCORRECT'' and ``Different''.

    If $A$ represents an accepting computation, then all
    Equality Checkers are balanced between ``Same'' and ``Different'',
    and the result is clear.
    Otherwise, there is at least one position (marked by an arrow in
    Figure~\ref{fig:checker}) where an error occurs, and the
    probability for
 ``Different'' is at least $2^{11}$ times larger than for
 ``Same'',
according to Lemma~\ref{EC-outcome}.
 In all other Equality Checkers, the probability is
     either balanced or it gives a further advantage for
     ``Different''.
     Thus, the product of the probabilities
     is at least $2^{11}$ times larger
for ``all Different''
     than for ``all Same''.
  \end{proof}

\subsection{Third-level aggregation: processing the whole input}
\label{sec:whole}

An Equality Checker aggregates the results of many coin flips into
an output ``Same'', ``Different'', or ``Undecided''.
We have further aggregated the result of many Equality Checkers
into a
Correctness Test for the word $A$ (with
output ``CORRECT'', ``INCORRECT'', or ``NULL'').
We add yet another level of aggregation in order to
decide whether the PFA should accept
the input word.
As mentioned, we feed the PFA with 
a huge number of copies of an accepting
computation~$A$.
Each copy of $A$ is subjected to the Correctness Test.

If we take the first definite result
(``CORRECT'' or ``INCORRECT'') as an indication whether to accept or
reject the input, we get an acceptance probability close to $1/2$ on
a valid input. (It is a little less than $1/2$ because of the 
chance that the input runs out before a definite answer is obtained.)
On the other hand,
if no accepting computation exists,
any input must consist of
``fake'' computations.
The algorithm
will recognize this and reject 
with probability at least $1-1/2^{11}$.

\subsubsection{Increasing the acceptance probability}
\label{sec:boost}
We modify the rules to make the acceptance probability larger, at
the expense of the rejection probability for fake inputs.  We
determine the overall result 
as follows.  As soon as a Correctness Test yields ``CORRECT'', we
accept the input.  However,
in order to reject the input, we wait until we have received 10 answers ``INCORRECT''
before receiving an answer ``CORRECT''.  If the end of
the input is reached before any of these events happens, this also
leads to rejection.
Of course, we also reject the input right away 
if any of the formal
checks fails.

\begin{theorem}\label{dichotomy}
  If there is an accepting computation $A$ for the 2CM, then the PFA accepts the
  input $A^t$, for sufficiently large $t$, with probability more than $0.99$.

  If there is no accepting computation, then the PFA rejects every
  input with probability at least $0.99$.
\end{theorem}
\begin{proof}
  If $A$ is an accepting computation, the distribution between
  ``CORRECT'' and ``INCORRECT'' is fair.
  Thus, the probability of receiving 10 
  outputs ``INCORRECT'' before receiving an output ``CORRECT'' is $1/2^{10}<0.001$.
  To this we must add the probability of rejection because the input runs out before receiving
  an
   output ``CORRECT'', but this can be made arbitrarily small by
   increasing~$t$.

  If there is no accepting computation, then
    ``INCORRECT'' has an advantage over
    ``CORRECT'' by a factor at least $2^{11}$.
    If the input runs out before a decision is reached, this is
    in the favor of rejection. 
Otherwise, the probability of receiving 10 
  outputs ``INCORRECT'' before receiving an output ``CORRECT'' is at least  
  \begin{displaymath}
    \left(\frac {2^{11}}{2^{11}+1}\right)^{10}=
  \left(1-\frac 1{2^{11}+1}\right)^{10}
  \ge (1-\tfrac 1{2000})^{10}
  \approx 1-\tfrac 1{200} = 0.995.\qedhere
  \end{displaymath}
\end{proof}

If the 2CM halts, there is an accepting computation $A$.
($A$ is unique since the
2CM is deterministic.)
In this situation, the language
recognized by the PFA with cutpoint $\lambda=\frac12$ contains the set $\{\,A^t\mid
t\ge t_0\,\}$ for some large $t_0$. Otherwise, the language is empty.

As a consequence, checking whether the language accepted by a PFA is
empty is undecidable.

\subsubsection{Who is afraid of small probabilities?}
As an exercise, we estimate the necessary number $t$ of
repetitions of~$A$.
Suppose that the accepting computation $A$ has $m$ transitions. Then the
counter values $l_i$ and $r_i$ are also bounded by $m$.
The probability of the outcome ``Same'' in the Equality Checker is roughly
$2^{-m}$,
and the probability that all $2m$
Equality Checkers for the computation $A$ yield ``Same'',
leading to the answer ``CORRECT'',
is roughly
$(2^{-m})^{2m}=4^{-m^2}$.

 We want the probability that none of $t$ experiments
 gets the answer ``CORRECT'' to be $\le 0.009$
 (the difference between the bound $0.001>1/2^{10}$ established in the
 proof of
Theorem~\ref{dichotomy}
 and the target tolerance $0.01$):
 \begin{displaymath}
   (1-4^{-m^2})^t\le 0.009
 \end{displaymath}
 Since
$
1-4^{-m^2} 
\approx
\exp(-4^{-m^2})$,
%
we need $t$ to be roughly 
$5\cdot4^{m^2}
$.

This dependence on the runtime $m$ of the 2-counter machine does not
appear 
so terrible; however,
when considering the overhead of simulating a
Turing machine (see footnote~\ref{simulate-TM-2CM}),
the dependence blows up to 
a triply-exponential growth
in terms of the runtime of a
Turing machine.

 \subsubsection{Boosting the decision probabilities}
\label{sec:further-boost}

We can boost the decision probabilities beyond 0.99 to become
arbitrarily close to 1. We simply run an odd number of copies of the PFA
simultaneously and take a majority vote.

Alternatively, we can adjust the parameters.
The number $K$ of times that we wait for ``INCORRECT''
before rejecting the input
can be increased above $K=10$.
As a compensation,
we have to increase the modulus $G$ (we have chosen~$G=12$) by which $i$ and $j$
are compared in the Equality Checker.
The acceptance probability in case of a valid input increases to
become arbitrarily close to $1-1/2^K$, and the rejection probability
for an invalid input is at least $(1-1/2^{G-1})^K$.

In summary,
for any $\eps>0$ we can construct the PFA in such a
way that it either accepts \emph{some word} with probability at least
$1-\eps$, or there is \emph{no word} that it accepts with probability larger than
$\eps$.
This does not mean that there cannot be words whose acceptance
probability is between those ranges, for example close to $1/2$.
 Candidates for such words are
the
words $A^t$ where $t$ is slightly too small.\footnote
{%
\label{isolatedcutpoint}%
  In fact, it is impossible to avoid the neighborhood of $1/2$ except
  for very simple languages:
  Rabin~\cite[Theorem~3]{rabin1963} showed in 1963 that
 an open gap interval $(p_1,p_2)$
of positive length, such that the acceptance probability
 never falls in this gap,
 can only exist if, for a cutpoint $\lambda$ in this interval,
the recognized language is regular,
see also
\cite[Theorem~2.3 in Section~IIIB, p.~160]{paz71}
or  \cite[\S 3.2.2, 
pp.~112--115]{claus71}.
Such a cutpoint $\lambda$ is called an
\emph{isolated cutpoint}.
}

\subsection{Summing up the proof of Theorem~\ref{thm-condon-lipton}}
\label{sec:formalxx}

We have described the algorithm for the PFA verbally as
a probabilistic algorithm, keeping in mind the 
finiteness constraints of a finite automaton. Eventually,
this algorithm must be translated into a set of states and transition
matrices.
Theorem~\ref{thm-condon-lipton} puts some extra 
constraints on the
PFAs whose emptiness is undecidable.

\CondonLiptonTwoMatrices*

\begin{proof} 
The extra constraints can be easily fulfilled:

(a) We encode the input $A$
with a fixed-length
binary code for the original input alphabet
$Q\cup \{\texttt{0},\texttt{1},\texttt{\#}
\}$. This means that
the set $\mathcal{M}$ can be restricted to only two matrices.
(Lemma~\ref{coding} in
Section~\ref{sec:coding} below
treats this transformation more formally.)

(b)
By padding the input, we can ensure that
the PFA algorithm needs to toss at most one coin per input symbol, and
thus
the entries of the matrices can be restricted to
$0,\frac12,1$.\footnote
{PFAs with this restricted set of probabilities are called
  \emph{simple PFAs} in \cite[Definition~2]{gimbert-oualhadj-2010:PFA}.}
In the algorithm as described, only 16 coin tosses are necessary per
input character
(four coins per Equality Checker running at any point in time).
Thus we simply
pad each codeword in the binary code 
with 15 zeros.

(c) Our algorithm does not need to make any coin flips before
reading the first symbol.
Thus, we can fix the start state to be a deterministic state.

(d) Finally,
a single accepting state is enough:
As soon as the algorithm has decided to accept the input, it will stay
committed to this decision.
The accepting state is an absorbing state, and
 there is another
 absorbing state for rejection.
In terms of vectors,
both the starting distribution~$\pi$
and the characteristic vector~$\eeta $ of accepting states are standard unit vectors.
(Since the empty input is not accepted, the accepting state is
distinct from the start state,
and
we can arrange the states so that the acceptance probability is found
in
the upper right corner of the product
$M_1M_2\ldots M_m$.)
\end{proof}
The dimension of the matrices $M_i$ will be huge. It is to the largest extent determined by
the number of states of the 2CM, and this is not under control.

\subsection{History of ideas}
\label{sec:ideas}

Condon and Lipton credit
the main ideas of their proof
to
R{\={u}}si{\c{n}}{\v{s}}
Freivalds~\cite{freivalds81}, who studied the Emptiness Problem for probabilistic \emph{2-way} finite
automata in 1981 (unaware of Nasu and Honda's earlier work).
In particular,
Freivalds developed the idea of
a competition between two players to recognize the language
$\{\,\texttt{a}^i\texttt{b}^i\mid i\ge0\,\}$
(Section~\ref{sec:eq}),
and
 aggregating the results of these competitions into ``macrocompetitions''
 (Section~\ref{sec:check-computation}).
A 2-way automaton can move the input head back and forth over the
input, and thus process the input as often as it wants.
 Freivalds claimed that the Emptiness Problem for
 such
 automata
is undecidable~\cite[Theorem~4]{freivalds81};
he gives only a hint that the reduction should be from the PCP
(Post's Correspondence Problem,
see Section~\ref{sec:PCP}),
without any 
details how to connect ``macrocompetitions''
with the PCP.
I have not been able to come up with an idea how the proof would
proceed.

For our present case of a (1-way)
finite automaton,
the repeated scan of the input is not possible; it
is replaced by providing an input which consists of many
repetitions of the same word.

\section{The Nasu--Honda--Claus proof via Post's Correspondence Problem
}
\label{sec:nasu-honda-claus}

This section presents the proof of
 Nasu and Honda~\cite{NasuHonda1969} from 1969 in the version of Claus
 \cite{claus71} from 1971,
leading to the undecidability results in
Propositions~\ref{ge1/2}--\ref{gt1/4},
which are then strengthened 
to
Theorems~\ref{thm-fixed-f}--
\ref{thm-fixed-pi}
in the rest of the paper.

\subsection{The binary PFA}
\label{sec:binary}

For a string $u\in \{\texttt 0,\texttt 1\}^*$, we denote by $(u)_2$
the
numeric 
value
of $u$ when it is interpreted as a binary number, and we write $|u|$
for the length of $u$. We define the 
stochastic matrix
\begin{displaymath}
  B(u) :=
  \begin{pmatrix}
    1-\frac{(u)_2}{2^{|u|}} &
    \frac{(u)_2}{2^{|u|}} \\[5pt]
    1-\frac{(u)_2+1}{2^{|u|}} &
    \frac{(u)_2+1}{2^{|u|}}
  \end{pmatrix}
  \text{, for example }
  B(\texttt{00110}) =
  \begin{pmatrix}
    \frac{26}{32} &
    \frac{6}{32} \\[5pt]
    \frac{25}{32} &
    \frac{7}{32}
  \end{pmatrix}.
\end{displaymath}
These matrices 
fulfill the remarkable 
 multiplication law
\begin{equation}\label{multiply}
  B(u)B(u') = B(u'u),
\end{equation}
which can be confirmed by
 a straightforward calculation.
Note the reversed order of the factors.

Note that the top right entry
$\frac{(u)_2}{2^{|u|}}$
of the matrix  $B(u)$ is the value $0.u$
when interpreted as a binary
fraction;
in our example, where $u=\texttt{00110}$, we have $0.u =
(0.00110)_2 = \frac 6{32} =
(0.0011)_2$.
We will continue to use the convenient notation
$0.u$ for this.
As we see, trailing zeros don't influence the value
of $0.u$, and we have to careful about this.

\subsection{Post's Correspondence Problem (PCP)}
\label{sec:PCP}

In the \emph{Post Correspondence Problem} (PCP),
we are given a list of pairs of words $(v_1,w_1),
\allowbreak
(v_2,w_2), \ldots,
(v_k,w_k)$. 
 The problem is to decide if there is a nonempty sequence
$a_1a_2\ldots a_m$
of indices $a_i \in \{1,2,\ldots,k\}$ such that
\begin{displaymath}
  v_{a_1}v_{a_2}\ldots v_{a_m}
=  w_{a_1}w_{a_2}\ldots w_{a_m}
\end{displaymath}
This is one of the well-known undecidable problems.\footnote
{A reduction from the Halting Problem for Turing
  Machines to a closely related problem,
the \emph{Modified}
Post Correspondence Problem (see 
Section~\ref{sec:RMPCP}) is described in detail
in
Sections~\ref{sec:MPCP}--\ref{sec:wordpairs}.}
It is no restriction to fix
the alphabet to $\{\texttt{0},\texttt{1}\}$,
since every alphabet can be encoded in binary.

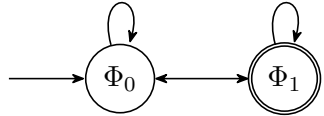
\begin{figure}[htb]
  \centering
  \begin{tikzpicture}[
    statenode/.style={draw,circle, minimum size=3mm},
    finalnode/.style={draw,double,circle, minimum size=5mm},
    style={->,shorten >=0.3pt,>={Stealth[round]},semithick},
    ]

      \node [statenode] (p00) {$\Phi_0$};
    \node [finalnode,right=1.7] at (p00) (p01) {$\Phi_1$};
\node [left=10mm] at (p00.west) (start) {}; 

\draw[->]  (start) -- (p00);

\draw[<->]  (p00) -- (p01);    

\draw[->] (p01) edge[loop above] (p01);
\draw[->] (p00) edge[loop above] (p00);

\end{tikzpicture}

\caption{The binary automaton with acceptance probability $\phi$}
  \label{fig:2PFA}
\end{figure}

Let us look at 
the first sequence of words $v_1,\ldots, v_k$. We
construct a PFA with input alphabet 
$\{1,2,\ldots,k\}$ and two states $\Phi_0$ and $\Phi_1$,
see Figure~\ref{fig:2PFA}.
The transition matrices
are $M_i = B(v_i)$. We take $\Phi_0$ as the start state
and $\Phi_1$ 
as the accepting state. Then
the acceptance
probability
 of the word $a=a_1a_2\ldots a_m$
 is found in the upper right corner of the product
$M_{a_1}M_{a_{2}}\ldots M_{a_{m-1}}M_{a_{m}}$
 of
the corresponding transition matrices, and 
it follows from \eqref{multiply} that
this 
is
\begin{equation}\label{prob-v}
  \phi(a)=0.v_{a_m}v_{a_{m-1}}\ldots v_{a_2}v_{a_1}.
\end{equation}

We can build an analogous PFA for the other sequence of words
$w_1,\ldots, w_k$, and then the acceptance probability of $a$ will be
\begin{equation}\label{prob-w}
  \psi(a)=0.w_{a_m}w_{a_{m-1}}\ldots w_{a_2}w_{a_1}.
\end{equation}
Due to the swapping of the factors in the multiplication law \eqref{multiply},
 the words are concatenated
in \eqref{prob-v} and
\eqref{prob-w}
in reverse order,
but this cosmetic change does
not affect the undecidability 
of the PCP.
Thus the PCP comes down to the question whether there is a nonempty word
$a$ with
equal acceptance probabilities
$\phi(a)=\psi(a)$ in the two PFAs.

We have to be careful because of the
\emph{trailing zeros issue}:
Trailing zeros don't change the
probabilities 
\eqref{prob-v} and
\eqref{prob-w}.
An easy way to circumvent
this problem
is to add a
\texttt{1} after every symbol of every word, thus doubling the length of
the words.
This ensures that there are no trailing zeros that could go unnoticed.

\subsection{Testing equality of probabilities}
\label{check:equality}
For recognizing the words
 $a$
 with $\phi(a)=\psi(a)$,
there is a construction of a PFA that does the job. It is 
based on the identity 
\begin{equation}\label{eq:equality-trick}
  \tfrac12 \phi\psi
+  \tfrac14(1-\phi^2)
+  \tfrac14(1-\psi^2)
= \tfrac 12 - \tfrac14 (\phi-\psi)^2.
\end{equation}
We will build a PFA for each term
$ \phi\psi$, $1-\phi^2$, $1-\psi^2$
on the left,
and we will
mix them in the right proportion.
As the right-hand side shows, we have then (almost) achieved our goal:
The acceptance probability
achieves its
maximum value $\frac12$ only for $\phi(a)=\psi(a)$.\footnote
{Section~\ref{sec:bertoni} describes an alternative way to achieve the
  same effect.}

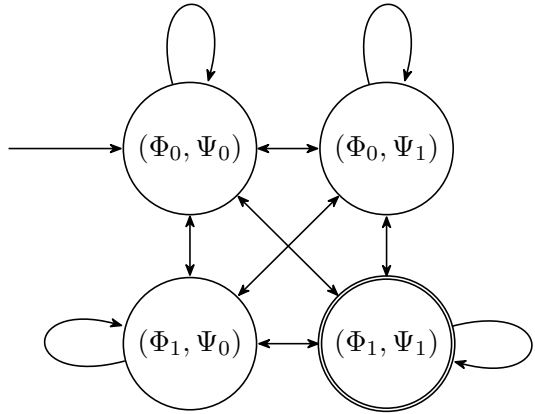
\begin{figure}[htb]
  \centering
  \begin{tikzpicture}[  
    statenode/.style={draw,circle, minimum size=3mm},
    finalnode/.style={draw,double,circle, minimum size=5mm},
    style={->,shorten >=0.3pt,>={Stealth[round]},semithick},
    ]

      \node [statenode] (p00) {$(\Phi_0,\Psi_0)$};
    \node [statenode,right=1.7] at (p00) (p01) {$(\Phi_0,\Psi_1)$};
    \node [statenode,below=1.7] at (p00) (p10) {$(\Phi_1,\Psi_0)$};
    \node [finalnode,below=1.7] at (p01) (p11) {$(\Phi_1,\Psi_1)$};
\node [left=15mm] at (p00.west) (start) {}; 


\draw[->]  (start) -- (p00);    

\draw[<->]  (p00) -- (p01);    
\draw[<->]  (p00) -- (p11);    
\draw[<->]  (p00) -- (p10);    
\draw[<->]  (p01) -- (p10);    
\draw[<->]  (p11) -- (p10);    
\draw[<->]  (p11) -- (p01);    

\draw[->] (p11) edge[loop right] (p11);
\draw[->] (p10) edge[loop left] (p10);
\draw[->] (p01) edge[loop above] (p01);
\draw[->] (p00) edge[loop above] (p00);




\end{tikzpicture} 
  
\caption{
  Acceptance probability $\phi\psi$
}
  \label{fig:22PFA}
\end{figure}

It is straightforward to
build a PFA whose acceptance probability is the product
$\phi(a)\psi(a)$,
see Figure~\ref{fig:22PFA}: This PFA simulates the two PFAs for $v_1,\ldots,v_k$
and for $w_1,\ldots,w_k$ simultaneously and accepts if both 
PFAs accept. The resulting \emph{product PFA} has four states
$\{\Phi_0,\Phi_1\}\times
\{\Psi_0,\Psi_1\}$.
Similarly, we can
build a PFA with acceptance probability
$\phi(a)^2$: We simulate two \emph{independent} copies of the PFA for
$v_1,\ldots,v_k$.  This leads again to four states.
%
To get acceptance probability $1-\phi(a)^2$, we complement the set of
accepting states.
The PFA for $1-\psi(a)^2$ follows the same principle. Finally, we mix
the three PFAs in the ratio $\frac12:\frac14:\frac14$, as shown in
Figure~\ref{fig:mix}a.

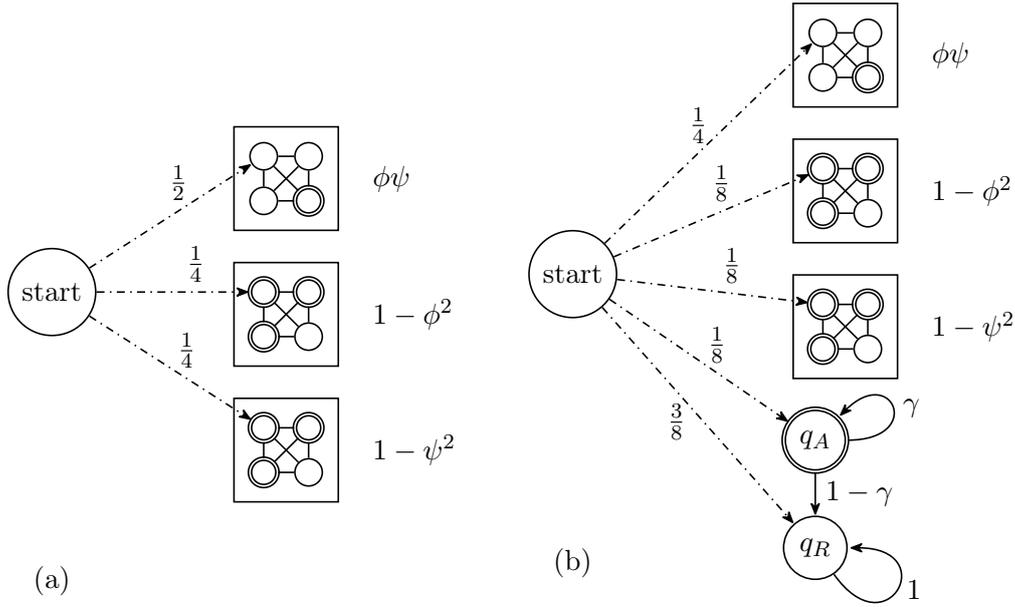
\begin{figure}[htb]
  \centering

\def\nodeboxes {    
    \node [finalnode] at (4,2.2) (phi00) {};
    \node [finalnode,right=0.4] at (phi00) (phi01) {};
    \node [finalnode,below=0.4cm
    ] at (phi00) (phi10) {};
    \node [statenode,below=0.4cm] at (phi01) (phi11) {};

\draw [tips=false] (phi00) -- (phi01) -- (phi11) -- (phi10) -- (phi00) -- (phi11);
\draw [tips=false] (phi10) -- (phi01);
    
\node[inner sep=2mm,draw,rectangle,fit=(phi00) (phi01) (phi10) (phi11)]
(PHI) {};
\node [right=10mm] at (PHI) {$1-\phi^2$};

    \node [finalnode] at (4,0.4) (psi00) {};
    \node [finalnode,right=0.4] at (psi00) (psi01) {};
    \node [finalnode,below=0.4cm
    ] at (psi00) (psi10) {};
    \node [statenode,below=0.4cm] at (psi01) (psi11) {};
\draw [tips=false] (psi00) -- (psi01) -- (psi11) -- (psi10) -- (psi00) -- (psi11);
\draw [tips=false] (psi10) -- (psi01);
\node[inner sep=2mm,draw,rectangle,fit=(psi00) (psi01) (psi10) (psi11)]
(PSI) {};
\node [right=10mm] at (PSI) {$1-\psi^2$};

    \node [statenode] at (4,4) (phipsi00) {};
    \node [statenode,right=0.4] at (phipsi00) (phipsi01) {};
    \node [statenode,below=0.4cm] at (phipsi00) (phipsi10) {};
    \node [finalnode,below=0.4cm] at (phipsi01) (phipsi11) {};
\node[inner sep=2mm,draw,rectangle,fit=(phipsi00) (phipsi01) (phipsi10) (phipsi11)]
(PHIPSI) {};
\draw [tips=false] (phipsi00) -- (phipsi01) -- (phipsi11) -- (phipsi10) -- (phipsi00) -- (phipsi11);
\draw [tips=false] (phipsi10) -- (phipsi01);
\node [right=10mm] at (PHIPSI) {$\phi\psi$};
}

\noindent\null     \hfill
\begin{minipage}[t]{0.42\linewidth}
{\begin{tikzpicture}[
    statenode/.style={draw,circle, minimum size=3mm},
    finalnode/.style={draw,double,circle, minimum size=3mm},
    style={->,shorten >=0.3pt,>={Stealth[round]},semithick},
    ]
\nodeboxes
\node [left=22mm,circle,draw] at (phi00) (start) {start};

\node [below=35mm] at (start) {(a)};
\node [below=45mm] at (start) {};

\draw[->,dashdotted] (start) -- (phipsi00) node[above,pos=0.55]{$\frac 12$};
\draw[->,dashdotted] (start) -- (psi00) node[above,pos=0.6]{$\frac 14$};
\draw[->,dashdotted] (start) -- (phi00) node[above,pos=0.65]{$\frac 14$};
  \end{tikzpicture}}
     \end{minipage}
     \hfill
     \begin{minipage}[t]{0.45\linewidth}
       {  \begin{tikzpicture}[
    statenode/.style={draw,circle, minimum size=3mm},
    finalnode/.style={draw,double,circle, minimum size=3mm},
    style={->,shorten >=0.3pt,>={Stealth[round]},semithick},
    ]
\nodeboxes
\node [finalnode,minimum size=5mm] at (3.9,-1.4) (hold) {$q_A$};
\node [statenode,minimum size=5mm,below=10mm] at (hold) (reject) {$q_R$};

\node [left=30mm,circle,draw] at (PSI.north) (start) {start};

\node [below=35mm] at (start) {(b)};

\draw[->,dashdotted] (start) -- (phipsi00) node[above,pos=0.45]{$\frac 14$};
\draw[->,dashdotted] (start) -- (psi00) node[above,pos=0.6]{$\frac 18$};
\draw[->,dashdotted] (start) -- (phi00) node[above,pos=0.55]{$\frac 18$};
\draw[->,dashdotted] (start) -- (hold) node[above,pos=0.6]{$\frac 1{8}$};
\draw[->,dashdotted] (start) -- (reject) node[below,pos=0.4]{$\frac 3{8}\,$};

\draw[->] (hold) edge[in=45,out=0,loop] node[midway,right]{$\gamma$} (hold);
\draw[->] (hold) -- (reject) node[midway,right]{$1-\gamma$};
\draw[->] (reject) edge[in=0,out=-55,loop] node[midway,right]{$1$} (reject);

  \end{tikzpicture}}
\end{minipage}
\hfill\null
  \caption{
    (a) Acceptance probability $\tfrac 12 - \tfrac14 (\phi-\psi)^2$.
    (b) $\tfrac 14 - \tfrac18 (\phi-\psi)^2+\eps$}
  \label{fig:mix}
\end{figure}

The dash-dotted arrows
from the start state to 
three ``local start states''
inside the square boxes
denote random transitions
that should be thought of as
happening before the algorithm reads its first input symbol.
In the PFA, such a transition is actually carried out
in combination 
with the subsequent
transition for the 
input symbol inside one of the square boxes,
as part of the transition out of the start state
when reading the first input symbol.

More formally, 
if $M_i$ is a $12\times 12$
transition matrix for the 12 states in the square boxes and
$\pi_0^T=(\frac12,0,0,0, \frac14,0,0,0, \frac14,0,0,0)$ is the
starting distribution, the
$13\times 13$ matrix that includes the start state as the first state is
\begin{displaymath}
  \begin{pmatrix}
    0 & \pi_0^T M_i\\
    0 & M_i
  \end{pmatrix}.
\end{displaymath}

The introduction of the new start
state has the beneficial side effect of eliminating the empty word
$\epsilon$ from the recognized language.
The empty word would otherwise 
satisfy the equation $\phi(a)=\psi(a)$, because
$\phi(\epsilon)=\psi(\epsilon)=0$.

In total, we have now 13 states, 7 of which are accepting.
As an intermediate undecidability result, we can thus state:
\begin{proposition}
  \label{ge1/2}
  The following problem is undecidable:
  
Given  a finite set $\mathcal{M}$ of stochastic
 matrices
of size $13\times 13$
 with binary fractions as 
   entries, 
  is there a product
  $M_1M_2\ldots M_m$, with $M_i\in\mathcal{M}$ for all
  $i=1,\ldots,m$, 
such that the sum of the $7$ rightmost entries in the top row is
$ \ge \tfrac 12$?
\qed
\end{proposition}

\subsection{Achieving strict inequality}
\label{sec:strict}

Proposition~\ref{ge1/2} almost describes a PFA,
except that the convention for
a PFA 
to recognize a word is strict inequality ($>\lambda$). We thus have to
raise the probability just a tiny 
bit, without raising any of
the values $<\lambda 
$ to become bigger than~$\lambda$. 

Since all probabilities are rational, this can be done as follows,
see Figure~\ref{fig:mix}b.
In our case,
all 
transition probabilities 
within the square boxes are multiples of
 some small unit 
 \begin{displaymath}
 \gamma := 4^{-\max\{|v_i|,|w_i|:
   1\le i \le k\}}.  
 \end{displaymath}


%
 The original PFA is entered with probability 1/2.
The transition probabilities from the 
start state into the original PFA are now multiples of~$\gamma/8$.
(Remember that such a transition consists of a transition from the start
state along a dash-dotted arrow 
combined with a transition
inside
a square box.)
We create a new accepting state $q_A$ that is chosen initially with
probability $1/{8}$. Whenever a symbol is read, the PFA stays in that
state with probability $\gamma$, and otherwise it moves to some
absorbing state~$q_R$.
With the remaining probability $3/{8}$, we go
to~$q_R$ directly.

The new part contributes $\eps := \frac 1{8} \gamma ^{|a|}$ to the acceptance
probability of every nonempty word~$a$.
From the old part we have
$\tfrac 14 - \tfrac18 (\phi-\psi)^2$,
and we know that this probability is a multiple of $\frac 18
\gamma^{|a|}=\eps$.
Thus, if this probability is less than $1/4$, it cannot become greater
than $1/4$ by adding~$\eps$.
If it was equal to $1/4$ (i.e., if $a$ is a solution to the PCP), it becomes
greater
than $1/4$.
\begin{proposition}\label{=1/4}
  It is undecidable whether the language recognized by a PFA with 15
  states with cutpoint $\lambda=1/4$ is empty.
  \qed
\end{proposition}
This PFA has a fixed start state.

The cutpoint can be changed to any positive rational value
less than 
1/2 
by adjusting the initial split probability between the original PFA
of Figure~\ref{fig:mix}a and the states $q_A$ and $q_R$.
 cutpoints between 1/2 and 1 can be achieved at the expense of adding
 another accepting state.

According to Neary~\cite{neary:PCP5:2015},
 the PCP is already undecidable with as few as five
 word pairs.
Therefore, the size of the input alphabet in
Proposition~\ref{=1/4}, or the number of matrices $\mathcal{M}$ in
Proposition~\ref{ge1/2} can be restricted to~5.
 
\subsection{History of ideas}

The binary automaton
(Section~\ref{sec:binary})
and its generalization to
other bases than 2 appears already in
Rabin's 1963 paper~\cite{rabin1963}, and it is credited to
 E.~F.~Moore.
The basic $m$-ary automaton processes
 single digits from $\{0,\ldots,m-1\}$.
%
The binary automaton matrix in
Section~\ref{sec:binary}
for variable-length input words $u$ is the product
of several such single-digit matrices.
Instead of binary automata,
Nasu and Honda~\cite{NasuHonda1969} use
ternary (triadic) automata with digits 
$\{\mathtt{0},\texttt{1},\texttt{2}\}$, of which only
$\{\texttt{1},\texttt{2}\}$ are used
in order to avoid the trailing zeros
issue.

The equality test for probabilities
 (constructing a PFA to accept words
$a$ with $\phi(a)=\psi(a)$
from two PFAs with acceptance probabilities
$\phi(a)$ and $\psi(a)$, 
Section~\ref{check:equality}),
 including the method of
  adding a small probability to change $\ge\lambda$
 into $>\lambda$
(Section~\ref{sec:strict})
is given in Nasu and Honda \cite[Lemma 11, pp. 259--260]{NasuHonda1969}.
 The authors 
credit
H. Matuura, 
Y. Inagaki, and T. Hukumura
for the key ideas (a technical report and a conference record, both
from 1968 and in Japanese)
\cite[p.~261]{NasuHonda1969}.

Claus already %
observed
\cite[p.~158, remark after the proof of Satz~28]{claus71}
that the construction leads to a bounded number of
states. The details have been worked out above.

As I haven't been able to survey the
rich literature on probabilistic automata,
I may very well
have overlooked some earlier roots of these ideas.

Nasu and Honda~\cite{NasuHonda1969},
in a footnote
to 
Theorem~21, their main result about the
 undecidability of PFA Emptiness,
 write
 that ``it reduces to a statement in p.~150''
of a paper of
Marcel Schützenberger~\cite{schutzenberger63:certain} from 1963\footnote
{\url{https://monge.univ-mlv.fr/~berstel/Mps/Travaux/A/A/1963-4ElementaryFamAutomataSympThAut.pdf}}
\cite[footnote~6 on p.~270,
 referring to the remark before Lemma 12, p.~261]{NasuHonda1969}.
In that paper, Schützenberger derives some undecidability results,
using, among others, the PCP, but
I am not 
able to see the connection.

Nasu and Honda prove the undecidability of two more
questions
in connection with the language recognized by a PFA: whether the language is regular,
or whether the language is context-free
\cite[Theorem 22, p.~270]{NasuHonda1969},
see Appendix~\ref{sec:context-free}.

\subsubsection{Gimbert and Oualhadj 2010, following Bertoni 1975}
\label{sec:bertoni}

In 2010, Gimbert and Oualhadj \cite{gimbert-oualhadj-2010:PFA}, after
joining the lamentations over the impenetrability of Paz's
treatment \cite{paz71} and the long and technical proof of Madani,
Hanks, and Condon
\cite
{jair03} (see Section~\ref{sec:2c}), presented ``a new simple
undecidability proof of the Emptiness Problem''
as one of the results of their paper.  I review this proof
in greater detail in order to put it in the proper historical
perspective.
I will point where it differs from the proof
shown above, and I wish to correct an erroneous attribution of a
crucial auxiliary result to an early paper of Bertoni.

As a first step,
the proof starts with the translation of the PCP to the equality test
$\phi(a)=\psi(a)$,
as in Sections~\ref{sec:binary}--\ref{sec:PCP}.


The treatment of
the equality test in the second step is different from Nasu and Honda's
(Section~\ref{check:equality}):
The test $\phi(a)=\psi(a)$ is first reduced to the \emph{Equality Problem}
$\frac12\phi(a)+\frac12(1-\psi(a))=\frac12$,
using complementation and convex combination of PFAs.\footnote
{See also Paz \cite[Lemma 6.1 in Section~IIIB, p.~183]{paz71}.}
This 
Equality Problem (``Is the acceptance probability of a given PFA \emph{equal} to~$\lambda=\frac12$?\,'')
is reduced to the emptiness question with the
$\ge \lambda$ criterion by the relation
$x=\frac12\iff x(1-x)\ge\frac14$, and from there one can conclude that the Equality Problem
is undecidable.

Putting everything together, this amounts to the formula
\begin{displaymath}
\bigl(  \tfrac12\phi+\tfrac12(1-\psi)\bigr)
\bigl(  \tfrac12\psi+\tfrac12(1-\phi)\bigr)
 = \tfrac 14 - \tfrac14 (\phi-\psi)^2,
\end{displaymath}
giving almost the same result as
formula~\eqref{eq:equality-trick}.\footnote
{The discrepancy 
  between the constant 1/4 in this formula and the claimed
  cutpoint $\lambda=1/2$
\cite[Theorem 1]{gimbert-oualhadj-2010:PFA}
  is not addressed in
  \cite{gimbert-oualhadj-2010:PFA}.
}
I find this two-step approach conceptually simpler than the mixture of three automata in
Section~\ref{check:equality}.
If implemented in a straightforward way, it would lead to 16 states.

The third and final step adds a small probability to achieve strict inequality,
in essentially the same was as described in Section~\ref{sec:strict},
using three additional states.

Fijalkow, in a short 8-page note from 2017
\cite[p.~15]{Fijalkow-2017-SIGLOG}, has given a nice half-page
survey 
of
this proof.\footnote
{However, Fijalkow shortcuts the third step, the transition from acceptance with
$\ge\lambda$ to acceptance with $>\lambda$, by appealing to an invalid
argument:
Since
the existence of input strings $a$ with acceptance probability
$\phi(a)\ge 1/2$
is undecidable and
the existence of input strings with $\phi(a)= 1/2$ (the Equality Problem)
is undecidable,
the existence of input strings with $\phi(a)> 1/2$
must also be undecidable.}
As a complementary result,
he also showed an alternative to the binary automaton
that has acceptance probability $0.u$
\cite[Figure~1, p.~14]{Fijalkow-2017-SIGLOG} and that satisfies a
similar
 multiplication law as the binary automaton~\ref{multiply},
see Figure~\ref{fig:accept-fijalkow}.
I find it less mysterious than the binary automaton, but it requires 3 states.

\begin{figure}[htb]
  \centering
  \begin{tikzpicture}[
    roundnode/.style={circle, draw=black, minimum size=9mm},
    style={->,shorten >=0.3pt,>={Stealth[round]},semithick},
    ]
    
\node[roundnode] at (0,0) (q0)    {$q_0$}  ;
\node[roundnode,double] at (2.5,-1.9) (Ac)    {$\top$}  ;
\node[roundnode] at (2.5,1.9) (Re)    {$\bot$}  ;
\node[left=10mm] at (q0.west) (start) {}; 

\draw[->] (start) -- (q0);

\draw[->] (q0) -- node[below,pos=0.25,anchor=west]{\ \ $0.u$}  (Ac);
\draw[->] (q0) -- node[above,pos=0.25,anchor=west]{\ \ $1-0.u-1/2^{|u|}$}  (Re);

\draw[->] (Re) edge[loop right,looseness=14] node[midway,right]{$1$} (Re);
\draw[->] (Ac) edge[loop right,looseness=14] node[midway,right]{$1$} (Ac);
\draw[->] (q0) edge[loop above,looseness=10] node[midway,above]{$1/2^{|u|}$} (q0);

  \end{tikzpicture}
  \caption{A simpler automaton for acceptance with a given binary
    probability $0.u$ when reading a particular symbol}
  \label{fig:accept-fijalkow}
\end{figure}
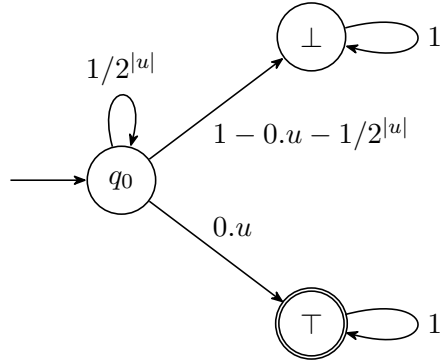

For the first one-and-a-half steps of the proof,
leading to the claim that the Equality Problem is
undecidable
\cite[Proposition~1]{gimbert-oualhadj-2010:PFA},  Gimbert and Oualhadj
credit a short paper of 
Alberto Bertoni from 1975~\cite{Bertoni1975}.

It is true that
Bertoni
\cite[p.~111, item a)]{Bertoni1975} \emph{did} build a PFA with acceptance probability
 $\frac12\phi(a)+\frac12(1-\psi(a))= \frac12 + (\phi(a)-\psi(a))/2$
 (and with 4 states)
 as an intermediate step.
However, a statement that the Equality Problem is
undecidable appears nowhere in
Bertoni~\cite{Bertoni1975}.
In fact, the PCP instances of 
Bertoni are
constructed in such a way that they
are guaranteed to have no (finite)
solutions,
and thus the Equality Problem would always have the answer ``no''.\footnote
{Bertoni's  
  goal was to prove that
 it is undecidable whether 1/2 is an isolated cutpoint
(see footnote~\ref{isolatedcutpoint}).
For this, he
constructed, for a given Turing machine $T$, a PCP with word pairs
$(v_i,w_i)$ with the following property: 
The longest common prefixes of the strings
$v_{a_1}v_{a_2}\ldots v_{a_m}$ and $w_{a_1}w_{a_2}\ldots w_{a_m}$ have
bounded length iff the Turing machine $T$ halts
\cite[End of Section~2, p.~110]{Bertoni1975}.
(The first condition is equivalent to 1/2 being an isolated
cutpoint
of $\frac12 + (\phi(a)-\psi(a))/2$
\cite[Theorem~3 of Section~3, p.~111]{Bertoni1975}.)

The classic construction of a PCP that has a solution iff $T$
halts
(Section~\ref{sec:MPCP})
would not work in this context. Clearly, if $T$ runs forever, there will be
arbitrarily long prefixes. However, if $T$ terminates and the 
PCP has a solution, one can create
arbitrarily long common prefixes
 simply by repeating the PCP solution
periodically
 and introducing a deviation at some point. Thus, $\lambda=1/2$ would
 never be isolated.

 As an additional puzzle,
the 2-state automaton that
Bertoni
  specifies \cite[Theorem~1 of Section~3, p.~110]{Bertoni1975}
 very roughly resembles a
  binary automaton. However, as written, it does not perform its claimed
  function.
  I assume he must have meant the binary automaton.
The $k$-ary automaton is written correctly in
a follow-up paper
\cite{BertoniMT77}
that deals with further decidability questions
about isolated cutpoints.
}

Bertoni refers to Rabin~\cite{rabin1963} and Paz \cite{paz71},
 but not to Nasu and Honda~\cite{NasuHonda1969}.

Gimbert and Oualhadj, in an appendix of the technical report
that is the full version of~\cite{gimbert-oualhadj-2010:PFA}%
\rlap,\footnote
{\url{https://hal.science/hal-00456538v3/file/gimbert_oualhadj_probabilistic_automata.pdf},
  p.~13}
present a proof of the undecidability of Equality Problem
that is supposedly due to Bertoni.
This proof is essentially the proof that I have sketched above, using
 binary automata and their mixture to generate the acceptance probability
$\frac12\phi(a)+\frac12(1-\psi(a))$,
except that
Gimbert and Oualhadj (and following them,
Fijalkow
\cite{Fijalkow-2017-SIGLOG}),
neglect to get rid of the empty solution of the PCP.\footnote
{Also, the third step
\cite[Proposition~2]{gimbert-oualhadj-2010:PFA}
of the proof
 depends having
a \emph{simple} PFA, whose transition probabilities are $0$,
$\frac12$, or 1.
 The claim that their construction for the second step yields
 such a PFA
 \cite[Problem~2 and Proposition~1]{gimbert-oualhadj-2010:PFA}
is not substantiated.}
Fijalkow \cite{Fijalkow-2017-SIGLOG}
adds to the 
mystifications by indirectly attributing
the
undecidability of the \emph{Emptiness Problem}
to
Bertoni~
\cite{Bertoni1975}.\footnote{``[Gimbert and Oualhadj 2010] gave a simple exposition of the
undecidability proof of Bertoni [Bertoni 1974] for the emptiness
problem.''
As for the date: December 1974 is the date of the conference whose proceedings
contain Bertoni's paper~\cite{Bertoni1975} and were published in 1975.
}

\subsection{Saving two states by merging indistinguishable states}
\label{sec:2less}

In the
PFA with acceptance probability
$\phi(a)^2$, where we simulate two {independent} copies of the same
PFA,
we can see
that the states
$(\Phi_0,\Psi_1)$ and
$(\Phi_1,\Psi_0)$ of
Figure~\ref{fig:22PFA}
 become indistinguishable when $\phi=\psi$. Thus, they can be merged
 into one state,
 denoted by $\{\Phi_0,\Phi_1\}$,
 and we reduce the number of states by one, see
 Figure~\ref{fig:22PFA-merged}a--b.

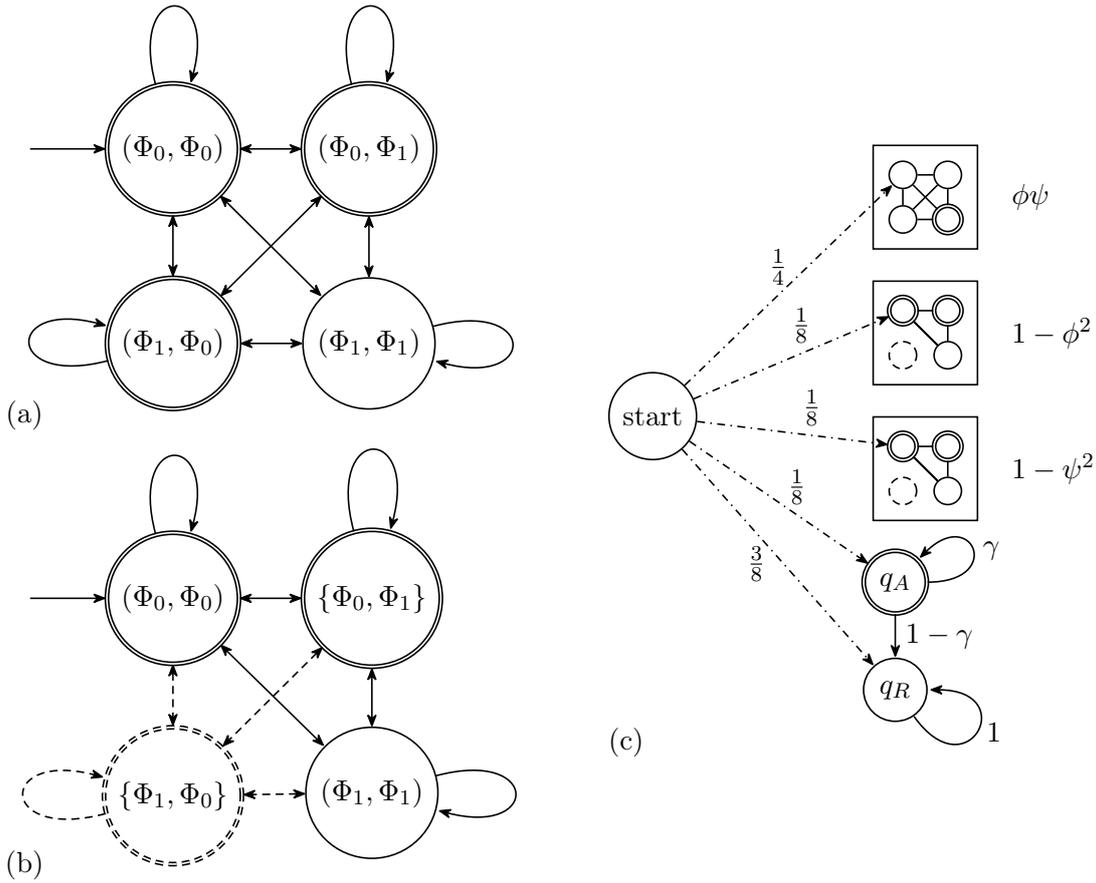
\begin{figure}[htb]
  \centering

  \begin{minipage}[t]{0.5\linewidth}
  {  \!\!\!\!
 \begin{tikzpicture}[  
    statenode/.style={draw,circle, minimum size=3mm},
    finalnode/.style={draw,double,circle, minimum size=5mm},
    style={->,shorten >=0.3pt,>={Stealth[round]},semithick},
    ]

      \node [finalnode] (p00) {$(\Phi_0,\Phi_0)$};
    \node [finalnode,right=1.7] at (p00) (p01) {$(\Phi_0,\Phi_1)$};
    \node [finalnode,below=1.7] at (p00) (p10) {$(\Phi_1,\Phi_0)$};
    \node [statenode,below=1.7] at (p01) (p11) {$(\Phi_1,\Phi_1)$};
\node [left=10mm] at (p00.west) (start) {}; 

      \node [finalnode, below=2.5] at (p10) (q00) {$(\Phi_0,\Phi_0)$};
    \node [finalnode,right=1.7] at (q00) (q01) {$\{\Phi_0,\Phi_1\}$};
    \node [finalnode,below=1.7,densely dashed] at (q00) (q10) {$\{\Phi_1,\Phi_0\}$};
    \node [statenode,below=1.7] at (q01) (q11) {$(\Phi_1,\Phi_1)$};
\node [left=10mm] at (q00.west) (startq) {}; 

\draw[->]  (start) -- (p00);    
\draw[->]  (startq) -- (q00);

\draw[<->]  (p00) -- (p01);    
\draw[<->]  (p00) -- (p11);    
\draw[<->]  (p00) -- (p10);    
\draw[<->]  (p01) -- (p10);    
\draw[<->]  (p11) -- (p10);    
\draw[<->]  (p11) -- (p01);    

\draw[->] (p11) edge[loop right] (p11);
\draw[->] (p10) edge[loop left] (p10);
\draw[->] (p01) edge[loop above] (p01);
\draw[->] (p00) edge[loop above] (p00);

\draw[<->]  (q00) -- (q01);    
\draw[<->]  (q00) -- (q11);    
\draw[<->,densely dashed]  (q00) -- (q10);    
\draw[<->,densely dashed]  (q01) -- (q10);    
\draw[<->,densely dashed]  (q11) -- (q10);    
\draw[<->]  (q11) -- (q01);    

\draw[->] (q11) edge[loop right] (q11);
\draw[->,densely dashed] (q10) edge[loop left] (q10);
\draw[->] (q01) edge[loop above] (q01);
\draw[->] (q00) edge[loop above] (q00);

\node [below=32mm] at (start) {\ (a)};
\node [below=32mm] at (startq) {\ (b)};

\end{tikzpicture} 
}
\end{minipage}
\hfill
\begin{minipage}[t]{0.46\linewidth}

\def\nodeboxes {    
    \node [finalnode] at (4,2.2) (phi00) {};
    \node [finalnode,right=0.4] at (phi00) (phi01) {};
    \node [statenode
    ,below=0.4cm,densely dashed] at (phi00) (phi10) {};
    \node [statenode,below=0.4cm] at (phi01) (phi11) {};

    \draw [tips=false] (phi00) -- (phi01) -- (phi11) 
    -- (phi00) -- (phi11);
    
\node[inner sep=2mm,draw,rectangle,fit=(phi00) (phi01) (phi10) (phi11)]
(PHI) {};
\node [right=10mm] at (PHI) {$1-\phi^2$};

    \node [finalnode] at (4,0.4) (psi00) {};
    \node [finalnode,right=0.4] at (psi00) (psi01) {};
    \node [statenode
    ,below=0.4cm, densely dashed] at (psi00) (psi10) {};
    \node [statenode,below=0.4cm] at (psi01) (psi11) {};
    \draw [tips=false] (psi00) -- (psi01) -- (psi11) -- 
    (psi00) -- (psi11);
\node[inner sep=2mm,draw,rectangle,fit=(psi00) (psi01) (psi10) (psi11)]
(PSI) {};
\node [right=10mm] at (PSI) {$1-\psi^2$};

    \node [statenode] at (4,4) (phipsi00) {};
    \node [statenode,right=0.4] at (phipsi00) (phipsi01) {};
    \node [statenode,below=0.4cm] at (phipsi00) (phipsi10) {};
    \node [finalnode,below=0.4cm] at (phipsi01) (phipsi11) {};
\node[inner sep=2mm,draw,rectangle,fit=(phipsi00) (phipsi01) (phipsi10) (phipsi11)]
(PHIPSI) {};
\draw [tips=false] (phipsi00) -- (phipsi01) -- (phipsi11) -- (phipsi10) -- (phipsi00) -- (phipsi11);
\draw [tips=false] (phipsi10) -- (phipsi01);
\node [right=10mm] at (PHIPSI) {$\phi\psi$};
}



       {  \begin{tikzpicture}[
    statenode/.style={draw,circle, minimum size=3mm},
    finalnode/.style={draw,double,circle, minimum size=3mm},
    style={->,shorten >=0.3pt,>={Stealth[round]},semithick},
    ]
\nodeboxes
\node [finalnode,minimum size=5mm] at (3.9,-1.4) (hold) {$q_A$};
\node [statenode,minimum size=5mm,below=10mm] at (hold) (reject) {$q_R$};

\node [left=30mm,circle,draw] at (PSI.north) (start) {start};

\node [below=40mm] at (start) {(c)\hskip 7mm \ };
\node [below=60mm] at (start) {\ }; 

\draw[->,dashdotted] (start) -- (phipsi00) node[above,pos=0.45]{$\frac 14$};
\draw[->,dashdotted] (start) -- (psi00) node[above,pos=0.6]{$\frac 18$};
\draw[->,dashdotted] (start) -- (phi00) node[above,pos=0.55]{$\frac 18$};
\draw[->,dashdotted] (start) -- (hold) node[above,pos=0.6]{$\frac 1{8}$};
\draw[->,dashdotted] (start) -- (reject) node[below,pos=0.4]{$\frac 3{8}\,$};

\draw[->] (hold) edge[in=45,out=0,loop] node[midway,right]{$\gamma$} (hold);
\draw[->] (hold) -- (reject) node[midway,right]{$1-\gamma$};
\draw[->] (reject) edge[in=0,out=-55,loop] node[midway,right]{$1$} (reject);

  \end{tikzpicture}}
  \end{minipage}

\caption{(a) Acceptance probability $1-\phi^2$
  with 4 states
(b)~with 3 states.
(c) Acceptance probability
    $\tfrac 14 - \tfrac18 (\phi-\psi)^2+\eps$
with 13 states
}
  \label{fig:22PFA-merged}
\end{figure}

 More precisely, if we denote
the transition probabilities of the original binary automaton by
 \begin{displaymath}
   B(u)=\raise 1.5ex\hbox{$
   \begin{array}{c@{\quad}cc}
     &\Phi_0&\Phi_1\\[1pt]
     \Phi_0&
\smash{\lower 1.7ex\llap{$\left(\vbox to 3ex{}\right.$}}             
             p_{00}&p_{01} 
\smash{\lower 1.7ex\rlap{$\left.\vbox to 3ex{}\right),$}}%
     \\
     \Phi_1&p_{10}&p_{11}
   \end{array}$}
 \end{displaymath}
 the 3-state PFA has the following transition matrix:
 \begin{equation}\label{eq:3x3}
   \let\phi=p
\raise 1.6ex\hbox{$   \begin{array}{c@{\ \quad}ccc}
& (\Phi_0,\Phi_0)
&       \{\Phi_0,\Phi_1\}
& (\Phi_1,\Phi_1)
\\[3pt]
(\Phi_0,\Phi_0) &         \phi_{00}^2& 2 \phi_{00} \phi_{01}& \phi_{01}^2\\
     \{\Phi_0,\Phi_1\}&
\llap{\smash{$\left(\vbox to 4.5ex{}\right.$}}
      \phi_{00} \phi_{10}& \phi_{01} \phi_{10}+\phi_{00} \phi_{11}&
     \phi_{01} \phi_{11}
\rlap{\smash{$\left.\vbox to 4.5ex{}\right)$}}
     \\
(\Phi_1,\Phi_1)& \phi_{10}^2& 2 \phi_{10} \phi_{11}& \phi_{11}^2
   \end{array}$}
\end{equation}
When the reduced automaton is in the state $\{\Phi_0,\Phi_1\}$, we can
think of the original 4-state automaton being in one of the states
$(\Phi_0,\Phi_1)$ or
$(\Phi_1,\Phi_0)$,
each with probability $1/2$.

For the PFAs in Propositions~\ref{ge1/2} and
~\ref{=1/4},
the number of states
can thus be reduced by 2,
as stated in the following proposition.
Figure~\ref{fig:22PFA-merged}c illustrates the automaton for
Proposition~\ref{ge1/2+}b.
We will show some explicit examples of 
transition matrices
for this automaton below, in Section~\ref{example}.

\goodbreak

\begin{proposition}
  \label{ge1/2+}\ 
  \begin{enumerate}[\rm (a)]
  \item 
  The following problem is undecidable:
  
Given  a finite set $\mathcal{M}$ of stochastic
 matrices
of size $11\times 11$
 with binary fractions as entries,
  is there a product
  $M_1M_2\ldots M_m$, with $M_i\in\mathcal{M}$ for all
  $i=1,\ldots,m$, 
such that the sum of the 5 rightmost entries in the top row is
$ \ge \tfrac 12$?
\item 
  It is undecidable whether the language recognized by a PFA with 13
  states with cutpoint $\lambda=1/4$ is empty.
  \qed
  \end{enumerate}
\end{proposition}

\subsection{Saving the start state by using the Modified
  Post Correspondence Problem
}
\label{sec:RMPCP}

We can eliminate the start state by using the \emph{Modified
  Post Correspondence Problem} (MPCP).
It differs from the PCP in one detail: The pair $(v_1,w_1)$ must be used
as the \emph{starting pair}, and it cannot be used in any other place.
In other words, the solution must satisfy the constraints $a_1=1$, and $a_i>1$ for $i=2,\ldots,m$.
The MPCP is often used as an
intermediate problem when reducing the Halting Problem for Turing
machines to the PCP,
and then it takes some extra effort to reduce the
MPCP to the PCP,
see for example
\cite[Lemma 8.5
]{hopcroft79} or
\cite[p.~189]{Sipser}.
In our situation,
the MPCP is actually the more convenient version of the
problem.

The idea is to apply the transition for the first letter
$a_1$ right away, and use the resulting distribution on the states as
the starting distribution~$\pi$.

There is still a small technical discrepancy:
In the formulas \eqref{prob-v} and
\eqref{prob-w} for the acceptance probability, the first letter of the
sequence $a$
determines the
\emph{last} pair of words to be concatenated. Thus we must reverse all words
$v_i$ and $w_i$ and turn the MPCP into a
\emph{reversed} MPCP, where the \emph{last} pair in the concatenation
is prescribed
to be the pair $(v_1,w_1)$:
\begin{quote}
  The Reversed Modified Post Correspondence Problem (RMPCP).

We are given a list of pairs of words $(v_1,w_1), (v_2,w_2), \ldots
(v_k,w_k)$ over the alphabet $\{\texttt{0},\texttt{1}\}$ such that
 $v_1$ and $w_1$ end with~\texttt{1}.
 The problem is to decide if there is a 
 sequence
$a_2,\ldots, a_m$
of indices $a_i \in \{2,\ldots,k\}$ such that
\begin{displaymath}
  v_{a_m}v_{a_{m-1}}\ldots v_{a_2}v_1
=   w_{a_m}w_{a_{m-1}}\ldots w_{a_2}w_1 \ .
\end{displaymath}
\end{quote}
This is of course just a trivial variation of the MPCP.
The translation of \eqref{prob-v} and
\eqref{prob-w} can now be applied directly. Moreover, the trailing zeros
issue disappears, since $v_1$ and $w_1$ end with~\texttt{1}.
This extra condition can be easily fulfilled by
appending a \texttt{1} to $v_1$
and $w_1$ if necessary.

\begin{proposition} \label{gt1/4}
  It is undecidable whether the language recognized by a PFA with 12
  states
 with cutpoint $\lambda=1/4$ is empty.
\end{proposition}
  \begin{proof}
    The above construction that has led to
 Proposition~\ref{ge1/2+}b 
    gives a set of $13\times13$ matrices such that
    the index sequence $a_1a_2\ldots a_m$ is a
    solution of the PCP if and only if
  $$e_1^TM_{a_1}M_{a_{2}}\ldots M_{a_{m-1}}M_{a_{m}}f > \tfrac 14,$$
  where $e_1$ is the first unit
  vector $(1,0,\ldots,0)$ and $\eeta $ is a vector with 6 zeros and 7 ones. For every other index sequence,
  the value of the expression is $<\frac14$.

  For the reversed MPCP the 
  first matrix $M_{a_1}=M_1$ is specified. Thus the product
  $  e_1^T\!M_{1}$ has a fixed value $\pi^T$, and we can replace it by
  this vector:
  \begin{displaymath}
\pi^T\!M_{a_{2}}\ldots M_{a_{m-1}}M_{a_{m}}f
  \end{displaymath} 
  This is the expression for the acceptance probability starting from an
  initial probability distribution $\pi$.
  The remaining matrix product is not allowed to use $M_1$, and this
  is easily ensured 
  by removing $M_1$ from $\mathcal{M}$.
  
  The original PFA goes from the start state to
  the 12 other states and never returns to the start state; thus we can eliminate the
  start state 
  and only use the $12\times12$
  submatrices for the remaining states. 
\end{proof}

We mention that with cutpoint $\frac12$ and the weak inequality $\ge \tfrac 12$
as acceptance criterion
 instead of
$> \tfrac 14$, we don't need the extra states $q_A$ and $q_R$, and the
number of states 
is reduced to~10. 

\section{Fixing the set of matrices by using a universal Turing
  machine}
\label{sec:fixed-matrices}
\label{sec:NEW}

We can achieve stronger and more specific results by
tracing back the undecidability of the PCP to the Halting Problem.
In particular, we will look at
a universal
Turing
machine and derive from it a ``universal'' PCP.
A universal Turing
machine is a fixed Turing Machine that can simulate any other Turing
machine.
In particular, the Halting Problem for such a machine is undecidable:
Given some initial contents of the tape, does the machine halt?
Sticking to one fixed machine allows us to choose a fixed set of
matrices that represents the PFA. The only variation is the starting
distribution $\pi$, or, in another variation, the vector $\eeta $
of output values.

\subsection{Constructing an MPCP for a Turing machine}
\label{sec:MPCP}

In order to adhere to the usual practice, we describe the translation
to the MPCP and not to the reversed MPCP. (For applying the RMPCP, the words simply have to be
reversed.)
Also, we temporarily use a general
 alphabet
for the word pairs of the MPCP.
 In the end, this alphabet will
 be encoded into the binary
alphabet
 $\{\texttt{0},\texttt{1}\}$ in order to be translated into a PFA.

The 
string
$  v_{a_1}v_{a_2}\ldots v_{a_m}
=  w_{a_1}w_{a_2}\ldots w_{a_m}$
 is built as a concatenation of successive
\emph{configurations} of the Turing machine,
separated by the marker \texttt{\#}.
The strings are built incrementally in such a way that the partial string
$v_{1}v_{a_2}\ldots v_{a_n}$
lags
one step (of the Turing machine)
behind the partial string
$w_{1}w_{a_2}\ldots w_{a_n}$.
Figure~\ref{fig:example-TM} shows an example.
\begin{figure}[htb]
  \centering
  \includegraphics[scale=1.05]{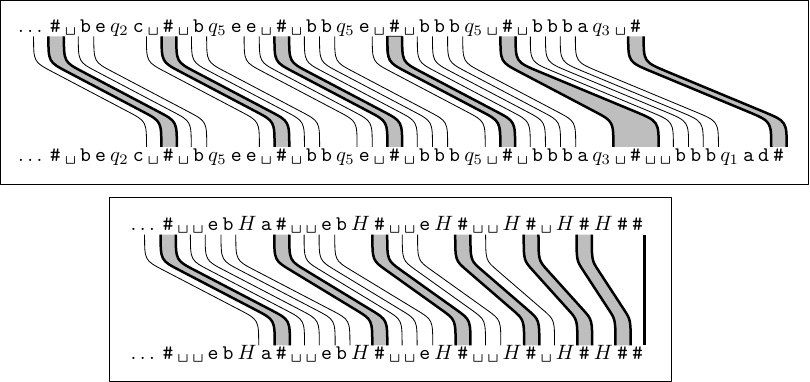} 
  \caption{Top: Building two partial strings
    $v_{1}v_{a_2}
    \ldots v_{a_n}$ (upper row in each box) and
    $w_{1}w_{a_2}
    \ldots w_{a_n}$ (lower row)
for a Turing machine with transition rules
$(q_2,\texttt c,\texttt e,L,q_5)$,
$(q_5,\texttt e,\texttt b,R,q_5)$,
$(q_5,\blank,\texttt a,R,q_3)$,
$(q_3,\blank,\texttt d,L,q_1)
$, among others.
For better visibility, the correspondences 
involving the separator \texttt{\#}
are highlighted.
Near the right end, the
{padding pair}
$(\texttt{\#},\texttt{\blank \#\blank })$ is used once in order to
produce extra blanks at the ends of the tape, preventing the state
symbol $q_3$ from becoming adjacent to the marker \texttt{\#}.
Bottom:
After the Turing machine has halted,
the tape is cleared and the common string
is completed.}

  \label{fig:example-TM}
\end{figure}
Following the common convention, a string such as
\texttt{\#\blank be{$q_2$}c\blank \#}
denotes the configuration where 
the Turing machine is in state $q_2$, the tape contains
the symbols \texttt{bec} padded by infinitely many blank symbols on
both sides (of which two are present in the string), and the
Turing machine is positioned over the third occupied cell, the one with the symbol~\texttt{c}.





The transition rules of the Turing machine are translated into pairs
$(v_i,w_i)$, as will be described 
below.
The important feature of this translation is
shown in Figure~\ref{fig:matrices-from-TMreductions}:
The input for the Turing machine is translated into the
starting pair
$(v_1,w_1)$. In the above translation to a PFA, 
leading to Proposition~\ref{gt1/4},
the
starting pair
$(v_1,w_1)$ affects only the starting distribution~$\pi$, whereas the
transition matrices $M_i$ depend only on the rules of the
Turing machine, which, for a
universal
Turing machine, are fixed!

\begin{figure}[htb]
  \centering
\fboxsep=5mm
\fbox{\begin{tikzpicture}[
  boxnode/.style={rectangle, draw=black, minimum size=6mm},
  style={->,shorten >=0.3pt,>={Stealth[round]},semithick}
  ]
  \node[boxnode,
  align=center] (TM)
  {TM 
    input tape $u$};
   
  \node[boxnode,
  align=center] (start)
  [below=8mm of TM]
  {starting pair $(v_1,w_1)$};

  \node[boxnode,
  align=center] (pi)
  [below=8mm of start]
  {starting distribution $\pi$};
  
 \node[boxnode,
] (rules)
  [right=23mm of TM]
  {\vphantom{p}TM 
    rules};

  \node[boxnode,
  ] (pairs)
  [below=8mm of rules]
  {other word pairs $(v_i,w_i)$};

  \node[boxnode,
  align=center] (M)
  [below=8mm of pairs]
  {\vphantom{g}matrices $M\in \mathcal{M}$};

\node (titleTM) [left=8mm of TM,align=left]{{Turing machine (TM):}};

\node (titlePCP) [below=8mm of titleTM.south west, anchor=north west]
{\vphantom{p}MPCP:};

\node (titlePFA) [below=8mm of titlePCP.south west, anchor=north west]
{PFA:};

  
  

\draw[->] (TM) -- (start);
\draw[->] (start) -- (pi);
\draw[->] (rules) -- (pairs);
\draw[->] (pairs) -- (M);
  
\end{tikzpicture}
}
\caption{How the PFA is constructed from a Turing machine 
  via an MPCP}
  \label{fig:matrices-from-TMreductions}
\end{figure}
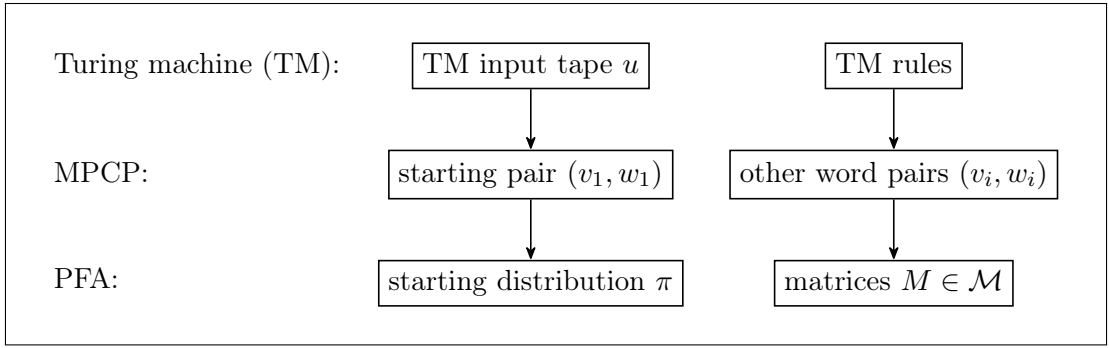

\subsection{List of word pairs of the MPCP}
\label{sec:wordpairs}

Since we want a MPCP with as few pairs as possible,
we review the construction of the MPCP from the Turing machine in detail.
We follow the construction
from 
Sipser~\cite[Section 5.2, Part~5, p.~187]{Sipser}
to ensure
that the configurations are padded with sufficiently many
blank symbols. This eliminates the need to deal with special cases
when the Turing machine reaches the ``boundary'' of the tape in the
 representation as a finite string.

 Let $\Gamma$ denote the
tape alphabet
including the blank symbol \blank,
and let $Q$ denote the set of states of the Turing machine.
 The words $v_i$ and $w_i$ of
the MPCP use 
the alphabet $\Gamma \cup Q \cup \{\texttt{\#},H\}$
 with 
two extra symbols:
a separation symbol~\texttt{\#} and
a halting
symbol~$H$.

\begin{itemize}
\item 
  If the input string for the Turing machine is $u \in \{
\texttt{0},\texttt{1}\}^*$,
  we define the \emph{starting pair}
$(v_1,w_1)=(\texttt{\#}, \texttt{\#\blank }q_0u\texttt{\blank \#})$, where $q_0$
is the start state of the Turing machine.
\end{itemize}
The other pairs $(v_i,w_i)$ are as follows:
\begin{itemize}
\item
  For \emph{copying} from the shorter string to the more advanced string, we have
  the pairs
  $(s,s)$ for all $s\in \Gamma
  $.
\item 
  We have another copying pair
  $(\texttt{\#},\texttt{\#})$, and 
  the \emph{padding pair}
  $(\texttt{\#},\texttt{\blank \#\blank })$. We are allowed to
  (nondeterministically) emit an additional blank symbol at
both ends of the configuration. 
\item 
  For each \emph{right-moving rule} $(q,s,s',R,q')$, the pair $(qs,s'q')$.
  (Such a rule means that when the Turing machine is
  in state $q$ and reads the tape symbol $s$,
  it 
  overwrites the
  symbol $s$ with $s'$, moves one step to the right on
the tape, and changes to state~$q'$.)
  
\item For each \emph{left-moving rule} $(q,s,s',L,q')$ and for each $t\in\Gamma$,
  the pair $(tqs,q'ts')$.
  
\item For each \emph{halting rule} $(q,s,-)$, the pair $(qs,H)$.
  The character 
  $H$ represents the fact that the machine has halted. 
\item 
  For each $s\in \Gamma 
  $, the
  \emph{erasing pairs} $(Hs,H)$ and $(sH,H)$. The halting symbol absorbs
  all symbols on the tape one by one.
\item 
 Finally,
the \emph{finishing pair}
$(H\texttt{\#\#},\texttt{\#})$. This is the only way how the two strings
can come to a common end.
\end{itemize}
In total, these are $3|\Gamma|+3$ pairs,
plus $|\Gamma|$ pairs
for each left-moving rule, 
plus one pair for each
right-moving or halting rule,
plus the starting pair.

We can save copying rules by encoding everything in binary, for
example using the codewords $\mathtt{bab},\mathtt{baab},\mathtt{baaab},\ldots$ over the
alphabet $\{\mathtt{a},\mathtt{b}\}$. No codeword is a substring of
another codeword, and the strings that are composed from such codewords
can always be uniquely decoded.
%
%
%
%
Then we can replace the $|\Gamma|+1$
copying pairs for $\Gamma\cup\{\texttt{\#}\}$ by just two pairs
$(\mathtt{a},\mathtt{a})$ and $(\mathtt{b},\mathtt{b})$.

This leads to a baseline of
 $2|\Gamma|+4$ pairs instead of $3|\Gamma|+3$. 

The Turing machine $U_{15,2}$ that we will use in the next section,
has only $|\Gamma|=2$ symbols. In this case, we would save one pair.
Thus, the number of matrices in the theorems of this section and
Section~\ref{sec:fix-start} can in fact be reduced by one. We have not
investigated the bit size of the transition matrices that we would get
from this approach.
 
\subsection{Using a universal Turing machine}
\label{sec:use-UTM}

We have looked at the parameters of various universal Turing machines
in the literature in order to see which ones give the smallest number
of pairs $(v_i,w_i)$ for our PCP.
The best result is obtained from the machine
$U_{15,2}$ of
Neary and Woods
\cite[Section 3.5]{4-small-2009}\rlap.\footnote
{see also
  \url{http://mural.maynoothuniversity.ie/12416/%1/Woods_FourSmall_2009.pdf%
  },
  with 
  incorrect page numbers, however}
Its tape alphabet, including the blank symbol, has size $|\Gamma|=2$.
It has 15 states, not counting the halting state. It has
15 left-moving rules,
14 right-moving rules, and 1 halting rule.
In terms of PCP pairs,
left-moving rules are more costly than right-moving rules,
but
we have the freedom to swap
left-moving with right-moving rules
by flipping the Turing machine's tape. We have to
switch to the nonstandard convention of
starting the Turing machine over the rightmost
input character, but this is easily accomplished
in the construction of
the starting pair $(v_1,w_1)$.
Thus, with
14 left-moving rules and
16 right-moving and halting rules, we
get
$3\times 2+3  + 14\times2+ 16=53$ pairs, plus the starting pair $(v_1,w_1)$
that encodes the input.

In some sense, this can be regarded as a \emph{universal} MPCP:
all pairs except the starting pair are fixed.

We can now establish a weaker version of
Theorem~\ref{thm-fixed-f}b,
with matrices of dimension $12\times12$
instead of $11\times11$.

\begin{proposition} 
  \label{thm-fixed-f-weak}
  There is a fixed set of 53 
  stochastic matrices
$\mathcal{M}''''$
of dimension $12\times 12$,
whose entries are multiples of $2^{-22}$,
  and a fixed $0$-$1$-vector $\eeta \in \{0,1\}^{12}$,
  for which the following question is undecidable\textup:

  Given a probability distribution $\pi\in \mathbb Q^{12}$ whose
  entries are {binary fractions},
 is there a product
  $M_1M_2\ldots M_m$, with $M_i\in\mathcal{M}''''$ for all $i=1,\ldots,m$,
such that
$$\pi^T\!M_1M_2\ldots M_m\eeta > \tfrac 14\ ?$$
In other words, is the language recognized by the PFA with starting
distribution $\pi$ and cutpoint $\lambda=\frac14$ nonempty?
\end{proposition}


\begin{proof}
 We specialize the proof of Proposition~\ref{gt1/4} to the current setting.
 The important point, as discussed above and shown
 in Figure~\ref{fig:matrices-from-TMreductions},
 is that
 the
 matrices in $\mathcal{M}$ depend only on the word pairs that reflect the rules of the
 universal Turing machine $U_{15,2}$, which are fixed, and we have
already calculated that there are 53 of these matrices.

We must not forget that the symbols of the alphabet
$\Gamma \cup Q \cup \{\texttt{\#},H\}$, in which the word pairs
$(v_i,w_i)$ of the MPCP are written, have to be encoded somehow into the binary
alphabet $\{\texttt{0},\texttt{1}\}$ in order to define the
matrices of the
PFA, and we have to ensure that the codes of
 $v_1$ and $w_1$ end with~\texttt{1}, for example by letting the
 code for \texttt{\#} end with~\texttt{1}.

There is one technicality that needs to be resolved.
The quantity $\gamma$ was
  required to be a common divisor of the matrix entries, and 
it depends on the maximum lengths $|v_i|$ and $|w_i|$ of the input words.
  However, the words $v_1$ and $w_1$ depend on the input tape,
  and thus, their lengths $|v_1|$ and $|w_1|$ cannot
be bounded in advance. (The remaining word pairs depend only on the Turing machine.)
  The solution is to carry out the imagined
  first transition (which is not encoded into a transition matrix in 
  $\mathcal{M}$, but determines the starting distribution~$\pi$) with a
  sufficiently small value of $\gamma$, namely $\gamma_1 =
  4^{-\max\{|v_1|,|w_1|\}}$,
where the lengths $|v_1|$ and $|w_1|$ are measured in the binary encoding.
  The other transitions from the state $q_A$ can
  be carried out with the fixed value $\gamma$ that is sufficient for those entries.
Table~\ref{tab:starting} shows the starting distribution $\pi$
resulting from this construction. Since $\pi$ is allowed to depend on
the input, we have solved the problem.

\begin{table}[htb]
  \centering
  \begin{tabular}{|@{\,}c@{\,}|c||@{\,}c@{\,}|c||@{\,}c@{\,}|c|}
    \hline
    state $q$& $\pi_q$
    &state $q$& $\pi_q$
    &state $q$& $\pi_q$
    \\[0.5pt]\hline
    $(\Phi_0,\Psi_0)$ & $\!\frac14(1-0.v_1)(1-0.w_1)\!$
&                        
    $(\Phi_0,\Phi_0)$ & $\frac18(1-0.v_1)^2$
&
$\{\Psi_0,\Psi_1\}$ & $\!\frac14(1-0.w_1)0.w_1\!
                      \raise 2,5pt\strut$
\\[3pt] 
$(\Phi_0,\Psi_1)$ & $\frac14(1-0.v_1)0.w_1$
&                        
    $\{\Phi_0,\Phi_1\}$ & $\frac14(1-0.v_1)0.v_1$
&
    $(\Psi_1,\Psi_1)$ & $\frac18  (0.w_1)^2$
\\[3pt] 
$(\Phi_1,\Psi_0)$ & $\frac14 \cdot 0.v_1(1-0.w_1)$
    &
$(\Phi_1,\Phi_1)$ & $\frac18(0.v_1)^2$     
&$q_A$&$\frac1{8}\gamma_1$       
\\[3pt]  
$(\Phi_1,\Psi_1)$& $\frac14\cdot 0.v_1\cdot0.w_1$
    &
$(\Psi_0,\Psi_0)$ & $\frac18(1-0.w_1)^2$
    &$q_R$&$\frac 12 -\frac1{8}\gamma_1 
    \lower2pt\strut$\\
    \hline
  \end{tabular}
  \caption{Starting probabilities $\pi$ for Proposition~\ref{thm-fixed-f-weak}}
  \label{tab:starting}
\end{table}

We have now established the existence of 53 fixed matrices
$\mathcal{M}''''$ and a finishing 0-1-vector $\eeta$ for which the decision
problem of
Proposition~\ref{thm-fixed-f-weak}
is undecidable.

\subsection{An efficient code}
\label{sec:efficientcode}

In order to say something about the entries of these matrices, we have
to be more specific about the way how the alphabet
 $\Gamma \cup Q \cup \{\texttt{\#},H\}$
is encoded. 
The words 
$v_i$ and $w_i$ that come from 
the Turing machine rules 
are
actually quite short: they have at most 3~letters.
More precisely, they consist of at
most one ``state'' symbol from
$Q \cup \{H\}$,
plus at most two letters from
the tape alphabet $\Gamma \cup \{\texttt{\#}\}$.
The Turing machine $U_{15,2}$ has $|Q|=15$ states and a tape alphabet
of size $|\Gamma|=2$.

In this situation, a variable-length 
 code is more efficient than
 a fixed-length code.
We can use 5-letter codes of the form \texttt{0****} for the 15 states plus the halting
state $H$. This leaves the 3-letter codes \texttt{1**} for the 3 symbols
$\Gamma \cup \{\texttt{\#}\}$, leading to word lengths bounded
by $5+3+3=11$.
In the binary automaton, the
transition
probabilities are therefore multiples of $2^{-11}$.
Since each box carries out two binary automata simultaneously,
the transition
probabilities are multiples of $4^{-11}$.
\end{proof}
With a weak inequality like $\ge \tfrac 12$ instead of
$> \tfrac 14$ as acceptance criterion, we don't need the extra states $q_A$ and $q_R$, and the
size of the matrices for which
Proposition~\ref{thm-fixed-f-weak}
holds can be reduced to 
$10\times 10$. 

As mentioned after
 Proposition~\ref{=1/4}, the cutpoint can be changed to a different
 value; then the constraint that the input distribution $\pi$
consists of binary fractions must be abandoned. Since the change only
affects the very first transition, the fixed matrix set
$\mathcal{M}$ remains unchanged.

The above variable-length code seems to be pretty efficient, but it
wastes one of the four codewords
\texttt{1**}. By looking at the actual rules of
the machine
$U_{15,2}$ and
fiddling with the code, it might be possible to improve the power $22$ in
the  denominator of the binary fractions.

\subsection{Example matrices}
\label{example}
For illustration, we compute
some matrices of the set
$\mathcal{M}''''$
explicitly.
We use the binary code
$\texttt{\#} \doteq \texttt{101}$,
$\texttt{\blank} \doteq \texttt{100}$.
The copying pair $(\blank,\blank)\doteq(\mathtt{100},\mathtt{100})$ is
then translated
into the block diagonal matrix 
\begin{displaymath}
\overline  M_{(\blank,\blank)} =
  \begin{pmatrix}
\small\dfrac 1{64}\left(\begin{array}{@{\,}r@{\ }r@{\ }r@{\ }r@{\,}}
16 & 16 & 16 & 16 \\
12 & 20 & 12 & 20 \\
 12 & 12 & 20 & 20 \\
9 & 15 & 15 & 25
                         \end{array}\right) & 0&0&0&0\\
          0 &
\small\dfrac 1{64}\small \left(\begin{array}{@{\,}r@{\ }r@{\ }r@{\,}}
16 & 32 & 16 \\
12 & 32 & 20 \\
9 & 30 & 25
\end{array}\right)
  &           0&0&0\\
          0&0&\dfrac1{64}\small\left(\begin{array}{@{\,}r@{\ }r@{\ }r@{\,}}
16 & 32 & 16 \\
12 & 32 & 20 \\
9 & 30 & 25
\end{array}\right)&0&0\\
0 & 0&0&\frac1{2^{22}}&1-\frac1{2^{22}}\\
0 & 0&0&0&1
\end{pmatrix},
\end{displaymath}
where the rows and columns correspond to the states as they are ordered in
Table~\ref{tab:starting}.\footnote
{Since $v_i=w_i=\blank$ in this case, we have a chance to compare
the straightforward $4\times 4$ construction of
the probability $\phi^2$
(the upper left block) with
the condensed representation with 3 states, in the two middle blocks.}
We use a bar in the notation $\overline  M_{(\blank,\blank)}$ to
remind us of the fact that
 we are dealing with the \emph{reversed} MPCP, and therefore the
 words should be reversed (which, in this case, has no effect
 because we have only one-letter words).

Let us look at the erasing
pair
$(\texttt{\#\blank},\texttt{\#})$. Here
the reversal does have an effect,
and the
strings
are actually
$(\texttt{\blank\#},\texttt{\#})\doteq(\mathtt{100\,101},\mathtt{101})$.
(The codewords don't have to be reversed.)
With these data,
the matrix $\overline M_{(\texttt{\#\blank},\texttt{\#})}$ looks as follows:
\begin{equation}\nonumber\label{erasing-matrix}
\begin{pmatrix}
\small\dfrac 1{512}\left(\begin{array}{@{\,}r@{\ }r@{\ }r@{\ }r@{\,}}
81 & 135 & 111 & 185 \\
54 & 162 & 74 & 222 \\
78 & 130 & 114 & 190 \\
52 & 156 & 76 & 228
\end{array}\right) & 0&0&0&0\\
          0 &
\small\dfrac 1{4096}\small \left(\begin{array}{@{\,}r@{\ }r@{\ }r@{\,}}
729 & 1998 & 1369 \\
702 & 1988 & 1406 \\
676 & 1976 & 1444
\end{array}\right)
  &           0&0&0\\
          0&0&\dfrac1{64}\small\left(\begin{array}{@{\,}r@{\ }r@{\ }r@{\,}}
9 & 30 & 25 \\
6 & 28 & 30 \\
4 & 24 & 36
\end{array}\right)&0&0\\
0 & 0&0&\frac1{2^{22}}&1-\frac1{2^{22}}\\
0 & 0&0&0&1
\end{pmatrix}
\end{equation}

%
%
%
%
%
%
%

Finally, as our most elaborate example, we consider a left-moving rule
of the Turing machine $U_{15,2}$ from~\cite{4-small-2009}:
%
%
%
$(q_9,\blank,\blank,L,q_1)$.
This was originally a right-moving rule, but has been converted into
a left-moving rule by flipping the tape.
It produces two word pairs, since $|\Gamma|=2$.
One of these pairs is
$(\texttt{b$q_9$\blank},
\texttt{$q_1$b\blank})$,
where
\texttt{b} is the other letter of the tape alphabet besides \blank.
Coding this letter as
$\texttt{b} \doteq \texttt{110}$
and
the states in the most straightforward way as
$q_1 \doteq \mathtt{00001}$ and
$q_9 \doteq \mathtt{01001}$, we get,
after reversal, the binary word pair
$(\texttt{\blank$q_9$b},
\texttt{\blank b$q_1$})\doteq
(\mathtt{100\, 01001\, 110},
\mathtt{100\, 110\, 00001})$
and
the following transition matrix:
\begin{displaymath}
\begin{pmatrix}
\!\small\dfrac 1{2^{22}}\!\left(\begin{array}{@{}r@{\ }r@{\ }r@{\ }r@{}}
786126 & 1151282 & 915762 & 1341134 \\
785180 & 1152228 & 914660 & 1342236 \\
785295 & 1150065 & 916593 & 1342351 \\
784350 & 1151010 & 915490 & 1343454
                         \end{array}\right)\hskip-1cm & 0&0&0&0\\
          0 &
\hskip-14mm\dfrac 1{2^{22}}\!\small \left(\begin{array}{@{}r@{\ }r@{\ }r@{}}
894916 & 2084984 & 1214404 \\
893970 & 2084828 & 1215506 \\
893025 & 2084670 & 1216609
\end{array}\right)\hskip-1cm
  &           0&0&0\\
          0&0&\hskip-1cm\dfrac1{2^{22}}\!\small\left(\begin{array}{@{}r@{\ }r@{\ }r@{}}
690561 & 2022654 & 1481089 \\
689730 & 2022268 & 1482306 \\
688900 & 2021880 & 1483524
\end{array}\right)\hskip-3mm&0&0\\
0 & 0&0& \!\frac1{2^{22}}\!&\!1{-}\frac1{2^{22}}\!\\
0 & 0&0&0&1
\end{pmatrix}
\end{displaymath}

\section{
  Output values instead of a set of accepting states}

\label{sec:fix-start}

 In the expression for the acceptance probability in
 \eqref{eq:accept},
 $\pi$ and $\eeta$ appear in symmetric roles.
We will now fix 
the starting
distribution $\pi$,
and in exchange,
we allow more general
values~$\eeta_q$.
%

In the classic model of a PFA, $\eeta$ is a 0-1-vector:
Once the input has been read and all probabilistic
 transitions have been made,
 acceptance is a yes/no decision.
 The state that has been reached is either accepting or not.

We can think of a general value $\eeta_q$ as a probability in a final acceptance
decision,
after the input has been read.
Another possibility is that
 $\eeta_q$ represents
 a \emph{prize} or \emph{value} that
that is gained when the process stops
in state~$q$, as in game theory.
Then
 $\eeta_q$
does not need to be restricted to the interval $[0,1]$.
In this view, instead of
the acceptance probability, we compute the \emph{expected} gain
(or loss) of the automaton.
Following Carl Page~\cite{page-1966}, who was the first to consider
this generalization,
we call $\eeta$
the \emph{output vector}
and $\eeta_q$
the \emph{output values}.
%
Mathematically, it make sense to take the outputs even from
 some (complex) vector space
(quantum automata?).

In our results,
the values $\eeta_q$ are restricted to $[0,1]$, and in fact, they have
an interpretation as probabilities.


    Turakainen~\cite
{turakainen69,Turakainen_1975} considered the most general setting, allowing    
arbitrary positive or negative entries also for the matrices $M\in \mathcal{M}$
and the vectors $\pi$ and $\eeta$.
He showed that the condition
\eqref{eq:accept} 
with these more general data
 does not define a more general class of languages than a
 classic PFA, see also
    \cite[\S 3.3.2, 
    pp.~120--126]{claus71}
or
  \cite[Proposition 1.1 in Section IIIB, p.~153]
  {paz71}.


\subsection{Saving one more state by maintaining four binary variables}
\label{sec:9states}

The PFA of Figure~\ref{fig:22PFA-merged}c mixes the PFAs for the three terms
$\phi\psi$, $1-\phi^2$, and $1-\psi^2$ by deciding \emph{in advance}
which sub-automaton they should enter.
As an alternative approach when arbitrary output values $\eeta_q$ are allowed, we can delay this decision to the end, when
we decide whether to (probabilistically) accept the input, and this
will allow us to further reduce the number of states by one.

The idea is to maintain four independent binary state variables
$\Phi',\Phi'',\Psi',\Psi''$
throughout the process.
Such a pool of variables is sufficient for any of the terms
$\phi\psi$, $1-\phi^2$, and $1-\psi^2$.
This would normally require $2^4=16$ states.
As discussed above, the combinations
$(\Phi'_0,\Phi''_1)$ and
$(\Phi'_1,\Phi''_0)$ need not be distinguished
and can be merged into one state, denoted by
$\{\Phi_0,\Phi_1\}$, and similarly for the $\Psi$ variables.
Thus, the overall number of states
is reduced from $
16$ to
$3\times 3=9$ combinations~$q$, 
one less than the 10 states in the three square boxes of
Figure~\ref{fig:22PFA-merged}c.

As we will see,
we have to set the nine entries~$\hat\eeta _q$ of the output vector to the
following values.
\begin{equation}
  \label{eq:acceptance-table}
\begin{tabular}{c|c@{ \ }c@{ \ }c}
&$(\Psi'_0,\Psi''_0)$&
$\{\Psi_0,\Psi_1\}$&
$(\Psi'_1,\Psi''_1)$\\\hline
\raise 2pt\hbox{\strut}%
  $(\Phi'_0,\Phi''_0)$&
1/2&1/2&1/4\\[2pt]
$\{\Phi_0,\Phi_1\}$&
    1/2& 5/8& 1/2\\[2pt]
$(\Phi'_1,\Phi''_1)$&
    1/4& 1/2& 1/2\\
\end{tabular}
\end{equation}
Beware that this is an output \emph{vector}
$\hat\eeta _q\in\mathbb{Q}^9$, which
has been arranged in $3\times3$ tabular form only
for convenience.
These output values
result from the contributions to the
three terms 
$  \tfrac12 \phi\psi$,
$  \tfrac14(1-\phi^2)$,
$  \tfrac14(1-\psi^2)$
of the overall acceptance probability as shown below, 
where  the states 
are arranged in the same matrix form as in~\eqref{eq:acceptance-table}:
\begin{displaymath}
  \frac12
  \begin{pmatrix}
    0&0&0\\
    0& 1/4& 1/2\\
    0& 1/2& 1\\
  \end{pmatrix}
  +
  \frac14
  \begin{pmatrix}
    1& 1& 1\\
    1& 1& 1\\
    0&0&0\\
  \end{pmatrix}
  +
  \frac14
  \begin{pmatrix}
    1& 1& 0\\
    1& 1& 0\\
    1& 1& 0\\
  \end{pmatrix}
  =
  \begin{pmatrix}
    1/2&1/2&1/4\\
    1/2& 5/8& 1/2\\
    1/4& 1/2& 1/2\\
  \end{pmatrix}
\end{displaymath}
The fractional values in the first matrix appear for the following
reason.
We have reduced the states for generating the acceptance probability
$\phi^2$ from 4 to 3 by merging two states into one.
Thus, when the PFA is, for example, in the state
$(\{\Phi_0,\Phi_1\}, (\Psi'_1,\Psi''_1))$,
it is ``really'' in one
of the two states
$(\Phi'_0,\Phi''_1, \Psi'_1,\Psi''_1)$
or
$(\Phi'_1,\Phi''_0, \Psi'_1,\Psi''_1)$,
each with a share 
of $50\,\%$.
If we consider the product $\phi\psi$ as built, say, from the
conjunction $(\Phi'_1,\Psi'_1)$,
ignoring the variables  $\Phi''$ and $\Psi''$,
only the second of these two states should lead to acceptance, and therefore
we get the fractional
output value
$ 1/2$. 

We can
change the cutpoint (for the original automaton, without the extra
states $q_A$ and $q_R$) from $\lambda=1/2$ to any 
rational value $\lambda$ strictly between
$0$ and $1$ 
by modifying the output values $\hat\eeta _q$ in
\eqref{eq:acceptance-table}:
By scaling both $\hat\eeta$ and $\lambda$ down by the same factor,
$\lambda$ can be brought arbitrarily close to 0.
On the other hand, by applying the transformation
$x\mapsto 1-\alpha(1-x)$ for some constant $0< \alpha\le 1$ to
 $\hat\eeta$ and $\lambda$, the cutpoint
$\lambda$ can be moved arbitrarily close to 1
 \cite[Proposition 1.4 of Section IIIB, p.~153]{paz71}.

\subsection{Making all transition probabilities positive}
\label{sec:positive-Kronecker}
By using an appropriate binary code,
we can ensure that all transition matrices
are strictly positive.
Rabin
calls such PFAs
\emph{actual automata} and
studies their 
properties~\cite[Sections~IX--XII, p.~242--245]{rabin1963},
see also~\cite[\S 3.2.3, 
pp.~115--118]{claus71}.

One can easily check that
the transition matrix $B(u)$
for the binary automaton
is positive except when
the string $u$
consists only of zeros or only of ones.
With only 3 symbols
$\Gamma \cup \{\texttt{\#}\}$ using
the 4 codewords \texttt{1**}, we can avoid
the all-ones codeword \texttt{111}
(as in the code used for the examples in Section~\ref{example}).

A state symbol other than $H$ never appears alone in a word $v_i$ or
$w_i$.
Thus, we can use the codeword \texttt{00000} for one of the original
states,
and thereby ensure that
the transition matrices $B(v_i)$ and  $B(w_i)$ are always positive.
As discussed earlier,
the encoded words $u_i$ and $v_i$ have at most 11 bits, and hence
the matrix entries 
are multiples of $2^{-11}$.
The entries of the
$3 \times 3$ transition matrix
\eqref{eq:3x3} are sums and products of entries of
the $2 \times 2$ matrices
$B(v_i)$ or $B(w_i)$, respectively, and are therefore positive
multiples of $2^{-22}$.
Each entry of the
 $9 \times 9$ transition matrix is obtained by multiplying appropriate entries
 of the two
 $3\times3$ matrices, and is hence a
 positive
 multiple of
 $2^{-44}$.
%
(To say it more concisely, the matrix
 is the Kronecker product, or tensor product, of the two
 $3\times3$ matrices.)

 More generally, the entries are multiples of
  $\gamma^2 = 16^{-\max\{|v_i|,|w_i|:
    1\le i \le k\}}$.

\subsection{Fixing everything except the output vector, proof of
  Theorem
  ~\ref{thm-fixed-pi}}
\label{sec:vary-f}

We will from now on use superscripts like  $M^i$ or $M^{(v_i,w_i)}$
 or $\smash{\overline M}^{(v_i,w_i)}$
for the matrices that are
associated to the word pairs
$(v_i,w_i)$, in order to distinguish them from the notation 
$M_j$ in the theorem below, where they are numbered in the order in
which they are used in the matrix product of the solution.

For the version with fixed starting distribution, we use the
original (unreversed) MPCP, where the \emph{first} word pair in the
solution, 
and hence the \emph{last} matrix in the matrix product, is fixed.

We can save a matrix by observing that the
 \emph{last} word pair in the
PCP 
is also known: It is the finishing pair
$(H\texttt{\#\#},\texttt{\#})$, and like the
starting pair, this pair is used nowhere else.
(This 
is the only
pair, besides the starting pair, that has a different number of
\texttt{\#}'s in the two components, and it is the only possibility
how the
string $v_{1}v_{a_2}\ldots v_{a_n}$ can catch up with
 the string
$w_{1}w_{a_2}\ldots w_{a_n}$.)

For clarity, we formulate
the (unreversed) Doubly-Modified Post Correspondence Problem (2MPCP),
with 
two special pairs:
a starting pair $(v_1,w_1)$ and a finishing pair $(v_2,w_2)$:
\begin{quote}
We are given a list of pairs of words $(v_1,w_1), (v_2,w_2), \ldots
(v_k,w_k)$ over the alphabet $\{\texttt{0},\texttt{1}\}$ such that
 $v_2$ and $w_2$ end with a~\texttt{1}.
 The problem is to decide if there is a 
 sequence
$a_2,\ldots, a_{m-1}$
of indices $a_i \in \{3,\ldots,k\}$ such that
\begin{displaymath}
  v_1v_{a_2}v_{a_{3}}\ldots v_{a_{m-1}}v_2
=   w_1w_{a_2}w_{a_{3}}\ldots w_{a_{m-1}}w_2 \ .
\end{displaymath}
\end{quote}
The PFA
starts deterministically in the state
$(\Phi'_0,\Phi''_0,\Psi'_0,\Psi''_0)$. 
Thus,
the 2MPCP has a solution if and only if the following inequality can
be solved:
\begin{equation}
  \label{eq:2MPCP}
e_1^T M^2M^{a_{m-1}}\ldots M^{a_2}M^1 \hat f \ge \tfrac 12,
\end{equation}
where $\hat f$ is the output vector defined in
\eqref{eq:acceptance-table}.
The matrix $M^2 =
M^{(H\texttt{\#\#},\texttt{\#})}
$ comes from the finishing pair
$(H\texttt{\#\#},\texttt{\#})$
and
is fixed, and $M^1$ depends on the input tape $u$ of
the Turing machine. With the substitutions
\begin{align*}
\pi^T & :=  e_1^T M^2,
  \\
 f & := M^1 \hat f ,  
\end{align*}
we can remove $M^2$ from the set of matrices $\mathcal{M}$, and
this directly leads to part (a) of the following theorem:

\NasuHondaFixedMatricesReverse*

\begin{proof}
For part (a), everything has already been said except for observing that
the entries of $f=M^1\hat f$ are in the interval
 $[\frac14,\frac58]$ because
  $M^1$ is a stochastic matrix
  and the entries of~$\hat f$
  are in that interval.

  For part (b), we add the same two states $q_A$ and $q_R$ as in
  Figure~\ref{fig:mix}b (p.~\pageref{fig:mix})
  and
Figure~\ref{fig:22PFA-merged}c, with $\gamma = 2^{-44}$.
  Initially, we select the original start state
  $(\Phi'_0,\Phi''_0,\Psi'_0,\Psi''_0)$ and the state $q_A$ each with
  probability $\frac12$.
  Denoting by $\pi_0$ the corresponding vector with two $\frac12$ entries,
  the initial distribution $\pi$ is then defined as
  \begin{equation}
    \label{eq:pi0-with-extra-states}
\pi_0^T M^2 =: \pi^T,
  \end{equation}
  and its entries are multiples of
  $\frac1{2^{45}}$. 
  %
  %

  The matrix $M^1$ is constructed from the starting
  pair $(v_1,w_1)$, and it uses the value
  \begin{equation}
    \label{eq:gamma1}
\gamma_1= 1/16^{\max\{|v_1|,|w_1|\}},    
  \end{equation}
where $|v_1|$ and $|w_1|$ are
the lengths after the binary encoding.
  
The output values of the  extra states are defined as
  $\hat f_{q_A} = 1/8$ and
  $\hat f_{q_R} = 0$.
Since the remaining output values
in $\hat f$ are multiples of $1/8$, 
the value
  $\hat f_{q_A} = 1/8$ is small enough
to ensure that it does not turn an acceptance probability $<\frac14$
into a probability $>\frac14$.
\end{proof}


To give a concrete example, here is the transition matrix
$M^{(\texttt{\#\blank},\texttt{\#})}\in \mathcal{M}''$
 for the erasing
pair
$(\texttt{\#\blank},\texttt{\#})\doteq(\mathtt{101\,100},\mathtt{101})$:\footnote
{If the words $v_i$ and $w_i$ weren't reversed between 
the MPCP
and the RMPCP,
the upper left $9\times 9$ block would be the Kronecker product of the two
middle $3\times3$ blocks in the corresponding
 $12\times12$ matrix
 $\smash{\overline M}^{(\texttt{\#\blank},\texttt{\#})}\in\mathcal{M}''''$
of Proposition~\ref{thm-fixed-f-weak}
for this pair, which was shown on p.~\pageref{erasing-matrix}
in Section~\ref{example}. If we substitute this Kronecker product as
it stands,
we get 
the matrix
$\smash{\overline M}^{(\texttt{\blank\#},\texttt{\#})}\in\mathcal{M}''$ of the opposite erasing pair.
}
\begin{displaymath}
   \frac 1{2^{18}}\small
\left(\begin{array}{r@{\ }r@{\ }r|r@{\ }r@{\ }r|r@{\ }r@{\ }r|cc}
3600 & 12000 & 10000 & 15840 & 52800 & 44000 & 17424 & 58080 & 48400 &0&0\\
2400 & 11200 & 12000 & 10560 & 49280 & 52800 & 11616 & 54208 & 58080 &0&0\\
1600 & 9600 & 14400 & 7040 & 42240 & 63360 & 7744 & 46464 & 69696 &0&0\\
\hline
 3420 & 11400 & 9500 & 15624 & 52080 & 43400 & 17820 & 59400 & 49500 &0&0\\
2280 & 10640 & 11400 & 10416 & 48608 & 52080 & 11880 & 55440 & 59400 &0&0\\
1520 & 9120 & 13680 & 6944 & 41664 & 62496 & 7920 & 47520 & 71280 &0&0\\
\hline
 3249 & 10830 & 9025 & 15390 & 51300 & 42750 & 18225 & 60750 & 50625 &0&0\\
2166 & 10108 & 10830 & 10260 & 47880 & 51300 & 12150 & 56700 & 60750 &0&0\\
1444 & 8664 & 12996 & 6840 & 41040 & 61560 & 8100 & 48600 & 72900 &0&0\\
\hline\raise2pt\strut
 0 & 0 & 0 & 0 & 0 & 0 & 0 & 0 & 0 & \frac1{2^{26}} & 2^{18}{-}\frac1{2^{26}} \\[4pt]
0 & 0 & 0 & 0 & 0 & 0 & 0 & 0 & 0 & 0 & 2^{18}
\end{array}\right)
\end{displaymath}

\subsection{Uniqueness of the solution}
\label{sec:unique}
In both parts of 
Theorem~\ref{thm-fixed-pi},
we can achieve that every problem
instance
that we construct has a unique solution if it has a solution at all.
This comes at the cost 
of increasing the number of matrices and relaxing the
bound on the denominators.
The Turing machine itself is deterministic.
The MPCP loses the determinism through
the padding pair
$(\texttt{\#},\texttt{\blank \#\blank })$.
We omit this pair and replace it by other word pairs.
In particular,
if a state symbol $q$ is adjacent to the separation symbol
\texttt{\#}
and is in danger of ``falling off'' the tape,
this must be treated as if a $\blank$ were present. 
This leads to one extra word pair for each state
plus
one extra word pair for each left-moving rule.\footnote
{In contrast to the construction found in most textbooks,
  we cannot assume that the Turing machine never moves to the left of
  its initial position, since we want to keep our Turing machine small.}
%
%

Since,
in addition to the starting pair, also the finishing pair $(v_2,w_2)$ is fixed in
the 2MPCP, the solution to the 2MPCP, and hence the matrix product
$M_1\ldots M_m$,
becomes unique.
(In the normal PCP or MPCP, a solution could be extended by appending arbitrary
copying pairs.)

We emphasize that this uniqueness property holds only for output
vectors~$f$ that are constructed according to the proof of
Theorem~\ref{thm-fixed-pi}. It is obviously impossible to achieve
uniqueness for every vector~$f\in[0,1]^d$.

One can check that uniqueness carries over, with the same provisos, to
the other theorems 
of this section.

\subsection{Eliminating the output vector, proof of
  Theorem
  ~\ref{thm-fixed-f}}
\label{eleven}

We will transfer these results to the
classic setting with a \emph{set} of
accepting states instead of an
output vector $f$.
The
{set} of
accepting states
will be fixed, and the input should come
through the starting distribution $\pi$.
Consequently,
the Turing machine input should be coded, via the first matrix in the
matrix product,
into the starting distribution $\pi$.  Hence we
\emph{reverse} the PCP again, as in Sections
\ref{sec:RMPCP} and~\ref{sec:fixed-matrices}.
We construct a set $\overline{\mathcal{M}}
$ of 53 positive
$9\times 9$ matrices,
including a matrix for
 the finishing pair
 $(H\texttt{\#\#},\texttt{\#})$,
in the same way as in the proof of Theorem~\ref{thm-fixed-pi}a,
but with reversed words.
We refrain from formulating
the  \emph{Doubly-Modified {Reversed} Post Correspondence Problem} (2MRPCP).
We just observe that,
in the expression for the acceptance probability
\begin{equation}
  \label{eq:2MPCP-b}
e_1^T M^2M^{a_{m-1}}\ldots M^{a_2}M^1 \hat f 
\end{equation}
from \eqref{eq:2MPCP},
the matrix that depends on the input tape $u$ of
the Turing machine now appears as
the matrix $M^2$ at the beginning of the product,
and the matrix that comes from the finishing pair
$(H\texttt{\#\#},\texttt{\#})$
is the last matrix $M^1$.
Then, $\pi^T :=  e_1^T M^2$ is the variable input to
the problem,
and $f := M^1 \hat f$ is some fixed vector of 
output values $f_q \in [\frac14,\frac58]$.
The acceptance probability becomes
\begin{displaymath}
\pi^T M^{a_{m-1}}\ldots M^{a_2} f.
\end{displaymath}
What remains to be done is to get rid of the fractional
values in the output vector $f$.
We will use two methods to convert a PFA with
an output vector $f$ with entries from $[0,1]$ to into one with
a 0-1 
vector~$f$.
The first method is a general method that does not change the
recognized language. It doubles the number of states,
and it maintains positivity.\footnote{\label{foot-square}%
  There is a method in
  the literature with the same effect, but it \emph{squares}
  the number of states,
see~\cite[proof of Theorem~1, p.~308]{turakainen69},
\cite[Step~V,
pp.~123--124]{claus71},
or Section~\ref{sec:output-0-1-epilogue}.}
This is formulated as
 part~(a) in the following theorem.
 As an alternative,
we will start with the construction of Theorem~\ref{thm-fixed-pi}b
and we will take the liberty to change the recognized language by adding a
symbol to the end of every word. This works without adding extra states
beyond the states $q_A$ and $q_R$ that are already there,
and it
will lead to
 part~(b) of the following theorem.

 \NasuHondaEleven*
 \begin{proof}
(a)
\smash{\hbox to 0pt{\vbox{
  \begin{enumerate}\def\makelabel#1{\hss}
  \item \ \label{proof-b-f}\end{enumerate}}\hss}}%
We interpret the output values $f_q$ as probabilities.
 If we arrive in state $q$ after reading the input, we still have
 to make a random decision whether to accept the input. 
The idea is to generate the randomness for making this acceptance decision
already
\emph{when} each 
symbol is read, and not \emph{afterwards},
in the end. Every state $q$ of the original PFA comes now in two versions, $q^+$ and $q^-$.
The transition probabilities to $q^+$ are multiplied by $\eeta_q$,
and the transition probabilities to $q^-$ are multiplied by $1-\eeta_q$.
The accepting states are the states $q^+$.

In terms of matrices, this can be expressed as follows.
Let  $M
$ be written in column form as
\begin{displaymath}
  M =
  \begin{pmatrix}
    m_1 \ m_2 \ \cdots\ m_{9}
  \end{pmatrix}
.
\end{displaymath}
This is converted to the following $18\times 18$ matrix for the set
$\mathcal{M}'$, arranging the states in the order
$q_1^+,\ldots,q_9^+,q_1^-,\ldots,q_9^-$:
\begin{displaymath}
  \begin{pmatrix}
    \eeta_1m_1 & \eeta_2m_2 &\cdots&\eeta_9 m_{9} &    
    (1-\eeta_1)m_1 & (1-\eeta_2)m_2 &\cdots&(1-\eeta_9) m_{9}\\
    \eeta_1m_1 & \eeta_2m_2 &\cdots&\eeta_9 m_{9} &    
    (1-\eeta_1)m_1 & (1-\eeta_2)m_2 &\cdots&(1-\eeta_9) m_{9}\\
  \end{pmatrix}
\end{displaymath}
Similarly,
the starting distribution
$\pi^T = (
    \pi_1 ,\pi_2 ,\ldots,\pi_{9}
    )$
is replaced
    by
$ (    
    \eeta_1\pi_1 ,\allowbreak \eeta_2\pi_2 ,\allowbreak\ldots,\allowbreak\eeta_9 \pi_{9},\allowbreak \     
    (1-\eeta_1)\pi_1 ,\allowbreak (1-\eeta_2)\pi_2 ,\allowbreak\ldots,\allowbreak(1-\eeta_9) \pi_{9})
    $.
The matrix consists of two equal $9\times 18$ blocks,
in accordance with the fact that
the distinction between $q^+$ and $q^-$
has no influence on the next transition.
%



As the output values $\eeta_q \in\bigl\{\frac14,\frac12,\frac58\bigr\}$ are
multiples of $\frac18$, all resulting probabilities are multiples of
$\frac1{2^{47}}$.

(b) 
The idea is to add to the set of matrices a matrix
$M^\infty$ that is necessarily the last matrix in
any solution, without imposing this as a constraint.
 
We start by constructing a set $\overline{\mathcal{M}'''}$
 of 53 matrices of size
$11\times 11$,
including a matrix for
 the finishing pair
 $(H\texttt{\#\#},\texttt{\#})$,
in the same way as in the proof of Theorem~\ref{thm-fixed-pi}b,
but with reversed words.
The states $q_A$ and $q_R$ are now already present.

We want to emulate the acceptance criterion of Theorem~\ref{thm-fixed-pi}b:
\begin{equation}
  \label{eq:variable-pi-f-hat}
\pi^T\!M_1M_2\ldots M_m\hat\eeta > \tfrac 14  
\end{equation}
Here, the variable vector $\pi^T$ as given by
\eqref{eq:pi0-with-extra-states} has already swallowed the
matrix $M^2$ representing the input tape of the Turing machine;
However, $\hat\eeta$ is the fixed output vector constructed in the proof of
Theorem~\ref{thm-fixed-pi}b with the values
\eqref{eq:acceptance-table} extended by the values
$\hat f_{q_A} = 1/8$ and $\hat f_{q_R} = 0$ for the two additional
states.
We do not yet merge the last matrix $M_m$ with $\hat f$.

To the 53 matrices
 $\overline{\mathcal{M}'''}$, 
 we add an extra
``final'' transition matrix $M^\infty$.
We declare $q_A$ to be the unique accepting state.
 In the transition  $M^\infty$,
 each state $q$ goes to $q_A$ with probability $\hat f_q$,
 and to $q_R$ with the complementary probability $ 1-\hat f_q$.
 This rule applies equally to the state $q_A$, which goes to itself
 with probability  $\hat f_{q_A}=1/8$, and otherwise goes to $q_R$.
 The state $q_R$ remains an absorbing state.

 It is clear that adding $M^\infty$ at the end of the product
\eqref{eq:variable-pi-f-hat} and accepting in state $q_A$ has the same
effect as accepting with the output vector $\hat\eeta$.
However, a priori we are not sure that
 $M^\infty$ really comes at the end of the product. 

The acceptance probability of our PFA is given as
\begin{equation}
 \label{eq:2MPCP-c}
  \pi^T
M_1M_2\ldots M_{m-1}M_m  
e_{q_A} 
,
\end{equation}
where the vector
of output values is the unit vector  $e_{q_A}$
corresponding to the accepting
state~$q_A$.

We will now argue that in any product of this form
with matrices $M_j$ from $\overline{\mathcal{M}'''}\cup \{M^\infty\}$
that is larger than $\frac14$,
the matrix $M^\infty$ must appear in the last position $M_m$, and it cannot
appear anywhere else.

 If we never use the matrix
 $M^\infty$
  in the matrix product, the only chance of reaching $q_A$ comes from
starting in $q_A$ at the beginning and staying there, and the probability
for this is negligibly small.
(Even the empty matrix product is not a solution:
Remember that $\pi^T$ is defined in
\eqref{eq:pi0-with-extra-states}
as $\pi^T = \pi_0^T M^2$,
where $M^2$
comes from the word pair representing the input of the Turing
machine.
Already in $M^2$, the probability $\gamma_1$ of remaining in $q_A$,
as given by \eqref{eq:gamma1},
is very small.)

On the other hand, when we use the matrix
$M^\infty$, the PFA will arrive in state $q_A$ or~$q_R$.
Any further matrices after $M^\infty$
 reduce
 the probability of staying in~$q_A$ by a factor $1/8$ or smaller,
 hence they will not lead to solutions.

Thus we can assume without loss of generality that
$M^\infty$ is the last matrix $M_m$ in the product, and that it is used only
in that position.
The acceptance probability is then the same as if the output vector
$\hat f$ had been used instead of $M^\infty$.
(Algebraically, $M^\infty e_{q_A} = \hat f$.)

Thus, the expression \eqref{eq:2MPCP-c} has the
same value as \eqref{eq:variable-pi-f-hat}, 
and it is already of the correct form for our claim.
The vector $e_{q_A}$ describes a unique accepting state.
As mentioned, we have changed the language recognized by the PFA by adding the
symbol  corresponding to $M^\infty$
to the end of each word, but this does not affect the emptiness question.

We can save one matrix by remembering
that the last word pair in the PCP solution is
always the finishing pair
$(v_2,w_2)=(H\texttt{\#\#},\texttt{\#})$, and this is used nowhere
else.
We therefore impose without loss of generality that the corresponding matrix
 $M^2=
 \overline M^{(H\texttt{\#\#},\texttt{\#})}$
 is the penultimate matrix $M_{m-1}$ in the product
 {before} $M^\infty$, and this matrix is used nowhere else. 
Accordingly,  we replace
 $M^2$ and
 $M^\infty$ by one matrix $M^{\mathrm{new}}=M^2M^\infty$, reducing the number of
 matrices back to $53$.
 Since the entries of
 $M^\infty$ are multiples of $\frac1{8}=\frac1{2^3}$, the entries of the new matrix
 are multiples of~$\frac1{2^{47}}$.

 This modification also ensures that the solution is unique:
 Since we now have enforced that the matrix product
 \eqref{eq:variable-pi-f-hat} ends with  $M^2M^\infty$, in term of
 the original set of matrices,
 we are only considering solutions of the MPCP that
 end with $(H\texttt{\#\#},\texttt{\#})$, and these are unique. 
\end{proof}

 
\subsection{Reduction to 2 input symbols, proof of
  Theorem
  ~\ref{thm-fixed-f-2-matrices}}
\label{sec:coding}

We have already 
used the reduction to a binary alphabet 
in the
 proof of Theorem~\ref{thm-condon-lipton}
(Section~\ref{sec:formalxx}), but now we will look at an explicit construction.
This method has been described, in a more general context, by
Hirvensalo in 2007, \cite[Step 3 of Section~3]{hirvensalo07} or \cite[p.~5]{tHirvensalo06a},
see also Section~\ref{sec:integer-hirv}.

\begin{restatable} 
  {lemma}
{CodingBinary}  
  \label{coding}
  Consider a PFA $A$ with input alphabet $\Sigma$ of size
  $k=|\Sigma|>2$,
  and
  let $\tau \colon \Sigma^* \to
\{\mathtt{a},
\mathtt{b}\} 
^*$ be a 
  coding function
  using 
  the codewords
  $
  \mathtt{b},
\mathtt{ab},
\mathtt{aab},
\ldots,
\mathtt{a}^{k-2}\mathtt{b},
\mathtt{a}^{k-1}
$. 

Then there is a PFA $A'$ with input alphabet
$\{\mathtt{a},\mathtt{b}\}$
that accepts
each word $\tau (u)$ with the same
probability as $A$ accepts $u\in \Sigma^*$.
  Words that are not of the form $\tau (u)$ are 
  accepted with probability~$0$.

   The number of states is multiplied by $k-1$ in this construction.
\end{restatable}

\begin{proof}
  Suppose $A$ has 
  transition matrices
  $M_1,\ldots,M_k$ corresponding to the
   $k$ input symbols.
We construct a PFA $A'$ that does the decoding in a straightforward
way.
It maintains the number of \texttt{a}'s that have been seen in a
counter
variable $i$ in the range
$0\le i \le k-2$.
 In addition, it maintains the state $q\in Q$ of the
original PFA $A$.
Thus, the state set of $A'$ is $Q'= \{0,\ldots,k-2\} \times Q$.
Initially,
$q$ is chosen according to the starting
distribution of $A$, and
 $i=0$.
\begin{itemize}
\item If $A'$ reads the letter \texttt{b}, it
changes the state
  $q$
  to a random new state 
  according to the transition matrix $M_{i+1}$, and resets $i:=0$.
\item If $A'$ reads an \texttt{a} and $i<k-2$, it increments the counter: $i:=i+1$.
\item If $A'$ reads an \texttt{a} and $i=k-2$, it changes the state
  $q$
  to a random new state 
  according to $M_k$, and resets $i:=0$.
\end{itemize}
An input is accepted if $i=0$ and $q$ is an accepting state of $A$.

The 
transition matrices 
for the symbols
\texttt{a} and \texttt{b} can be written in block form
as
\begin{displaymath}
  M'_\mathtt{a}=
  \begin{pmatrix}
  0 & I &  &    &&  \\
    & 0 & I &   &&  \\
    &   & 0 &   &&  \\[-0.8ex]
 & 
 & 
 & \!\!\ddots
 && 
\\
  &  & &
 & 0 & I\\
 M_{k}\!\!\! &  & & &  & 0
\end{pmatrix}
\text{ and }
M'_\mathtt{b}=
  \begin{pmatrix}
    M_1 & 0 & \cdots & 0\\[0.4ex]
    M_2 & 0 & \cdots & 0\\[0.4ex]
    M_3 & 0 & \cdots & 0\\[0.4ex]
    \vdots & \vdots & 
 & \vdots\\[0.9ex]
    {M_{k-1}} & 0 & \cdots & 0
  \end{pmatrix}.
  \qedhere
\end{displaymath}
\end{proof}
The construction works more generally for any prefix-free code.
The set of states $Q'$ 
will
have the form $K \times Q$, where the states in $K$
do the decoding.  

Applying this to Theorem
  ~\ref{thm-fixed-f}b, we get:

\NasuHondaTwoMatrices*

The number $572=52 \times 11$ of states is an overcount. For example, the absorbing
state $q_R$ can be left as is and need not be multiplied with~52.

If we are more ambitious,
we can achieve that all matrix
 entries are from the set $\{0,\frac12,1\}$,
 as in Theorem~\ref{thm-condon-lipton}, instead of multiples of $2^{-47}$.
 We apply the 
 technique from 
item (b) in the proof of that theorem
(Section~\ref{sec:formalxx}):
We simply add a block of 46 padding \texttt{a}'s after every
codeword.\footnote
{%
We can roughly estimate the required number of states as follows.
Let $|Q|=11$ be the number of states of the original automaton,
and let $k=|\Sigma|=53$ be its number of symbols.
For each combination in $Q\times \Sigma$, whenever
the algorithm in
Lemma~\ref{coding} asks to ``change the state
  $q$
  to a random new state 
  according to $M_i$'', we have to set up a binary decision tree of
  height 47 to determine the next state. We can think of this tree as
follows: For
  a random
  number
  $x=0.x_1x_2\ldots x_{47}$,
we want to
  determine which of $|Q|$ intervals $[0,c_1],
  (c_1,c_2],
  (c_2,c_3],\ldots,
  (c_{|Q|-1},1]$  contains $x$,
  by
  looking at the successive bits $x_j$ of~$x$.
  This tree has at most $(|Q|-1)\times 46$ nodes where the
  outcome has not been decided:
  each such node lies on a root-to-leaf path to some interval endpoint~$c_i$.
  In addition we need up to $46\times |Q|$ states for the situation when the next
  state has been decided and the algorithm only needs to count to the end
  of the padding block.
  In total, this gives an upper bound of $(k-1)|Q| + |Q|k(|Q|-1)46 + 
  46 |Q|$ states, which is $572 + 46\times 11 \times(53\times 10+1)
  = 269\,258
  $.
}

\section{Using integer matrices}
\label{sec:thm-variable-matrices}

As an alternative to constructions involving PFAs only,
we start from
matrix product problems of arbitrary \emph{integer} matrices.
These can be converted to stochastic matrices at some cost, but
overall,
we get the automata with the smallest number of states that are known
for having an undecidable Emptiness Problem.

\subsection{The smallest number of states without regard to the size
  of the alphabet, 
  Theorem~\ref{thm-claus-9-states}}
\label{sec:fewest-possible}

\ClausSmallMatrices*

\begin{proof}

 Neary~\cite{neary:PCP5:2015} has shown that
 the PCP with five word pairs
 is undecidable.
Suppose the word pairs
$(v_1,w_1),\,(v_2,w_2),\,(v_3,w_3),\,(v_4,w_4),\,(v_5,w_5)$
are encoded over the alphabet $\Sigma=\{1,2\}$.
We denote
the base-3 
value of  $u=u_1u_2\ldots u_n\in \Sigma^*$ by
 $
(u)_{3}=
 \sum_{j=1}^n u_j3^{n-j}$.

\subsubsection{Step  \ref*{thm-claus-9-states}.1. Modeling the PCP by
  integer matrices}
\label{sec:claus1}
Following Claus \cite[Definition 5, p.~143]{claus81}, we define the
following $6\times 6$ matrix $A(v,w)$ for words $v,w \in \{1,2\}^*$.
\begin{equation}
  \label{eq:mat-Claus-3-original}
  A(v,w)=
  \begin{pmatrix}
1 & -2 (v)_{3} & 2 (w)_{3} & ((v)_{3})^{2} & ((w)_{3})^{2} & -2(v)_{3} (w)_{3}
    \\
0 & 3^{|v|} & 0 & -(v)_{3}\cdot 3^{|v|} & 0 & (w)_{3}\cdot 3^{|v|} &
    \\
    0 & 0 & 3^{|w|} & 0 & (w)_{3}\cdot 3^{|w|} & -(v)_{3}\cdot 3^{|w|} &
    \\
    0 & 0 & 0 & 3^{2  |v|} & 0 & 0 & 
    \\
    0 & 0 & 0 & 0 & 3^{2  |w|} & 0 & 
    \\
    0 & 0 & 0 & 0 & 0 & 3^{|v| + |w|} &
    \\
  \end{pmatrix}
\end{equation}
It is not straightforward to see, but can be checked by a
calculation that $A $ satisfies the multiplicative law
\cite[Lemma~2
]{claus81}
\begin{equation}
  \label{eq:multi-gamma-forward}
A (v_1,w_1)
A (v_2,w_2) =
A (v_1v_2,w_1w_2).
\end{equation}
%
%
With the vectors $
e_1^T=(1,0,0,0,0,0)$ and $ f_1^T= (1,0,0, -1, {-1}, {-1})$, one gets
\begin{equation}
  \label{eq:check-u=v}
e_1^T A(v,w) f_1 =  1-\bigl((v)_{3}-(w)_{3}\bigr)^2 .
\end{equation}
This expression
checks whether $v=w$:
Since all data are integral, $e_1A (v,w)\eeta_1>0$ iff
$v=w$.

From the PCP word pairs $(v_i,w_i)$,
we now construct the  matrices $B_i=A(v_i,w_i)$.
By repeated application of \eqref{eq:multi-gamma-forward}, 
$
A(v_{a_1}\ldots v_{a_m}, w_{a_1}\ldots w_{a_m})=
B_{a_1}B_{a_{2}}\ldots 
B_{a_m}$.
Thus,
the index sequence $a= a_1a_2\ldots a_m$ is a 
solution of the PCP if and only if
    \begin{equation}
      \label{eq:matrix-C-2}
e_1^T
B_{a_1}B_{a_{2}}\ldots B_{a_{m-1}}B_{a_m}
\eeta_1 > 0.
\end{equation}

\subsubsection{Step  \ref*{thm-claus-9-states}.2. Introducing a unit vector as an ending vector}
\label{sec:step2-claus}

This step and the next one are
 a standard part of the proof of Turakainen's
Theorem~\cite{turakainen69,Turakainen_1975} (mentioned in the introduction of
Section~\ref{sec:fix-start}), and, in one form or another, they are
described in many places, see 
\cite[Section 3.3.3, p.~120--125]{claus71},
\cite[
Steps 4b--4c in Section~2, p.~237--238]{blondel-canterini-03}),
or 
\cite[Steps 4--6 in Section~3]{hirvensalo07}.)

First we essentially introduce a new
 final state (state~7), extending the matrices to size
$7\times 7$:
\begin{equation}\label{eq:matrices-D}
  D_i =
  \begin{pmatrix}
    B_i &  B_i f_1\\
    0 & 0\\
\end{pmatrix} .
\end{equation}
One can then check that
\begin{equation}
  \label{eq:D-B}
  D_{a_1}\ldots D_{a_m}
  =
  \begin{pmatrix}
      B_{a_1}\ldots B_{a_m}
 &    B_{a_1}\ldots B_{a_m}
 f_1\\
    0 & 0\\
\end{pmatrix}.
\end{equation}
With the final vector
$f_2 = e_7$,
the expression
    \begin{equation}
      \label{eq:matrix-DD}
e_1^T
      D_{a_1} D_{a_{2}}\ldots D_{a_{m-1}}D_{a_m}
      e_7 
      ,
    \end{equation}
    which represents the upper right corner of the matrix product
\eqref{eq:D-B},
has the same value as
the expression in \eqref{eq:matrix-C-2} whenever we have a nonempty
matrix product.\footnote 
{Actually, it is the $7\times 7$ matrices
\eqref{eq:matrices-D}
  that 
Claus takes as the starting point \cite[Definition 5, $\phi(u,v)$,
p.~143]{claus81}.
In his proof, there is no analog of
Step  \ref*{thm-claus-9-states}.1 that would be separate from
Step  \ref*{thm-claus-9-states}.2.
He investigated the ``POGAMOR-problem'' of
checking whether the upper right corner of a product of matrices can
be made positive.}
For $m=0$,
\eqref{eq:matrix-DD} has value 0, and thus we have now excluded the
empty PCP solution.

\subsubsection{Step  \ref*{thm-claus-9-states}.3. Conversion to stochastic matrices}
\label{sec:claus-stochastic}

The conversion of the matrices to stochastic matrices
 will introduce two more states.

As a preparation, 
we 
extend the
matrices $D_i$ by two extra rows and columns to the matrices
\begin{displaymath}
  E_i =
  \begin{pmatrix}
 D_i & 0 &    t_i\\
 r_i^T & 0 & s_i\\
    0 & 0 & 0\\
  \end{pmatrix},
\end{displaymath}
where the vectors $r_i$ and $t_i$ and the scalar $s_i$ are chosen to
make all row and column sums zero.
It can be checked that the product $E_iE_j$ of two such matrices has the same form, its
row and column sums are zero,
and its upper left corner contains the matrix $D_iD_j$.
We extend the starting vector $e_1$ and the output vector $e_7$
by adding two zeros at the end.
Due to this setup, the extra rows and columns play no role for the result.
We have thus transformed \eqref{eq:matrix-C-2} into the equivalent
condition
\begin{equation}
  \label{eq:matrix-D+}
  e_1^T
  E_{a_1} E_{a_{2}}\ldots 
  E_{a_m}
  e _7 > 0,
\end{equation}
The matrices $E_i$ have dimension $d=9$.
Let $J$ denote the doubly-stochastic $d\times d$ matrix of
the ``completely random'' transition with all
entries equal to $1/d$.
Then, with a small constant $\alpha>0$,
we form the matrices $F_i := J + \alpha E_i$.
The constant $\alpha$ is chosen small enough such that $F_i>0$ for all
matrices~$F_i$.
Since the row sums are 1, the matrices $F_i$ are now stochastic.

The product of the new matrices $F_i$ is
\begin{displaymath}
      e_1^T
  F_{a_1} F_{a_{2}}\ldots 
  F_{a_m}
      e _7
=
      e_1^T
J
e _7
+ \alpha^{m}  
      e_1^T
  E_{a_1} E_{a_{2}}\ldots 
  E_{a_m}
      e _7.      
\end{displaymath}
The reason is that $E_iJ = JE_i=0$, and hence all ``mixed'' terms in
the expansion of the product $\prod F_{a_i}=\prod (J + \alpha E_{a_i})$ vanish.
The ``pure'' product $J^{m}$ simplifies due to the relation $JJ=J$.
Since
$e_1^T J e _7 =
1/d =  1/9
$, \eqref{eq:matrix-D+} becomes equivalent to
    \begin{displaymath}
      e_1^T
  F_{a_1} F_{a_{2}}\ldots 
  F_{a_m}
      e_7 > \tfrac19.
    \end{displaymath}
    The matrices $F_i$ are positive by construction. This proves the result.
\end{proof}



\subsection{Binary alphabet, 
  Theorem~\ref{thm-variable-matrices-binary}}
\label{sec:hirv}

For reducing the size of the input alphabet (or the number of
matrices) to the absolute minimum, namely two, we modify some steps of the
previous proof and  combine them with other steps.

\HirvensaloAdaptedSmallMatrices*
\begin{proof} 
We closely follow the proof of
Hirvensalo \cite{hirvensalo07}.

\subsubsection{Step \ref*{thm-variable-matrices-binary}.1. Modeling the PCP by integer matrices} 

We use the following variation of the matrix $ A$
from \eqref{eq:mat-Claus-3-original}. (It has been reflected at the
SW--NE diagonal.)
\begin{displaymath}
\tilde  A(v,w)=
  \begin{pmatrix}
3^{\strut|v| + |w|} & 0 & 0 & -(v)_{3}\cdot 3^{|w|} & (w)_{3}\cdot 3^{|v|} & -2 (v)_{3} (w)_{3} \\
0 & 3^{2 |w|} & 0 & (w)_{3}\cdot 3^{|w|} & 0 & ((w)_{3})^{2} \\
0 & 0 & 3^{2 |v|} & 0 & -(v)_{3}\cdot 3^{|v|} & ((v)_{3})^{2} \\
0 & 0 & 0 & 3^{|w|} & 0 & 2 (w)_{3} \\
0 & 0 & 0 & 0 & 3^{|v|} & -2 (v)_{3} \\
0 & 0 & 0 & 0 & 0 & 1 \\
  \end{pmatrix}
\end{displaymath}
It also satisfies the
multiplicative law, but
with reversed order of factors:
\begin{displaymath}
\tilde A (v_1,w_1)
\tilde A (v_2,w_2) =
\tilde A (v_2v_1,w_2w_1)
\end{displaymath}
Thus it works in the same way as the matrix $ A$.\footnote
{\label{differ-blondel}%
  Our definition of the matrix $\tilde A$ differs from the corresponding matrix $A(u,v)$ in
  \cite[p.~235]{blondel-canterini-03}) 
and the matrix $\gamma(u,v)$ in 
\cite[Eq.~(4)]{hirvensalo07}
in some nonessential details.
 Most importantly, the matrix is transposed.
 In addition, the first and second rows,
 and the first and second columns are swapped.
The matrices in \cite{blondel-canterini-03}
and 
\cite{hirvensalo07} have no negative signs.
Blondel and Canterini
\cite{blondel-canterini-03} use radix 10 instead of radix~3,\footnotemark\
and Hirvensalo
\cite{hirvensalo07} uses radix~2.$^{\ref{arbitrary-base}}$ 
We have chosen the form of $\tilde A$ to remain consistent with
 \cite{claus81} and \eqref{eq:mat-Claus-3-original}.

In the matrix formulation of the acceptance probability of a
PFA in Hirvensalo \cite[Eq.~(1), p.~310]{hirvensalo07},
the starting
distribution $\pi$
(or $\mathbf{y}$, in the notation of \cite{hirvensalo07})
is written on the right end of the matrix product and the
output vector $f$ (or $\mathbf{x}$) on the left, and the transition matrices are
 column-stochastic.
 Thus, despite the fact that the matrices are transposed, our
 interpretation as a PFA is the same as Hirvensalo's.
}
\footnotetext{ ``The notation is quaternary, decimal, etc., according
  to taste.'' (Paterson~\cite{paterson70})}%
The important
feature is that the last row is a unit row, and thus it acts like an
absorbing state. This will allow us to save two states when reducing the
number of matrices to two.%
%
{\setbox0=\hbox{\kern 24cm 
  \footnote{\label{arbitrary-base}The reader may worry that the binary value function
  $(u)_2$ might not be injective for strings $u$ with ternary digits
  $1,2$.  In fact, the function $(u)_2$ is even a \emph{bijection}
  between the strings $u\in\{1,2\}^*$ and the nonnegative
  integers. This can be extended to any radix.}}%
\dimen0=\wd0 \unhbox0\kern-\dimen0 
}

We can test equality of strings $v=w$ in terms of the matrix  $\tilde A$ with
a starting vector
$\pi_1=(-2,-2,-2,0,0,1)^T$
and an output vector
$\eeta_1=(0,0,0,0,0,1)^T$:
\begin{displaymath}
 \pi_1^T\tilde A (v,w)\eeta_1 = 1-2\bigl((v)_{3}-(w)_{3}\bigr)^2 
\end{displaymath}
Thus, since all data are integral, $\pi_1\tilde A (v,w)\eeta_1>0$ iff
$v=w$.
The value 0 cannot appear, and thus the construction works
equally with
the weak inequality
$\pi_1\tilde A (v,w)\eeta_1\ge 0$.\footnote
{Hirvensalo uses the starting vector $(-1,-1,-1,0,0,1)$ (translated to
  our formulation), and this works only for acceptance with strict inequality.
He treats the weak acceptance variant (with $\ge \lambda$)
by adjusting the matrix later,
which requires an additional state
\cite[Modification Step 3.5]{hirvensalo07}.
Alternatively, he could have used the starting vector $(-1,-1,-1,0,0,0)$.
}

For a PCP with word pairs $(v_i,w_i)$, we thus form the matrices
$B_i := \tilde A (v_i,w_i)$, and then the PCP solutions  $a_1a_2\ldots a_m$
(including the empty solution)
are characterized by the inequality
\begin{displaymath}
\pi_1^T B_{a_m}B_{a_{m-1}}\ldots B_{a_2}B_{a_1} \eeta_1 > 0.
\end{displaymath}

\subsubsection{Step \ref*{thm-variable-matrices-binary}.2. Merging the first and last matrices into the
  boundary vectors}
\label{sec:merge-canterini}

Hirvensalo
\cite{hirvensalo07}
based his proof on the undecidable PCP instances with
 7 word pairs of
 Matiyasevich and S{\'e}nizergues~\cite{matiyasevich2005decision},
 which had the smallest number of pairs until 2015.
These instances are actually instances of what we have called
the Doubly-Modified Post Correspondence Problem (2MPCP,
 Section~\ref{sec:vary-f}).
They have
two special pairs
$(v_1,w_1)$ and $(v_2,w_2)$, which must be used at the beginning
and at the end of the solution, respectively, and can be used nowhere else.
Thus, as in \eqref{eq:2MPCP}, the two corresponding matrices $B_1$ and $B_2$
can be multiplied with $\pi_1$ and $\eeta_1$, respectively,
yielding new matrices
$\pi_2^T =  \pi_1^T B_2$
and $f_2 := B_1  f_1 $
This reduces the number of matrices, and thus the size of the
alphabet, from 7 to 5.
By this change, we have also eliminated the unwanted empty solution ($m = 0$).

We can reduce the number of matrices by substituting the undecidable
PCP instance of Neary~\cite{neary:PCP5:2015} with only 5 pairs.  In
these instances, 
there is also a starting pair $(v_1,w_1)$ and an ending pair
$(v_2,w_2)$, such that every solution necessarily starts with 
$(v_1,w_1)$ and ends with $(v_2,w_2)$.  The ending pair 
can appear nowhere else.  However, unlike in the 2MPCP, the starting
pair
is also used in the middle of the solutions.
(This multipurpose usage of the
starting pair is one of the devices to achieve such a  remarkably small number of
word pairs.\footnote
{The proof of Theorem~11
in Neary~\cite{neary:PCP5:2015}
contains a mistake which needs to be fixed:
Neary's PCP instances encode binary tag systems.
 When showing that the PCP solution must
 follow the intended patters of the simulation of the binary tag system,
 Neary \cite[p.~660]{neary:PCP5:2015}
needs to show that the ending pair $(v_5,w_5)=
(10^\beta1111,1111)$ cannot be used except to bring the two strings to
a common end. He
 claims that
 a block $1111$ cannot appear in the encoded string because in $u$
(the unencoded string of the binary tag system, which is described in Lemma~9)
 we
cannot have two
$c$ symbols
next to each other.
This is not true. The paper contains plenty of examples,
and they
contradict this claim; for example, the string $u'$ in (7)
\cite[p.~657]{neary:PCP5:2015} contains seven $c$'s in a row.
The mistake can be fixed by taking a longer block of 1s:
Looking at the appendants in Lemma~9, it is clear that every block of
length $|u|+1$ must contain a $b$. Thus the word pair
$(v_5,w_5)=
(10^\beta1^{|u|+99},1^{|u|+99})$ will work.
})
Thus, we can 
merge
the boundary matrices
 into $\pi_1$
 and $\eeta_1$ as above.
But this reduces the number of matrices only from 5 to~4. 

Thus, we are left with four $6\times 6$ integer matrices, which we call
$C_1,C_2,C_3,C_4$.

\subsubsection{Step \ref*{thm-variable-matrices-binary}.3. Reduction to  two matrices 
}
\label{sec:Hirvensalo-to-binary}

We have in integer-weighted automaton with 6 states and inputs
$a_{m-1}a_{m-2}\ldots a_2
\in \{1,2,3,4\}^*$ from an alphabet of 4
symbols.
We reduce the input alphabet to a binary alphabet along the lines
of Lemma~\ref{coding}, using the codewords
\texttt{b}, \texttt{ab}, \texttt{aab}, \texttt{aaa}.

Applying Lemma~\ref{coding} directly would multiply the number of states
by~3.
However, we observe that state 6 behaves like an absorbing state,
because the last line of all matrices is the unit vector that leads
back to state 6.
Thus, when the automaton is in state 6, we can stop decoding the
input,
and we need not split state 6 into three states.

Formally, we split the $6\times 6$ matrices $C_i$ and the vector
$\pi_2$ and $\eeta_2$ into
blocks of size $5+1$:
\begin{displaymath}
  C_i =
  \begin{pmatrix}
    \hat C_i &c_i\\
    0 & 1
  \end{pmatrix}
, 
  \pi_2 =
  \begin{pmatrix}
    \hat \pi_2\\
p_2
  \end{pmatrix}
,
  \eeta_2 =
  \begin{pmatrix}
    \hat \eeta_2\\
g_2
  \end{pmatrix}
.
\end{displaymath}
Following the construction in the proof of Lemma~\ref{coding},
the new transition matrices and starting and ending vectors
are written in block form
with block sizes $5+5+5+1=16$:
\begin{displaymath}
  M'_{\mathtt{a}}=
  \begin{pmatrix}
  0 & I & 0 & 0\\
  0 & 0 & I & 0 \\
  \hat C_4  &  0 & 0 &c_4\\
  0&0&0&1  
\end{pmatrix}
,\ 
M'_\mathtt{b}=
  \begin{pmatrix}
    \hat C_1 & 0 & 0 & c_1\\
    \hat C_2 & 0 & 0 & c_2\\
    \hat C_3 & 0 & 0 & c_3\\
    0 & 0 & 0 & 1
  \end{pmatrix}
,\ 
\pi_3 =
\begin{pmatrix}
 \hat \pi_2\\0\\0\\p_2
\end{pmatrix} 
\text{, and }
f_3 =
\begin{pmatrix}
\hat f_2\\\hat f_2\\\hat f_2\\g_2
\end{pmatrix}
.
\end{displaymath}
With the intermediate product
\begin{displaymath}
  M'_{\mathtt{a}}
M'_{\mathtt{a}}
=
  \begin{pmatrix}
  0 & 0 & I & 0 \\
  \hat C_4  &  0 & 0 &c_4\\
 0&  \hat C_4  & 0 &c_4\\
  0&0&0&1  
\end{pmatrix},
\end{displaymath}
one can work out the matrices
$M'_{\mathtt{b}}$,
$M'_{\mathtt{a}}
M'_{\mathtt{b}}$,
$M'_{\mathtt{a}}
M'_{\mathtt{a}}
M'_{\mathtt{b}}$, and
$M'_{\mathtt{a}}
M'_{\mathtt{a}}
M'_{\mathtt{a}}$
and check that they
are of the form
\begin{displaymath}
  \begin{pmatrix}
    \hat C_i & 0 & 0 & c_i\\
    * & * & * & *\\
    * & * & * & *\\
    0 & 0 & 0 & 1    
  \end{pmatrix}
\end{displaymath}
for $i=1,2,3,4$, thus simulating the original automaton on the states
1--5 and 16:
It is easy to establish by induction that
multiplying the initial distribution $\pi_3^T$ with a sequence of such matrices produces a
vector $x^T$ of the form
\begin{equation}
  \label{eq:result-form-Hirv}
x^T = (\hat x^T \ 0\ 0\  x_{6})
\end{equation}
whose nonzero entries
$\underline x^T=(\hat x^T\  x_{6})$ 
coincide with the entries of
the corresponding vector produced from $\pi_2^T$ with the original matrices~$C_i$.
If the scalar product with $f_3$ is taken, the result is the same as with
 the original vectors $\underline x^T$ and~$f_2$.

One technicality remains to be discussed: Some  ``unfinished'' words in
$\{\mathtt{a},\mathtt{b}\}^*$  do not factor into codewords but
end in a partial
codeword \texttt{a} or \texttt{aa}.
To analyze the corresponding matrix products, we also look at the products
$M'_{\mathtt{a}}$, and
$M'_{\mathtt{a}}M'_{\mathtt{a}}$. They have the form
\begin{displaymath} 
  \begin{pmatrix}
    0 & I & 0 & 0\\
    * & * & * & *\\
    * & * & * & *\\ 
    0 & 0 & 0 & 1    
  \end{pmatrix}
\text{ and  }
  \begin{pmatrix}
    0 & 0 &I & 0\\
    * & * & * & *\\
    * & * & * & *\\ 
    0 & 0 & 0 & 1    
  \end{pmatrix},
\end{displaymath}
and therefore, multiplying them
with the vector $x^T$ of \eqref{eq:result-form-Hirv}
yields
$ (0 \ \hat x^T\, 0\  x_{6})$ and
$ (0 \ 0\ \hat x^T\,  x_{6})$, respectively.
If this is multiplied with $f_3$, the result is still the same as with
 the vector $x^T$ of \eqref{eq:result-form-Hirv}. 
Thus, \emph{input sequences whose decoding process leaves
a partial codeword $\mathtt{a}$ or $\mathtt{aa}$ at the end produce the same value
as if that partial codeword
were omitted}.
This is the desired behavior in the context of the emptiness
question.\footnote
{\label{Hirvensalo-mistake}
  We are thus committing the (inconsequential) ``error''
  which took
  Nasu and Honda
\cite{NasuHonda1969}
and
Paz \cite{paz71} some effort to fix, see
Section~\ref{sec:reflection-binary}.
  Hirvensalo
  \cite
  {hirvensalo07} erroneously claimed that such
  incomplete inputs 
  give value~0.
He mistakenly defined the vector $f_3$ ($\mathbf{y}_3$ in his notation)
``analogously
as the vector $\pi_3$ ($\mathbf{x}_3$ in his notation). With this
definition, the result for incomplete inputs cannot be controlled at all.
}
We can describe the issue in terms of the algorithm of
Lemma~\ref{coding}:
The acceptance decision of this algorithm takes into account the value
of the
counter $i$: It accepts only when $i=0$. However, in state $6$, our algorithm ``loses track'' of~$i$.
Thus we have to ignore the counter also in the other states. This is
reflected in the vector $f_3$ by copying $\hat f_2$ for every value
of the counter ($i=0,1,2$).

\subsubsection{Step \ref*{thm-variable-matrices-binary}.4. Introducing
  unit vectors as starting and ending vectors}
\label{sec:hirv-unit-vector}
We have now
two $16\times 16$ matrices
$M'_{\mathtt{a}}$, and
$M'_{\mathtt{b}}$
and the
corresponding starting and ending vectors $\pi_3$ and~$\eeta_3$.

This step and the next one follow the standard proof of Turakainen's
Theorem~\cite{turakainen69,Turakainen_1975}, and they have already been described in Section~\ref{sec:fewest-possible}.

First we 
introduce a new
start state (state 1)
and a new final state (state~18), extending the matrices to size
$18\times 18$.
This is similar to Step~\ref*{thm-claus-9-states}.2
(Section~\ref{sec:step2-claus}),
where we introduced only a final state, because the starting vector
was already a unit vector.

We set $\pi_4=e_1$, $f_4 = e_{18}$, and
\begin{displaymath}
  D_1 =
  \begin{pmatrix}
    0 & \pi_3^T M'_{\mathtt{a}}  & \pi_3^T M'_{\mathtt{a}} f_3\\
0 & M'_{\mathtt{a}}  & M'_{\mathtt{a}}f_3 \\
    0 & 0 & 0\\
\end{pmatrix}
,\
  D_2 =
  \begin{pmatrix}
    0 & \pi_3^T M'_{\mathtt{b}}  & \pi_3^T M'_{\mathtt{b}} f_3\\
0 & M'_{\mathtt{b}}  & M'_{\mathtt{b}}f_3 \\
    0 & 0 & 0\\
\end{pmatrix}.
\end{displaymath}
One can check that this leads to the same results.

\subsubsection{Step \ref*{thm-variable-matrices-binary}.5. Conversion to stochastic matrices}
This conversion has been described above in
Step \ref*{thm-claus-9-states}.3
(Section~\ref{sec:claus-stochastic}). It introduces two more states,
leading to a total of 20 states.
This concludes the proof of Theorem~\ref{thm-variable-matrices-binary}.
\end{proof}


\subsection{History of ideas}

Mike Paterson~\cite{paterson70} pioneered the modeling of the PCP by multiplication of
 integer matrices in 1970 in order to show that mortality for $3\times 3$
 matrices is undecidable. This problem asks whether the zero matrix
 can be obtained as a product of matrices from a given set.
 Paterson 
 introduced the matrix
 \begin{equation}
   \label{eq:Paterson}
   \begin{pmatrix}
   10^{|v|} &0&0\\
   0&  10^{|w|}&0\\
   (v)_{10}
&   (w)_{10} & 1   
   \end{pmatrix},
 \end{equation}
to represent a pair $(v,w)$ of a PCP.\footnote
{Paterson 
 specified these matrices informally by way of example, without
 committing to a particular base: ``The notation is quaternary, decimal, etc., according
 to taste.''}
These matrices satisfy a multiplicative law like~\eqref{eq:multi-gamma-forward}.

It was pointed out both by
Claus \cite[p.~153]{claus81}
and by
Blondel and Canterini \cite[p.~235]{blondel-canterini-03} that one
can generate the quadratic terms that are necessary to form
the expression $1-\bigl((v)_{10}-(w)_{10}\bigr)^2$
by taking the Kronecker product of two $3\times 3$ matrices
of the form~\eqref{eq:Paterson},
which would be a $9\times 9$ matrix.
The
matrix
\eqref{eq:mat-Claus-3-original}
of smaller size  $6\times 6$, which is sufficient for this task, was
found by Claus
 \cite[p.~153]{claus81} in 1981, and in slightly different form
 by
Blondel and Canterini \cite{blondel-canterini-03} in 2003, see footnote~\ref{differ-blondel}.

As repeatedly mentioned, the conversion from general integer
matrices and vectors $\pi$ and $f$ to the PFA setting
is due to Turakainen~\cite{turakainen69} from 1969.
 The techniques were later sharpened by
 Turakainen~\cite{Turakainen_1975} in 1975.

We mention
that
  Theorem~\ref{thm-claus-9-states} can be slightly
  strengthened by using techniques from the proof of
  Theorem~\ref{thm-variable-matrices-binary}.
Neary's PCP is an RMPCP with an ending pair that that is used nowhere else.
Thus, we can merge the rightmost matrix $B_m$ with the vector
$f_1$, analogous to
Step \ref*{thm-variable-matrices-binary}.2
(Section~\ref{sec:merge-canterini}). If this is done right at the beginning,
\emph{before}
introducing
a new final state
(Step~\ref*{thm-claus-9-states}.2, Section~\ref{sec:step2-claus}), it
does not affect the remainder of the proof.
This reduces the number of matrices in
  Theorem~\ref{thm-claus-9-states} from 5 to 4.
At the same time,
this eliminates the empty PCP solution.
In addition, with the
 alternative final vector $f_1^T=(0,0,0,-1,-1,-1)$
instead of $(1,0,0,-1,-1,-1)$ in \eqref{eq:check-u=v},
the theorem
extends to weak inequality $(\ge1/9)$ as the  acceptance criterion.



\section{Alternative universal Turing machines}
\label{sec:alternative}

Our proofs rely on a particular small universal Turing machine.
In the literature, some 
``universal'' Turing machines with smaller numbers of states and symbols are
proposed. We review these machines 
and discuss
whether they could possibly be used to decrease the number of
matrices in
Theorems~\ref{thm-fixed-f}--\ref{thm-fixed-pi}.


\subsection{Watanabe, weak and semi-weak universality}
\label{sec:watanabe-weakly}

A universal Turing machine $U_W$
with 3 symbols and 7 states was
published by
Shigeru Watanabe~\cite{watanabe72} in 1972,
but I haven't been able to get hold of this paper.
According to the survey
\cite[Fig.~1]{woods-survey-2009},
this is
a \emph{semi-weakly} universal Turing machine.
In \emph{semi-weakly} and
 \emph{weakly} universal machines,
 the empty parts of the tape
on one or both sides of the input
 are initially filled
with
some 
repeating pattern
instead of uniformly
blank symbols.
Such a
repeating pattern can be
easily accommodated in
the translation to the MPCP
by modifying the padding pair of words
$(\texttt{\#},\texttt{\blank \#\blank })$.

In the worst case,
the 21 rules contain only one halting
rule and the remaining 20 rules are balanced between left- and
right-moving rules.
Then, with
10 left-moving rules, 1 halting rule,
10 right-moving rules, and $|\Gamma|=3$, we get
$3\times 3+3  + 10\times3+ 11=53$ matrices, the same number
as from the machine
$U_{15,2}$ of
Neary and Woods.
Any imbalance in the distribution
of left-moving and right-moving rules
would allow to reduce the number of matrices
in Theorem~\ref{thm-fixed-f}b
 from
53 to 51 or less.

This speculative 
improvement depends
on an 
assumption,
which would need to be verified. 
According to \cite[Section 3.1]{woods-survey-2009},
Watanabe's weak machine $U_W$ simulates other
Turing machines $
T$ directly.
What would be most useful for us is that
the periodic pattern that initially fills the tape of $U_W$ is
a fixed pattern that is specified as part of the definition of~$U_W$
and does not depend on $ T$ or its input.

If this is the case, we can use them for
our construction, 
where only the first (or
last) pair of the PCP should depend on the input.

We could even accommodate some weaker 
requirement, namely that the 
periodic pattern depends on the Turing machine $ T$ that is being
simulated, as long as it is independent of the input $u$ to that Turing machine.
 We could then let $U_W$ simulate a fixed universal
Turing machine $T=U_0$ (universal in the usual, standard, sense),
and then the periodic pattern would also be fixed. 

\subsection{Wolfram--Cook, rule 110}
\label{tag-systems}
Some small machines
are based on simulating a particular cellular
automaton, the so-called \emph{rule-110 automaton} of Stephen Wolfram.
These machines are given in \cite[Fig.~1, p.~3]{cook2004}
and \cite{small-weakly-universal-2009},
see also the survey~\cite{woods-survey-2009}.
The machines of \cite{small-weakly-universal-2009} have as few as 6
states and 2 symbols, or 3 states and 3 symbols, or 2 states and 4 symbols.
The {rule-110 automaton} was shown to be universal by
Matthew Cook~\cite{cook2004},
see also Wolfram~\cite[
  Section 11.8, pp.~675--689]
{wolfram-NKS}\rlap.\footnote
{on-line at
  \url{https://www.wolframscience.com/nks/p675--the-rule-110-cellular-automaton/}}%
The universality of the rule-110 automaton comes from the fact that
rule 110 can simulate 
\emph{cyclic tag systems}.
\emph
{Tag systems} are a special type of string rewriting systems, where
symbols are deleted from the front of a string, and other symbols are
appended to the end of a string, according to certain rules.
\emph{Cyclic tag systems}
 are a particularly simple variation of tag
systems. Tag systems as well as cyclic tag systems are known to be
universal,
because they can simulate Turing machines.


One difficulty
with these small Turing machines, which makes them not directly suitable
for our
purposes is that,
like in the weakly universal machines of
Section~\ref{sec:watanabe-weakly},
the repeating patterns by which the ends 
of the 
tape are filled are not fixed, but depend on the tag system,
see
\cite[Note on initial conditions, p.~1116]
{wolfram-NKS}\rlap.\footnote
{
  \url{https://www.wolframscience.com/nks/notes-11-8--initial-conditions-for-rule-110/}}
%
As a consequence, we don't have a fixed replacement for
 the padding pair
 $(\texttt{\#},\texttt{\blank \#\blank })$.
 Thus we cannot 
 use them for the proof of
Theorems~\ref{thm-fixed-f} and \ref{thm-fixed-pi},
where only the first (or
last) pair of the PCP should depend on the input.

As in Section~\ref{sec:watanabe-weakly}, one could start with a universal Turing
machine $U_0$ and
construct from it a fixed cyclic tag system.
The classic way to do this is the  method of Cocke and Minsky from
1964~\cite{cocke-minsky-tag-system-1964}.
It converts a Turing machine $U$ to a tag system $\mathcal{C}$. The transition
rules of $U$ are converted to the rules (appendants) of $\mathcal{C}$, and
the input tape of $U$ is translated to the starting
string of the tag system $\mathcal{C}$.
If we start with a universal Turing machine $U_0$, the rules of
$\mathcal{C}$ are fixed.
 The simulation of
the tag system $\mathcal{C}$
by a cyclic tag system $\mathcal{C}'$ is  easy~\cite{cook2004}\rlap.\footnote{An alternative simulation
of Neary and Woods~\cite{neary-woods-rule110-2006} simulates a
Turing machine directly by a cyclic tag system $\mathcal{C}'$, and is 
also much more time-efficient.}
Eventually, in the Turing machine $T$ that
simulates the rule-111 automaton via the cyclic tag system  $\mathcal{C}'$,
this translates
to a fixed pattern by which the empty tape of $T$ is filled on the two sides
of the input.
Therefore,
 the padding pair
 $(\texttt{\#},\texttt{\blank \#\blank })$ in the PCP has a fixed replacement,
 which is translated into a fixed stochastic matrix for the PFA.

 There is, however, another reason why the 
machines
simulating the rule-110 automaton cannot be used
directly: 
They have no provision for \emph{halting},
or for otherwise determining some set of inputs that they 
{accept}.%
\footnote{
Curiously, 
while the survey of Woods and Neary
\cite{woods-survey-2009} carefully distinguishes 
 \emph{semi-weak} and
 \emph{weak} universality,
the fundamental characteristic 
whether the 
Turing machine has a provision for halting 
is
treated only as an afterthought.%
\footnotemark}%
\footnotetext{
Incidentally,
in Turing's original article 
\cite {turing-1936} from 1937,
where the machines that are now called Turing machines were first defined, 
the good 
machines are those that
\emph{don't} halt or go into a loop 
(the \emph{circle-free} ones),
because
they are capable of producing an infinite sequence of zeros and ones
on dedicated \emph{output cells} on the tape, forming the fractional bits
of a \emph{computable number}.
The question about his machines that Turing proved to be undecidable
is:
Does this machine ever print a \texttt{0}?
}
This is natural in the context
of a 
cellular automaton,
which performs an infinite process.
However,
a 
{cyclic tag system}, which the automaton supposedly simulates,
\emph{does} have a way of terminating, 
namely when the string on which
it operates becomes empty.
Fortunately, 
Cook gives
a few hints about termination
and about undecidable questions for the corresponding Turing machines:
\emph{Questions about their behavior,
such as ``Will this sequence of symbols ever appear on the tape?'', are
undecidable}
\cite[{Note [7]}, p.~38]{cook2004}.
More specifically,
Cook
mentions some particular undecidable questions for so-called \emph{glider systems}.
Some consequences of 
this discussion
 for Rule 110 are briefly touched upon
 in
\cite[Section 4.6, p.~37]{cook2004}:
\emph{So another specific example of an undecidable question for Rule 110
is: Given an initial middle segment, will there ever be an $F$?}
Here, an $F$ is a particular type of ``glider'', a cyclically repeated
sequence of
patterns
that moves at constant speed through the cellular automaton as long
as it does not hit other gliders or irregularities.
Hypothesizing that 
the presence of such a glider 
could be detected by the occurrence of a particular
pattern $\hat F$ in the cellular automaton, or on the Turing machine tape,
such a criterion could be translated into a word pair $(\hat F,H)$ that
introduces the halting symbol~$H$, and this would lead to small
undecidable instances of the MPCP.

All these arguments
require careful examination, and the approach depends on a cascade of reductions,
so we have not pursued it.

\subsection{Wolfram's 2,3 Turing machine}

An even smaller Turing machine with only 2 symbols and 3 states was
proposed by
 Wolfram~\cite[
 Section 11.12,
 p.~709]
{wolfram-NKS}\footnote
{on-line at
\url{https://www.wolframscience.com/nks/p709--universality-in-turing-machines-and-other-systems/}}
and was shown to be universal, in a certain sense,
by
Alex Smith~\cite{smith20}\rlap.\footnote
{The reader should be warned
that the journal version \cite{smith20}
  is partly incomprehensible, since the proper horizontal alignment in the
  tabular presentation of the Turing machines has been destroyed in
  the typesetting process.  
  Understandable 
  versions can be found elsewhere on the web, see for example
  \url{https://www.wolframscience.com/prizes/tm23/TM23Proof.pdf}.
Apart from this, the paper could definitely have benefited from some reviewing and editorial guidance.  
}
As above in Section~\ref{tag-systems}, 
the proof performs 
a reduction from
cyclic tag systems, 
and again, this machine does not halt.
Smith showed that
the
2,3 Turing machine
 can simulate
 cyclic tag systems, but
unfortunately, it is
not addressed at all what happens when
the operation of the
simulated cyclic tag system terminates.
Besides the issue of halting, there is a more severe obstacle
for using this machine for our purpose:
The input is not some finite word bounded by repeated patterns on
both sides, but an infinite string that has to be set up
by an independent process.


\section{Outlook}
\label{sec:outlook}

\subsection{Equality testing}
\label{equality-testing}
For the reader who has well digested the basic ideas 
in the two different
proof approaches of PFA Emptiness undecidability,
it is an instructive 
exercise to see how Nasu and Honda's 
method of testing acceptance probabilities for equality by
formula~\eqref{eq:equality-trick}
(Section~\ref{check:equality})
would
apply to the Equality Checker problem
of
Section~\ref{sec:eq}
for the string
$\texttt{a}^i\texttt{b}^j\texttt{\#}$:
It is straightforward to set up a PFA with two states that accepts
$\texttt{a}^i\texttt{b}^j\texttt{\#}$ with probability $\phi=1/2^i$,
and another that accepts it with probability $\psi=1/2^j$.

The construction of
Figure~\ref{fig:mix}a, translated into the language of
the Equality Checker, leads to the following algorithm:
The coins are flipped as usual. In the end, when reading the symbol
\texttt{\#}, the PFA flips two more coins, and
\begin{compactitem}
\item   with probability 1/4, it accepts if the red coin was unlucky;
\item 
  with probability 1/4, it accepts if the orange coin was unlucky
\item 
  with probability 1/2, it accepts if the blue coin was lucky.
\end{compactitem}
The green coin is ignored.
The resulting probability of acceptance is
\begin{equation}
  \label{eq:equality-by-coins}
\textstyle \frac14(1-1/4^i)
+\frac14(1-1/4^j)
+\frac12\cdot 1/2^{i+j} = \frac12 -\frac14 (1/2^i-1/2^j)^2,
\end{equation}
which reaches its maximum value $1/2$ if and only if $i=j$.
(To save coin flips, one would of course rather take the decision between the
three branches in advance.)

We notice a sharp contrast between the  character of
the outcome 
in the two cases.
The equality test
 by
 formulas~\eqref{eq:equality-trick}
and~\eqref {eq:equality-by-coins}
capitalizes on 
the capability of a PFA to detect
 a tiny fluctuation of the acceptance probability
 above 
 the cutpoint.
On the other hand,
the Equality Checker,
as illustrated in
Figure~\ref{fig:EC},
almost always leaves the answer 
``Undecided'', but if it makes a decision, the probabilities of the
two outcomes, 
in case of inequality, 
differ by several orders of magnitude.
%

%

We remark that the fluctuation above $\lambda$ can be amplified by
a technique of
Gimbert and Oualhadj
\cite{gimbert-oualhadj-2010:PFA} from 2010,
which was considerably simplified by
Fijalkow
\cite{Fijalkow-2017-SIGLOG}:
The Emptiness Problem with strict threshold $>\frac12$
can be reduced
to
the dichotomy
``there are acceptance probabilities arbitrarily close to~1''
versus
 ``all acceptance probabilities are $\le 1/2$''.
See Section~\ref{sec:amplification-study}.

\subsection{Shortcutting the reduction?}
\label{sec:shortcut}
Figure~\ref{fig:reductions} illustrates the chain of reductions
leading to the three undecidability proofs of PFA Emptiness.
We have not drawn the detour from the PCP via integer matrices to PFA Emptiness.
In all cases, undecidability ultimately
stems from the Halting Problem for Turing machines.

\begin{figure}[htb]
  \centering
\fboxsep=5mm
\fbox{\begin{tikzpicture}[
  boxnode/.style={rectangle, draw=black, minimum size=6mm},
  style={->,shorten >=0.3pt,>={Stealth[round]},semithick}
  ]
  \node[boxnode,
  align=center] (TM)
  {Turing machine};
   
  \node[boxnode,
  align=center] (tag)
  at (-3,-1.2)
  {tag system};

  \node[boxnode,
  align=center] (UTM)
  [below=8mm of tag]
  {universal Turing machine};
  
 \node[boxnode,
  align=center] (PCP)
  at (-2,-4.1)
  {Post's Correspondence Problem (PCP)};
  
 \node[boxnode,
  align=center] (PFA)
  [below=50mm of TM]
  {Probabilistic Finite Automaton (PFA) Emptiness};
  
 \node[boxnode,
  align=center] (mp)
  [below=8mm of PFA]
  {various problems about matrix products};
  
 \node[boxnode,
  align=center] (2CM)
  at (3.34, -3.34) 
  {2-counter machine};
\node[right=3mm of 2CM]{{\ }};

\draw[->] (TM) -- (tag);
\draw[->] (TM.330) -- (2CM);
\draw[->,dotted] (TM.230) -- (UTM.25);
\draw[->,dashed] (TM.300) -- (PCP.13);
\draw[->] (tag) -- (UTM);
\draw[->] (UTM.290) -- (PCP.140);
\draw[->] (PCP.300) -- (PFA);
\draw[->] (PFA) -- (mp);
\draw[->] (2CM) -- (PFA);
  
\end{tikzpicture}
}
\caption{Reductions for proving undecidability.
   The four topmost boxes concern the \emph{Halting Problem} for the
   respective systems.
   The dashed arrow
   represents the reduction that is sufficient for
   the plain undecidability result of PFA Emptiness
(Propositions~\ref{ge1/2}--\ref{gt1/4}), without the specializations
   of Theorems~\ref{thm-fixed-f}--\ref{thm-fixed-pi}.
 }
  \label{fig:reductions}
\end{figure}
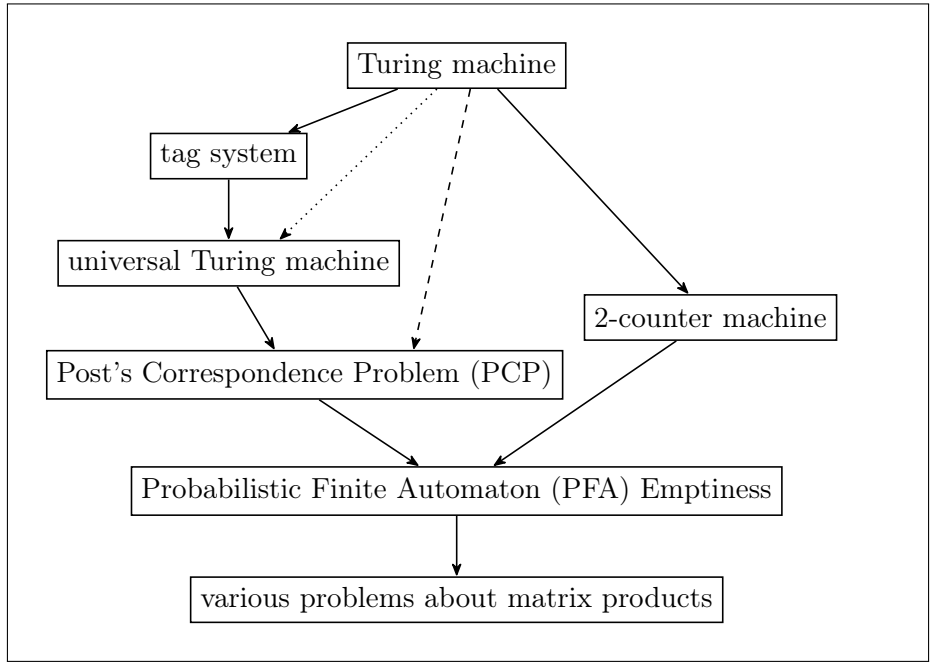

The earliest universal Turing machines simulate general
Turing machines directly, as indicated by the dotted arrow.
However,
the smallest universal Turing machines known today do not simulate
Turing machines, but  
{tag systems} (Section~\ref{tag-systems}).
In particular, this is true for the machine
$U_{15,2}$ of
Neary and Woods, on which
 Theorems~\ref{thm-fixed-f}--\ref{thm-fixed-pi} are based.
This has the
somewhat curious effect that our construction of
specialized undecidable instances of PFA Emptiness 
proceeds by
reduction from
tag systems, which operate 
on \emph{strings}, via universal Turing \emph{machines},
to another problem
on \emph{strings}:
the PCP. 
It would seem natural to shortcut this detour and try to go
from
tag systems
to
the PCP
directly.
Tag systems are universal in the sense that every Turing machine can
be simulated by some tag system.
What might be useful for us is a tag system with a
(small) fixed set of rules for which the halting problem is
undecidable, depending on the starting string.
It seems that such tag systems have not been studied in their own right.
Of course, one can take a tag system that simulates a
universal Turing machine, an idea that was sketched 
in
 Sections~\ref{sec:watanabe-weakly}
and~\ref{tag-systems}.
This adds another round to the detour,
but it might be interesting to pursue this approach.

 
\subsection{The minimum number of states}
\label{sec:states}
It is a natural question to ask for the smallest number of states
for which the PFA Emptiness Problem is undecidable.
The number of states can be reduced
to~9 (Theorem~\ref{thm-claus-9-states}).
 When the transition matrices are fixed, the bounds are not much larger:
Depending on the precise technical formulation of the question,
9~states (Theorem~\ref{thm-fixed-pi}a)
or 11~states (Theorems~\ref{thm-fixed-f}b and~\ref{thm-fixed-pi}b) suffice.

Claus \cite[Theorem 7 and Corollary, p.~155]{claus81} showed that
the emptiness question can be decided for
PFAs with two states.

\section{Epilogue: How to present a proof
}
\label{sec:reflection}

Struggling through the literature and writing this article has prompted
me to reflect on different 
ways of presenting things. 
In this final section, which has become quite long, I want to discuss two issues:
(1)~Choosing the right level of abstraction.
(2)~Presenting material in a self-contained way versus relying
on powerful general results.

\subsection{Levels of abstraction}
We have initially defined PFAs by a high-level informal description
as an \emph{algorithm}, referring to finite-range variables and
using metaphors such as flipping of coins.
We have complemented this with a low-level, formal definition
in terms of transition \emph{matrices}.

An even larger range of abstraction levels exists for Turing machines.
We know that a Turing machine can implement any algorithm.
So we may just specify at a high level what the machine should do,
trusting that the reader has internalized the Church--Turing thesis.
At an intermediate level, we may describe how the
Turing machine
marks cells or
carries information back and forth on the tape in fulfilling its task.
The lowest level
is the
description in terms of the transition function, 
in ``assembly language'', so-to-speak.

Each level 
has its proper place.
For example, the description of the PFA for the Condon--Lipton proof
in Section~\ref{sec:2c} remains exclusively at an abstract level.
On the other hand, the binary PFA in Section~\ref{sec:binary} is best
described in terms of its transition matrix. There is nothing 
to say about its behavior besides that it performs its transitions as
specified by the matrix.

People might have different preferences in this matter.
If an informal and high-level description is complemented by
a more concrete ``implementation'' at a lower level, such
a presentation may offer something for every taste.
The low-level description may be useful to
confirm the understanding
or to dispel doubts.

Anyway, 
an experienced reader will be able to translate between the
levels. For example,
when a machine does several things simultaneously or keeps track of
several counters, 
this
corresponds to taking the product of
the state sets, and the Kronecker product of transition matrices
(cf.\ Section~\ref{sec:positive-Kronecker}).

The first two case studies contrast different levels of abstraction.

\subsection{Case study 1: Restricting the output vector
  \texorpdfstring{$f$}{f} to a  \texorpdfstring{$0$-$1$}{0-1}-vector}
\label{sec:output-0-1-epilogue}

In the proof of Theorem~\ref{thm-fixed-f}a in
Section~\ref{eleven}
(p.~\pageref{proof-b-f}), 
we convert a PFA
with arbitrary output values (acceptance probabilities) $f$ from the interval $[0,1]$ to
 an equivalent PFA with $0$-$1$-values, i.e., with a set of accepting states.

 Let us review this proof from the point of view of the abstraction
 levels used.
 After stating the idea,
 we give an informal description of the conversion process.
 Then we describe the process formally in terms of transition
 matrices and the starting distribution.
 Finally, we make an observation on the resulting matrices, and
 confirm the understanding by interpreting it in terms of the
 original idea.
A formal proof that the new PFA yields the same acceptance
probabilities is omitted.

The same statement is proved as
the last step of the first proof of Turakainen's more general theorem~\cite
[p.~308]{turakainen69} from 1969
that has been mentioned
in Section~\ref{sec:fix-start} (see footnote \ref{foot-square}).
In this proof,
the probabilities are first rescaled to ensure that $\sum f_q=1$.
Then the basic idea is the same as in our proof of Theorem~\ref{thm-fixed-f}a:
Concurrently with every step,
we generate a random variable~$X$.
In contrast to our proof,
this random variable has $d$ possible outcomes.
Each outcome corresponds to one of the states.
Moreover,
the distribution of $X$ is always the same, independently of the current state:
The $d$ outcomes of $X$ are chosen according to
the distribution $f$.
This variable is sufficient to make a decision whether
the input word should really be accepted: 
It is accepted if the value of $X$ matches the current state $q$, which
happens with probability~$f_q$.
Since there are $d$ copies of every state, the number of states
increases 
quadratically.

The idea, however, is not explained in~\cite{turakainen69}.
 Turakainen 
  writes down the
$d^2\times d^2$ transition
matrix, the corresponding starting distribution, and
the accepting set.
Correctness is proved formally by multiplying out the matrices and
showing that they lead to the same acceptance probabilities.

Our own proof in
Section~\ref{eleven} is not only less wasteful of states,
by generating a customized random variable $X$ depending on the state,
but it also gives explanations and some intuition of what
one tries to achieve.

Incidentally, the same statement is proved by Paz with a quite
different argument,
which is
very elegant.
It appears on
 \cite[p.~151]{paz71} as part of the proof of a more general statement, Theorem 2.4 of Section~IIIA.
The output vector $f\in[0,1]^d$ can
  be represented as a convex combination of 0-1-vectors
  (vertices of the hypercube $[0,1]^d$).
 Appealing to a general statement
about  convex combinations
\cite[Corollary 1.7 of Section IIIA, p.~148]{paz71})
concludes the proof.
The number of states of the construction is not discussed.
As stated, the construction may lead to an exponential blow-up of the number of states.
However, with appropriate extra considerations, for example,
by appealing to Carathéodory's Theorem, the
blow-up can be reduced to quadratic.

Turakainen's second proof \cite
{Turakainen_1975} from 1975 is
much more efficient, see Sections~\ref{sec:step2-claus} and
\ref{sec:hirv-unit-vector}: It uses just one additional state to get
to a
unit vector $f$, possibly losing the empty word from the recognized language.

\subsection{Case study 2: Probability amplification}
\label{sec:amplification-study}

The classical acceptance criterion is the
strict threshold $>\lambda$.
We will now see a general method by which
the
 Emptiness Problem
with strict acceptance criterion $>\frac12$ can be strengthened to
the dichotomy ``all acceptance probabilities are $\le 1/2$'' versus
``there are acceptance probabilities arbitrarily close to~1'':

\begin{theorem}
[Gimbert and Oualhadj 2010
\cite{gimbert-oualhadj-2010:PFA}]

  \label{thm-amplify}
We are given a PFA $A$ with input alphabet $\Sigma$, and we denote its
acceptance probability for an input $u\in \Sigma^*$ 
by $\phi(u)$.
We can construct a PFA $B$ with input alphabet
$\Sigma\cup\{\mathtt{end},\mathtt{check}\}$ with the following
property.
\begin{enumerate}
\item
    If every input $u$ for $A$ has acceptance probability $\phi(u)\le1/2$,
  then the same holds for $B$.
\item If $A$ has an input with acceptance probability $\phi(u)>1/2$,
  then $B$ has inputs with acceptance probability arbitrarily close to~$1$.
\end{enumerate}
\end{theorem}

\subsubsection{Algorithm F,  Fijalkow}
\label{sec:alg-F}

\begin{proof}[First proof of Theorem~\ref{thm-amplify}]

  We will first describe a simplified construction
  due to Nathanaël
Fijalkow
  \cite{Fijalkow-2017-SIGLOG} from  2017 and show
 the original proof later.

To avoid confusion with the acceptance of $B$, we call the output
of the original automaton $A$ ``Plus'' or ``Minus'', accordingly as it arrives
in an accepting state or in a nonaccepting state after reading the input.

The input for the PFA $B$
is structured into \emph{rounds}. Each round processes an input
of the form
\begin{displaymath}
  u_1\ \mathtt{end}\ u_2\ \mathtt{end}\  \ldots \ u_n\ \mathtt{end}\ \mathtt{check}
\end{displaymath}
with $u_i\in \Sigma ^*$ and is processed as follows.
Each $u_i$ is processed as in $A$, and when reading the symbol \texttt{end}, the
machine notes the output of $A$ (``Plus'' or ``Minus'').
The machine flips a coin for each round, and with probability 1/2, it does one of the
following two things   when the symbol
  \texttt{check} is read:
\begin{itemize}
\item [($+$)]
  If the output of $A$ was ``Plus'' for \emph{all} inputs
  $u_1,\ldots,u_n$ of this round,
  the machine $B$ \emph{accepts} the input once and
  for all 
  and ignores the
  remainder of the input.
  Otherwise, it gets ready for the next round.
\item [($-$)]
If the output of $A$ was ``Minus'' for \emph{all} inputs
$u_1,\ldots,u_n$ of this round,
the machine $B$ \emph{rejects} the input once and 
for all 
and ignores the
  remainder of the input.
  Otherwise, it gets ready for the next round.
\end{itemize}
We call the three outcomes of a round
ACCEPT, REJECT, and INDECISION.
This type of aggregation is similar to 
the
Correctness Test in Section~\ref{sec:check-computation}
(see Figure~\ref{fig:checker})
except that the machine $A$ has no analog of the output ``Undecided''.

If all outcomes were INDECISION at the end of the input and no definite decision has been reached,
the machine $B$ rejects the input.
The same is true for an incomplete round,
i.~e., when the last symbol read was not \texttt{check}.
If the automaton encounters an ill-formed
input, where a \texttt{check} symbol follows a symbol
other than \texttt{end}, it also rejects.

It is easy to write the probabilities of the outcomes of a round:
\begin{align}
  \label{eq:plus}
  \Pr[\text{ACCEPT}]
 & = \frac12
 \phi(u_1)\phi(u_2) \cdots \phi(u_n) 
 \\
  \Pr[\text{REJECT}]
 & =\frac12
 \label{eq:minus}
 \bigl(1-\phi(u_1)\bigr)\bigl(1-\phi(u_2)\bigr) \cdots \bigl(1-\phi(u_n)\bigr) 
\end{align}
Figure~\ref{fig:amplify} represents the outcome
in analogy to
Figure~\ref{fig:EC} for the
Equality Checker of Section~\ref{sec:eq}.

\begin{figure}[htb]
  \centering
\begin{tikzpicture}[
roundnode/.style={circle, draw=black, minimum size=7mm},
squarednode/.style={rectangle, draw=red!60, fill=red!5, very thick, minimum size=5mm},
boxnode/.style={rectangle, draw=black, minimum size=5mm},
]

\node[boxnode,anchor=west] at (1,2.5)     (undec1)             {INDECISION};
\node[boxnode,anchor=west] at (4,4.5)     (eq1)             {ACCEPT};
\node[boxnode,anchor=west] at (4,3.5)     (diff1)             {REJECT};
\node[boxnode,anchor=west] at (1,4)        (dec1)       {decide};
\node[roundnode,label={left:Case 1: $\phi \le \frac12\quad $}
] at (-1,3.25)        (start1)     {}  ;

\draw[->] (start1) -- (dec1.west) node[above,midway,sloped]{(rarely)};
\draw[->] (start1) -- (undec1.west) node[below,midway,sloped]{(often)};
\draw[->] (dec1.east) -- (eq1.west) node[midway,above]{$\le\frac12$};
\draw[->] (dec1.east) -- (diff1.west) node[midway,below]{$\ge\frac12$};

\node[boxnode,anchor=west] at (1,-1.2)     (undec2)             {INDECISION};
\node[boxnode,anchor=west] at (4,1)     (eq2)             {ACCEPT};
\node[boxnode,anchor=west] at (4,0)     (diff2)             {REJECT};
\node[boxnode,anchor=west] at (1,0.5)        (dec2)       {decide};
\node[roundnode,label={left:Case 2: $\phi > \frac12\quad $}
] at (-1,-0.25)        (start2)     {}  ;

\draw[->] (start2) -- (dec2.west) node[above,midway,sloped]{(rarely)};
\draw[->] (start2) -- (undec2.west) node[below,midway,sloped]{(often)};
\draw[->] (dec2.east) -- (eq2.west) node[midway,above]{$\ge1{-}\eps$\quad \null};
\draw[->] (dec2.east) -- (diff2.west) node[midway,below,yshift=-1mm]{$\le\eps$};


\end{tikzpicture}

\caption{The outcome of processing one round by the PFA $B$,
  assuming a large number $n$ of repetitions.}
  \label{fig:amplify}
\end{figure}
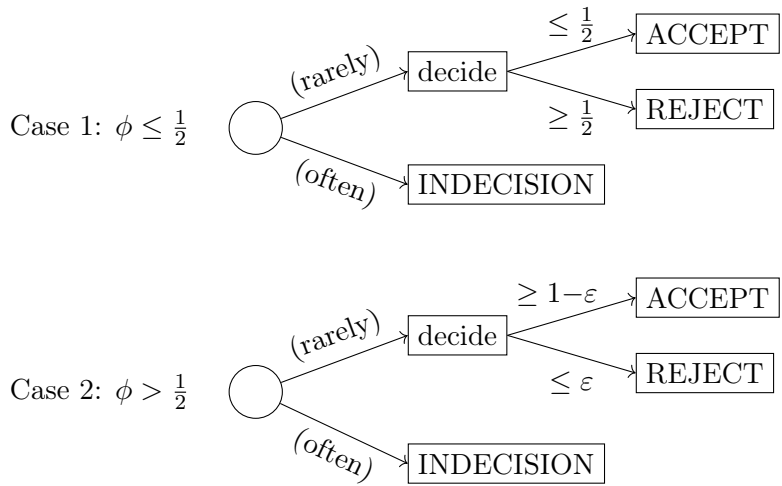

In case 1 of the theorem, when
 $\phi(u)\le1/2$ for every $u\in \Sigma^*$,
the probability  
  \eqref{eq:plus} cannot be larger than
  \eqref{eq:minus}.
This holds for every round, establishing the first statement of the theorem.  

Let us analyze the other case.
Let $u\in \Sigma^*$ be an input for $A$ with
$
\phi(u)>1/2$.
Then, in the round with $n$ repetitions of this input,
\begin{displaymath}
  (u\ \mathtt{end})^n\ \mathtt{check},
\end{displaymath}
the ratio $(\phi(u)/(1-\phi(u)))^n$ between
  \eqref{eq:plus} and 
  \eqref{eq:minus}
  can be made as large as desired by choosing $n$ sufficiently large.
  We choose $n$ such that
  \begin{equation}\label{eq:reject-conditional}
    \frac{  \Pr[\text{REJECT}] }
{  \Pr[\text{ACCEPT}] +  \Pr[\text{REJECT}]}
\end{equation}
becomes less than $\eps$.
As $n$ increases,  
  the chance of reaching any decision at all become smaller and smaller.
 We can counteract this effect by trying a large number $t$ of rounds:
\begin{equation} \label{equal-rounds}
  \bigl((u\ \mathtt{end})^n\, \mathtt{check}\bigr)^t
\end{equation}
Even if the single probability of reaching a decision,
$\Pr[\text{ACCEPT}] +  \Pr[\text{REJECT}]$, is very small,
the probability of INDECISION in $t$ rounds
\begin{equation}
  \label{eq:failure}
\bigl(1-\Pr[\text{ACCEPT}] - \Pr[\text{REJECT}]\bigr)^t
\end{equation}
converges to 0 as $t$ increases.

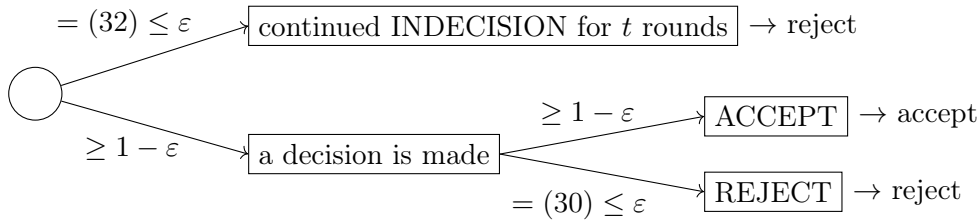
\begin{figure}[htb]
  \centering
\begin{tikzpicture}[
roundnode/.style={circle, draw=black, minimum size=7mm},
squarednode/.style={rectangle, draw=red!60, fill=red!5, very thick, minimum size=5mm},
boxnode/.style={rectangle, draw=black, minimum size=5mm},
]
\node[boxnode,anchor=west,label={right: $\to$ reject}]
at (1,2.2)     (undec2)
{continued INDECISION for $t$ rounds};
\node[boxnode,anchor=west,label={right: $\to$ accept}] at (7,1)     (eq2)             {ACCEPT};
\node[boxnode,anchor=west,label={right: $\to$ reject}] at (7,0)     (diff2)             {REJECT};
\node[boxnode,anchor=west] at (1,0.5)        (dec2)       {a decision
  is made};
\node[roundnode] at (-1.8,1.3)        (start2)     {}  ;

\draw[->] (start2) -- (dec2.west)
node[below,midway,xshift=-3mm]{$\ge1-\eps$};
\draw[->] (start2) -- (undec2.west) node[midway,above,xshift=-4mm,yshift=1mm]{$=\eqref{eq:failure}\le\eps$};
\draw[->] (dec2.east) -- (eq2.west) node[midway,above]{$\ge1-\eps$\quad \null};
\draw[->] (dec2.east) -- (diff2.west) node[midway,below,yshift=-1mm,xshift=-3mm]{$=\eqref{eq:reject-conditional}\le\eps$};
\end{tikzpicture}

\caption{The outcome of $t$ rounds.}
  \label{fig:reject}
\end{figure}

The overall outcome after $t$ rounds is illustrated in Figure~\ref{fig:reject}.
The
conditional probability
\eqref{eq:reject-conditional}
of rejection, conditioned on the event that a decision is made,
can be reduced
below $\eps$
by increasing $n$.
The
probability \eqref{eq:failure} of rejection due to indecision
can be reduced
below $\eps$
by increasing $t$ (depending on
 $n$).\footnote{Fijalkow proposes to take $t=2^n$, and choose $n$
   large enough: The sequence of indecision probabilities
   $\bigr(1-\frac12 x^n-\frac12(1- x)^n\bigr)^{2^n}$
   converges to 0 if $x\ne 1/2$.}
 In total, the acceptance probability is at least $(1-\eps)^2$.
\end{proof}
The procedure is very much reminiscent of the third-level aggregation
in the
Condon--Lipton proof (Section~\ref{sec:whole})
whose idea goes back to Freivalds~\cite{freivalds81}.

\subsubsection{Algorithm GO, Gimbert and Oualhadj}
\label{sec:algorithm-GO}

\begin{proof}[Original proof of
  Theorem~\ref{thm-amplify}]

  The proof of
  Hugo Gimbert and Youssouf Oualhadj \cite{gimbert-oualhadj-2010:PFA} from 2010
constructs a slightly different automaton~$B$. For contrast, we
present it here.
The input is processed in rounds as before,
  but acceptance and rejection is decided differently.

The machine flips a \emph{single coin} at the beginning, and with probability 1/2, it does one of the
following two things:
\begin{itemize}
\item [($+$)]
  If, in some round,
  the output of $A$  for {all} inputs
  $u_1,\ldots,u_n$ in that round
  was ``Plus'',
  the machine accepts; otherwise it rejects.
\item [($-$)]
 If, in some round,
  the output of $A$  for {all} inputs
  $u_1,\ldots,u_n$ in that round
  was ``Minus'',
    the machine rejects; otherwise it accepts.
\end{itemize}
For the analysis,
the treatment of Case 1 of the theorem is the same as before.

For Case 2, let us assume that there is an
input $u\in \Sigma^*$ for $A$ with
$x 
:=
\phi(u)>1/2$.
\begingroup \let\phi=x
Due to the single branch at the beginning of the algorithm, the
rounds are not independent, and the construction of an appropriate input becomes more
tricky.
Let us consider, as before,
 a number $t$ of identical rounds, each consisting of
 $n$ repetitions of $u$, as in \eqref{equal-rounds}.
 With $\bar\phi = 1-\phi$, we can write
 the conditional acceptance and rejection probabilities in the two branches as follows:
 \begin{align*}
   \Pr[\mathrm{ACCEPT}\mid (+)]
   &= 1-(1-\phi^n)^t&
   \Pr[\mathrm{ACCEPT}\mid(-)]
   &= (1-\bar\phi^n)^t\\
   \Pr[\mathrm{REJECT}\mid (+)]
   &= (1-\phi^n)^t
     &
   \Pr[\mathrm{REJECT}\mid(-)]
   &= 1-(1-\bar\phi^n)^t
 \end{align*}
 Keeping $n$ fixed (and large), we can assume that
both $\phi^n$ and $\bar\phi^n$ are close to~0, and
 we see a dilemma:
 If $t$ is small,  the $(+)$ branch
will almost always lead to rejection, because this is the default
result when the input runs out
before a decision
is reached.
If $t$ grows to infinity, we will almost always come to a definite
decision, no matter how small $\phi^n$ or  $\bar\phi^n$ is, and this
decision is REJECT in the $(-)$ branch.

In both of these extreme situations, the acceptance probability will be close to $1/2$.
We must therefore find the ``sweet spot'' with just the right number
$t$ of rounds.
This can be done as follows:
\smallskip
\begin{compactenum}
  
\item Choose some bound $\eps>0$ on the rejection probability ($\eps<1$).
\item Set $g := \eps/ \ln  \frac2\eps$.
\item Pick $n$ large enough such that $\bar\phi^n/ \phi^n \le g$
  and $\phi^n \le 1/2$.
\item Set $t := \lfloor
  \ln  \frac2\eps / \phi^n \rfloor$.
\end{compactenum}
\smallskip
Then, in the $(-)$ branch
\begin{align*}
   \Pr[\mathrm{REJECT}\mid(-)]
  &= 1- (1-\bar\phi^n)^t
    \le 1-(1-t\cdot \bar\phi^n)
  =t\cdot \bar\phi^n
  \\
  &\le  \ln  \tfrac2\eps / \phi^n \cdot g\phi^n
=      \ln  \tfrac2\eps\cdot g
    = \eps
\end{align*}
In the $(+)$ branch, we have
\begin{align*}
   \Pr[\mathrm{REJECT}\mid(+)]
  &= (1-\phi^n)^t
    \le (1-\phi^n)^{\ln  \frac2\eps / \phi^n -1},
\end{align*}
where the term $-1$ in the exponent accounts for the rounding of~$t$.
We estimate the corresponding factor by
$ (1-\phi^n)^{-1} \le (1-1/2)^{-1} = 2$, and get
\begin{align*}
   \Pr[\mathrm{REJECT}\mid(+)]
  &\le (1-\phi^n)^{\ln  \frac2\eps / \phi^n}\cdot 2
                          =2\left( (1-\phi^n)^{1/ \phi^n}\right)^{\ln
                          \frac2\eps}
         \le2( 1/e )^{\ln  \frac2\eps} 
         = \eps.
         \qedhere
\end{align*}
The original construction of the input by Gimbert and Oualhadj \cite[Section 7 of the technical report]{gimbert-oualhadj-2010:PFA}
is different:
 They use an increasing sequence $n_1,n_2,\ldots$ of round lengths in a special way.
With this setup, they can reach
 a small rejection probability
 by simply choosing the number $t$ of rounds large enough.  
\end{proof}

\endgroup

\subsubsection{The no-coin algorithm NC}
\label{sec:no-coin}

Comparing Algorithms F and GO, we can
speculate that Algorithm GO
might have its origin in a shortage of coins. We find that we actually
don't need any coins at all,
as shown in the following
simplified version of Algorithm~F.

Every round is processed as follows
when the symbol
  \texttt{check} is read:
\begin{itemize}
\item [($\pm$)]
  If the output of $A$ was ``Plus'' for {all} inputs
  $u_1,\ldots,u_n$ of the round,
  the machine $B$ \emph{accepts} the input once and
  for all 
  and ignores the
  remainder of the input.

  If the output of $A$ was ``Minus'' for {all} inputs $u_1,\ldots,u_n$ of the round,
the machine~$B$ \emph{rejects} the input once and 
for all 
and ignores the
remainder of the input.

In case of an empty round or an ill-formed round, the input is also rejected.

  Otherwise, the machine gets ready for the next round.
\end{itemize}
We had to come up with a rule 
for an empty round, because the first two
sentences contradict each other in this case. We have decided to reject such an input as
ill-formed. Algorithm~F would have chosen final acceptance or final rejection
with probability 1/2, but we are not in the mood for throwing coins.

The calculation goes through as before. The only change is that
the factor 1/2 is removed from
  \eqref{eq:plus} and 
  \eqref{eq:minus}.

\subsubsection{High-level verbal description versus state diagrams}
\label{sec:state-charts}

The above description in terms of an algorithm that is formulated in
words is not how the PFA $B$ is presented in
\cite
{Fijalkow-2017-SIGLOG}
and~
  \cite{gimbert-oualhadj-2010:PFA}.

Fijalkow
  \cite[Figure~1]{Fijalkow-2017-SIGLOG}, following
  \cite
  {gimbert-oualhadj-2010:PFA},
starts by showing   
the \emph{expanding automaton} of
Figure~\ref{fig:expanding}. This is an abstraction of the automaton $B$ where
every simulation of the automaton~$A$ (for some input ``$u\ 
\mathtt{end}$'') is represented by 
the symbol \texttt{sim}
with a 
probability $x$ of output ``Plus''.
Two absorbing states $\top$ and $\bot$ represent ACCEPT and REJECT.\footnote
{We mention that Fijalkow's rendition of the expanding automaton
  \cite[Fig.~1]{Fijalkow-2017-SIGLOG} is an improvement over
  the corresponding automaton in
  \cite[Fig.~1
  ]{gimbert-oualhadj-2010:PFA}, which has 7 states,
  not only in representing an algorithm that is easier to analyze:
  He also chose meaningful symbols 
  \texttt{sim}, \texttt{check}, and \texttt{end}
  for
what is simply called $a$, $b$ and $\sharp$ in
  \cite{gimbert-oualhadj-2010:PFA}, as well as  meaningful names for some states.}

\begin{figure}[htb]
  \centering
  \begin{tikzpicture}[
    roundnode/.style={circle, draw=black, minimum size=9mm},
    style={->,shorten >=0.3pt,>={Stealth[round]},semithick},
    ]   
\node[roundnode] at (0,0) (q0)    {$q_0$}  ;
\node[roundnode] at (3.5,0) (RR)    {$-$}  ;
\node[roundnode] at (-3.5,0) (LL)    {$+$}  ;
\node[roundnode,double] at (-3.5,2.5) (Ac)    {$\top$}  ;
\node[roundnode] at (3.5,2.5) (Re)    {$\bot$}  ;
\node[below=10mm] at (q0.south) (start) {}; 
\draw[->] (start) -- (q0);
\draw[->] (q0) to [bend right] node[below,midway,sloped]{$\texttt{check}:\frac12$} (RR);
\draw[->] (RR) to [bend right] node[above,midway]{$\texttt{sim}:x$}  (q0);
\draw[->] (q0) to [bend left] node[below,midway]{$\texttt{check}:\frac12$} (LL);
\draw[->] (LL) to [bend left] node[above,midway]{$\texttt{sim}:1-x$}  (q0);
\draw[->] (LL) to node[left,midway]{$\texttt{check}:1$}  (Ac);
\draw[->] (RR) to node[right,midway]{$\texttt{check}:1$}  (Re);
\draw[->] (Ac) to [out=115,in=65,looseness=8]
     node[above,midway]{$\texttt{sim},\mathtt{check}:1$} (Ac);
\draw[->] (Re) to [out=115,in=65,looseness=8]
     node[above,midway]{$\texttt{sim},\mathtt{check}:1$} (Re);
\draw[->] (q0) to [out=115,in=65,looseness=9]
node[above,midway]{$\texttt{sim}:1$} (q0);
\draw[->] (LL) to [out=240,in=190,looseness=10] 
node[below,midway]{\strut$\texttt{sim}:x$} (LL);
\draw[->] (RR) to [out=350,in=300,looseness=10] 
node[below,midway]{\strut$\texttt{sim}:1-x$} (RR);
  \end{tikzpicture}
  \caption{The expanding automaton}
  \label{fig:expanding}
\end{figure}
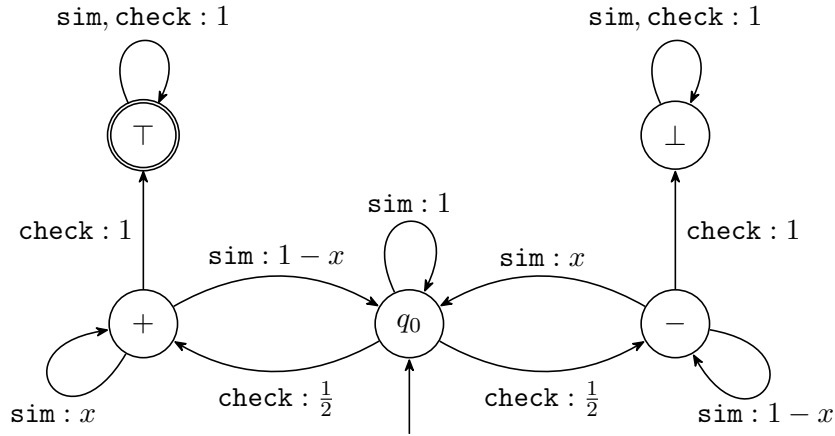

With this interpretation, the behavior of the automaton is very close
to Algorithm~F, with some insignificant differences.
We see that the coin is thrown \emph{before} each round, a
detail that was left unspecified in Algorithm~F.
Consequently,
the automaton ignores everything until the first \texttt{check} is
read,
at which point the coin for the first round is 
thrown. Every subsequent \texttt{check}
has a dual role: It moves the automaton to one of the two absorbing states
$\top$ or $\bot$ when a decision is made, or otherwise it throws the
coin for the next round.

With the expanding automaton, one can very nicely analyze and illustrate the
distinction between the cases $x\le 1/2$ and $x>1/2$, for which
the respective conclusions in
  Theorem~\ref{thm-amplify} arise.

  The actual automaton $B$ is obtained from the expanding automaton by
  substituting two copies of the automaton $A$, identifying the start
  state of $A$ with the states $+$ and $-$, respectively.
  When reading the symbol \texttt{end}, the automaton makes the
  appropriate transitions to $q_0$, $+$, or $-$, accordingly as the
  automaton $A$ is in an accepting or rejecting state (for $A$),
and depending on whether
  we are in the left ($+$) or right ($-$) branch.

   Fijalkow
  \cite[Figure~2]{Fijalkow-2017-SIGLOG} gives a schematic drawing of
    the automaton~$B$,
    while
    \cite[Section 8 of the technical report]
    {gimbert-oualhadj-2010:PFA} specifies all transitions in detail
  in words.
    
A state diagram may have advantages through its immediate visual
appeal. However,
I think the difference between Algorithms~F and GO can be better
explained in verbal terms than by studying state diagrams and
abstracting the behavior from it.
(Fijalkow 2017
  \cite[p.~15--16]{Fijalkow-2017-SIGLOG} does not explain the distinction but
  rather downplays the difference.\footnote
{
%
``The construction is essentially the same as for the undecidability of
the value 1 problem by Gimbert and Oualhadj [Gimbert and Oualhadj
2010], the main improvements are in the correctness proof.''})
Moreover, constructing a PFA for a particular purpose via a state
diagram is
a tricky business.
Indeed,
the dual task of the \texttt{check} symbol as aggregating the outcomes
of the previous round and throwing the coin for the next round leads
to an
ambiguity whether this symbol
(the symbol
$b$ in the notation of   \cite{gimbert-oualhadj-2010:PFA})
should come at the beginning or at the end of a round.
Accordingly, some statements in
\cite[Section 7 of the technical report]
{gimbert-oualhadj-2010:PFA} are incorrect as written.
%

A more significant oversight is shared by
 \cite{Fijalkow-2017-SIGLOG} and
  \cite{gimbert-oualhadj-2010:PFA}:
For a word $u_0\in \Sigma^*$ that leads back to the start state of
${A}$ with positive probability,
 the ill-formed input
  ``$\texttt{check}\ u_0\ \mathtt{check}$''
is accepted by the automaton ${B}$ with positive probability,
 contrary to the intended meaning and
 contrary to what is claimed.
%

There are two ways to address this: One simple way is to
fix the automaton ${A}$ by adding a new start state,
ensuring that  ${A}$ never goes
back to the start state from another state.
Another possibility is to extend the analysis by observing that
the input ``$\texttt{check}\ u_0\ \mathtt{check}$''
makes equal contributions to ACCEPT and REJECT, and thus it cannot raise
the acceptance probability over $1/2$.\footnote{\label{fix-fij}%
  Nathanaël Fijalkow, private communication, May 17, 2024.}
A similar analysis applies to other ill-formed inputs
like
 ``$\texttt{check}\ u_1\ \mathtt{end}\ u_0\ \mathtt{check}$''.

The reader may ask, why one should be so picky about technical
mistakes. Isn't it more important to convey the main ideas?
Well, 
this precisely underlines the advantages of a high-level
verbal description
that I am advocating,
over an apparently more precise (semi-)formal
treatment.
If the details
are not important, why bother with them?
We can just \emph{specify} that the automaton should reject incomplete
or ill-formed inputs, leaving out any details
about representing the algorithm with states and
transitions until someone is interested in the precise number of states.

Indeed, such a study can be made: If the automaton $A$ has $d$ states,
Fijalkow's automaton $B$, which has the required properties of
  Theorem~\ref{thm-amplify} and which
carries out Algorithm~F with the above-mentioned modifications of
behavior, has $2d+3$ states. The automaton
$A$ can have an arbitrary starting distribution, and it can even have
output values $f_i\in[0,1]$. The automaton $B$ will have a single
start state and a single accepting state.

In this respect, Algorithm F actually has an advantage over the
simple no-coin algorithm~NC: This algorithm would need 3 copies
of the automaton $A$: one for the first input $u_1$ in each round, and two
separate copies when the result of the first input is known. Thus, the
number of states would be $3d$ plus some constant overhead.

\subsection{Case study 3: Coding in binary}
\label{sec:reflection-binary}

This example and the following case study~4
are about issue~(2): Presenting material in a self-contained way versus relying
on powerful general results.
We
look at
Lemma~\ref{coding} from
Section~\ref{sec:coding}
about encoding the input alphabet in binary for input to the PFA.
Binary coding is a fundamental and common
procedure 
in Computer Science, and when one thinks about it in
terms of a (randomized or deterministic) algorithm, it is
a small technical issue that should hardly be worth
mentioning.
It has been explicitly formulated as a lemma in order to say something about the number of states.

\CodingBinary*
 The proof of
 this lemma 
 in
Section~\ref{sec:coding}
is direct and straightforward.
It describes the PFA $A'$ that does the job.
The proof is perhaps even a little too much on the
verbose side.
%
Writing the two explicit
transition matrices is an add-on,
to make the description more concrete.

For comparison, let us look at the
treatment in the original sources in
Nasu and Honda
\cite[Lemma 19, p.~269]{NasuHonda1969}
and in
Paz \cite[Lemma 6.15, p.~190]{paz71}.

They appeal to a general statement that
acceptance probabilities are preserved under
GSM-mappings (mappings induced by a \emph{generalized sequential machine})
\cite[Definition 6.2 and Theorem 6.10, p.~186]{paz71}.
A 
sequential machine is a finite automaton with output.
The machine is generalized in the sense
that, in one step,
the machine can also produce several symbols or no
 output at all.

 Decoding from the binary alphabet
$\{\texttt{a},\texttt{b}\}$
to the original alphabet $\Sigma$ can be easily carried out by a GSM.
However,
the application of a GSM
incurs a complication.
More concretely,
the encoding function $\tau$ in this construction encodes
some 7-letter alphabet $\Sigma$ (see Appendix~\ref{sec:nhp})
 by the binary codewords 
$\texttt{a}^i\texttt{b}$ for $i=1,\ldots,7$.
(Lemma~\ref{coding} uses a slightly more efficient variant of this
code, which is optimized to minimize the number of states.)

Any GSM that performs the decoding
$\tau^{-1}\colon \{\texttt{a},\texttt{b}\}^* \to
\Sigma^*$
has
the undesirable property
that ``unfinished'' strings, whose last codeword is incomplete,
are mapped to some word in $\Sigma^*$, namely to the
decoding of the last complete prefix of the form $\tau(u)$.
The same happens for completely illegal words, for example those that
contain \texttt{bb}.
This behavior is inherent in the way of operation of a GSM.

As a consequence, words $x
\in\{\mathtt{a},
\mathtt{b}\}^*$
 that are not of the form $\tau(u)$
for some $u$ can be accepted with some probability different from~0,
contrary to what is required in Lemma~\ref{coding}.
(In terms of the algorithm in our proof of Lemma~\ref{coding},
the PFA $A'$ would accept the input based solely on the state
$q$, ignoring the counter $i$.)

Thus,
in a separate step \cite[p.~269, last line]{NasuHonda1969},
the probabilities
have to be patched up to correct this.
A finite automaton, and in particular a PFA, can easily check the
well-formedness of $x$, i.e., membership in the regular language
  $\{
\mathtt{ab},
\mathtt{aab},
\ldots,
\mathtt{a}^{7}\mathtt{b} \}^*$.
Both in \cite{NasuHonda1969} and
in Paz
  \cite[p.~190, lines 12--13]
  {paz71},
the correction is expressed very concisely
in two lines
  in terms of a characteristic function $\chi$
  of this 
  regular language (``event'')
  and the elementwise minimum operation
  $g' = g \land \chi$.
This refers to the  
intersection $f \land g$ of two fuzzy events (probabilistic languages),
which is defined
in \cite[Definition~2, p.~251]{NasuHonda1969}.

In Nasu and Honda's proof~\cite[p.~269, last line]{NasuHonda1969},
there is an additional complication: The probabilities
have to be patched up in two different ways, because of the more
ambitious goal of establishing that the encoded language is an
$E$-set. This requires two PFAs:
In one PFA, the
ill-formed words are accepted with probability 1, in the other,
with probability 0.\footnote
{In Paz  \cite{paz71},
 the same situation occurs at a different place, namely
 when dealing with deterministic linear languages
 \cite[Proof of Lemma 6.11, p.~188]{paz71}.}

 Closure under intersection and union with a \emph{regular} set
 (or event)
is treated by Nasu and Honda 
  in \cite[Theorem 4, p.~252
]{NasuHonda1969}
in somewhat obfuscated form, 
citing their earlier paper~\cite{NasuHonda1968} from  1968.
In Paz, it appears  
in the same form in
\cite[Proposition 1.9 of Chapter IIIA, pp.~148--9]{paz71}
with proper credit to \cite{NasuHonda1968}
(see \cite[Section IIIA.3, p.~152]  {paz71}).

However, when one looks up these statements,
one does not directly find
 the required
 closure property but something more general.
     Formulated in our notation, the statement says the following:
If $\phi(u)$ and $\psi(u)$ are acceptance probabilities of
  (rational?) PFAs and the set
  $  \{\,u\in \Sigma^*\mid
  \phi(u)>\psi(a)\,\}$ is a regular language,
  then
  $\max\{\phi(u),\psi(u)\}$ and 
  $\min\{\phi(u),\psi(u)\}$
 (the union and intersection of two fuzzy languages)
     are also acceptance probabilities of
 (rational?)
 PFAs.
Still, the required statement (closure under intersection or union with a
regular language) it is derivable as an easy corollary.


Overall,
 applying the results on GSM-maps makes the proofs very concise:
 The proof of Lemma~19
by Nasu and Honda
\cite[p.~269]{NasuHonda1969} takes 17 lines, of which 7 are used for
the complete and detailed specification of the GSM.
The proof
 in
 Paz \cite[Proof of Lemma 6.15, p.~190]{paz71},
 where the GSM is left to the reader (see footnote~\ref{inaccurate2}),
consists of only
8 lines. 
However, as we have discussed, the application of the general 
results about GSMs
comes with
a considerable conceptual and technical overhead, some of which is hidden by appealing to statements that don't
quite fit.

There is a final ironic twist: Since we are concerned with the
emptiness question, the inaccuracy introduced by the GSM would be
of no consequence!  It we assign to a word $x$ that is not of the form
$\tau(u)$ the acceptance probability of some other word $u'$, this does not
change the answer to the question whether acceptance probabilities
$>\lambda$ exist.
In the context of the emptiness question, patching up the probabilities
of ill-formed words is unnecessary.

After this confusing turmoil, the reader may want to go back and savor the
directness of the self-contained 13-line 
proof
of Lemma~\ref{coding} on p.~\pageref{coding}.

\subsubsection{Integer matrices}
\label{sec:integer-hirv}
I want to come back to the first issue,
(1)~choosing the right level of abstraction, or rather, choosing an appropriate
presentation.
I will use another instance of binary coding for illustration.
I will contrast a representation of transition matrices
 as block matrices with a representation that defines the entries by a
 formula via case distinctions.
\paragraph{Hirvensalo 2007.}

The binary coding technique can be used 
for 
products of integer matrices.
In this context, the method was used by 
Hirvensalo \cite[Step 3 of Section~3]{hirvensalo07} (see also
\cite[p.~5]{tHirvensalo06a}).
However, when
 a number of integer matrices $M_1,\ldots,M_k$ is given,
 the metaphor
of Lemma~\ref{coding}  
 of an algorithm
that performs random
transitions
``according to $M_i$'' is no longer appropriate.
In Section~\ref{sec:Hirvensalo-to-binary}, where 
the method is described,
we have therefore resorted to
our next-best choice: a description in terms
of block matrices.

Hirvensalo 
describes a code that is essentially the same as the code in
Lemma~\ref{coding}  
and explains the idea of the algorithm that does the decoding
as it reads the binary input.
The rows and columns of the new matrices
are indexed by ``states''
$(i,q) \in
\{0,1,2 
\} \times \{q_1,q_2,q_3,q_4,q_5,q_6\} 
$
(in notation of the proof of
Lemma~\ref{coding}). 

In
\cite{hirvensalo07,tHirvensalo06a},
the transitions are specified very concisely by four formulas that, for each
pair of states $(i,q),(i',q')$, gives
the entry for row
$(i,q)$ and column $(i',q')$ in the matrices
$M'_{\mathtt{a}}$ and $M'_\mathtt{b}$, in terms of the
entries of the
 original matrices $M_j$
and a case distinction 
involving the relation between $i$ and~$i'$. 

There is the added twist that the sixth rows acts like an
absorbing state, and we need not bother about continuing the
decoding process for the corresponding state $q_6$ (cf.~the remark after
Theorem~\ref*{thm-fixed-f-2-matrices} in Section~\ref{sec:coding}).
This is accommodated in Hirsenhalo's formulas by using only the state
 $(0,q_6)$ but not the states
 $(1,q_6)$ and
 $(2,q_6)$ in the new automaton, thus saving two states.
A schematic figure
\cite[Figure 1]{hirvensalo07}
 accompanies the description
 of the decoding process.

 With this type of description, while it is very condensed, it is challenging
 to make a formal correctness proof, and it would be next to impossible to
 uncover the mistake in the setup of the final vector $f_3$ that is
 mentioned in footnote~\ref{Hirvensalo-mistake}.  Here, the block
 matrix representation of Section~\ref{sec:Hirvensalo-to-binary},
 is more convenient, and it allowed us to discover the proper setting
 of~$f_3$.



\paragraph{Blondel and Canterini 2003.}
A simpler but slightly wasteful method for reducing the number of matrices to two
was used earlier by
Blondel and  Canterini \cite[Step 2, p.~235]{blondel-canterini-03}.
%
One matrix, $M_\mathtt{b}'$, is simply a block diagonal matrix of the
original matrices $M_i$. The other matrix
$M_\mathtt{a}'$ performs a cyclic permutation of the blocks.
By a suitable power of
$M_\mathtt{a}'$, the ``action'' can be brought to any desired block of
 $M_\mathtt{b}'$, so-to-speak.
\begin{equation}\label{eq:cyclic}
  M_\mathtt{a}'
  =
  \begin{pmatrix}
  0 & I &  &    &&  \\
    & 0 & I &   &&  \\
    &   & 0 &   &&  \\[-0.8ex]
 & 
 & 
 & \!\!\ddots
 && 
\\
  &  & &
 & 0 & I\\
 I\!\!\! &  & & &  & 0
\end{pmatrix}
\text{ and }
M_\mathtt{b}'
=
  \begin{pmatrix}
    M_1 & 0 & 0 & \cdots & 0\\[0.4ex]
    0 &M_2 & 0 & \cdots & 0\\[0.4ex]
    0&0&M_3 & \cdots & 0\\[0.4ex]
    \vdots & \vdots & \vdots &\ddots
 & \vdots\\[0.9ex]
    0&0 & 0 & \cdots & M_k
  \end{pmatrix}.
  \qedhere
\end{equation}
The number of states is multiplied by $k$ instead of $k-1$.
This is particularly striking as the main objective of the paper
\cite{blondel-canterini-03}
is
to get a small number of states.\footnote
{To be fair, one should remember that the method is used for of products of arbitrary, not
  necessarily stochastic matrices.
While the idea that one is just coding the input in another (binary)
alphabet comes very natural in a setting where one thinks about an
automaton that reads a sequence of symbols,
it is not so suggestive 
when thinking about matrix products.

In
Hirvensalo
\cite[Section~4]{hirvensalo07,tHirvensalo06a}, the method
\eqref{eq:cyclic}
is used in the context
of quantum finite automata. In this setting, it is the method of choice because the
transformation matrices have to be unitary.}
In any case, the block matrix representation makes the idea very
clear, and it is
suited for clean proofs.

\subsection{Case study 4: Testing equality}
The equality test requires a PFA with acceptance probability
$1/2 - (\phi(a)-\psi(a))^2/4$
according to formula~\eqref{eq:equality-trick}.
In Section~\ref{check:equality}, we have 
directly the PFA for this particular example.
By contrast,
the original proof
\cite[Lemma 11 on p.~259]{NasuHonda1969}
refers to general closure properties
of
probabilistic events (languages) %
under complementation, convex combination, and multiplication
\cite[Propositions 1--2 and Theorem 3 on p.~252,
which are taken from \cite{NasuHonda1968}]
{NasuHonda1969}.

In contrast to the previous example of binary coding, there is not such a clear win
for one version of the proof or the other.
While the proof in Section~\ref{check:equality} is direct and does not
build on auxiliary results,
it
talks informally about ``mixing''
several PFAs, which is just a different way of saying that
one forms a convex combination.
It introduces the dash-dotted transitions
of Figure~\ref{fig:mix}
that are taken
before reading the input. 
 The formal details, how these transitions can be eliminated and how this is expressed in terms of
 transition matrices, are also given. 
The proof in 
\cite[Theorem 4.2]{NasuHonda1968} treats
convex combinations with a different method:
 one does not need an extra start state; one can just adapt
 the
 starting distribution~$\pi$. (In our situation, the introduction of the new start
state was helpful in order to eliminate the empty word.)

\subsection{Using auxiliary results or starting from scratch}

In a monograph, a textbook, or a longer treatise, it is natural to
accumulate a body 
of techniques and results,
as well as 
notation and abbreviations,
on which further
results rest.

Of course,
after erecting such an edifice it makes sense to build on it.
Such an approach may allow short and potentially elegant proofs.
However,
if a proof comes out too terse,
especially if this is coupled with
inaccuracies\footnote
{Besides the cases discussed above, see 
  footnotes
\ref{inaccurate1} and
\ref{inaccurate2}.
Also,
in \cite[p.~190, l.~14]{paz71}, 
$\tau(g',\lambda)$ should be replaced by
 $T(g',\lambda)$.
 This typographical error is disturbing because $\tau$
 appears frequently in the same paragraph (see Appendix~\ref{sec:nhp}),
 whereas 
 $T(g',\lambda)$ is defined in a quite remote location
 \cite[Definition 1.1 of Section IIIB, p.~153]  {paz71}.
}
and a
lack of explanations (What is the main idea? 
Which parts of the construction are only technical machinery to change acceptance
$\ge\lambda$ into
$>\lambda$ or to reduce the alphabet size to~2?), 
it makes a proof practically inaccessible to a reader who is
interested only in a particular result,
to the point of becoming 
hermetic.\footnote
{For an experience with the contrary outcome, let me cite the following
  passage about the 
  philosopher Thomas Hobbes (1588--1679)
  from his contemporary biographer John Aubrey
  \cite[p.~150]{aubrey_brief_lives}\footnotemark:
  ``He was 40 yeares old before he looked on Geometry; which
  happened accidentally. Being in a Gentleman's Library, Euclid's
  {Elements} lay open, and 'twas the 47 \emph{El.\ libri}~1.
  He read the Proposition. `By G---,' sayd he (he would now and
  then sweare an emphaticall Oath by way of emphasis) `this is
  impossible!'
  So he reads the Demonstration of it, which referred him back to such
  a Proposition; which proposition he read. That referred
  him back to another, which he also read. \emph{Et sic deinceps} [and
  so on]
  that at last
  he was demonstratively convinced of that trueth. This made him in
  love with Geometry.''
}\footnotetext{\url{https://archive.org/details/AubreySBriefLives/page/n273/mode/2up}}%
$^,$
\footnote{
There is a similar dilemma in writing software. Should one
use library functions for a particular task,
or should one simply program everything from scratch?
Using a library often comes with an overhead, because the functions
that it offers don't usually fit exactly and require adaptation.

There is, however, a crucial difference between writing software and  writing mathematics:
In programming,
the primary goal is to get things done, as effectively as possible. The target
is the computer, and the machine does not have to understand what it
is doing, and why. (But yes, software needs to be maintained, and
another 
programmer will have to read and understand the computer code.)
Also in mathematics, one wants to get things proved, but
the primary target audience is the mathematical reader.
Besides convincing the reader of correctness, one also wishes to
elicit some understanding.}

In this article, I have enjoyed the freedom of
concentrating on a single result and
proving everything
from scratch.
Working with concrete automata 
in terms of explicit transition matrices
has allowed me to come up with the
specializations and strengthenings that I have presented.

\goodbreak
{ \let\oldsection=\section
  \def\section*#1{\oldsection*#1%
    \addcontentsline {toc}{section}{\numberline {}References}%
  }
\bibliography{PFA-undecidable}
}

\appendix



\section{The original Nasu--Honda proof in a nutshell}
\label{sec:nhp}

For comparison, we summarize the proof  originally
given by Nasu and Honda in
\cite{NasuHonda1969} and reproduced in the monograph of Paz
\cite[Chapter IIIB, Section 6]{paz71}.
It illustrates how the people had to struggle with the material when
it was new.
The presentation 
of Paz is quite faithful to the original, to the point of
using identical notation for many notions. It
is more condensed, 
slightly simplified, but
defaced by the occasional typographic or editing error, such as a
mysterious passage in a
proof that should
obviously belong to another proof. 

We have adapted the notation to our notation.
All references to \cite{paz71} are to Chapter IIIB. (Items in that book are
numbered separately in each chapter.)

\paragraph{Post's Correspondence Problem.}
The PCP with $k$ word pairs
$(v_i,w_i)$ over
$\{\texttt{0},\texttt{1}\}$
is represented by the language $
L(v,w)$ of words of the form
$$
\texttt{0}^{a_m}\texttt{1}
\ldots
\texttt{0}^{a_2}\texttt{1}
\texttt{0}^{a_1}\texttt{1}
\texttt{+}
v_{a_1}
v_{a_2}
\ldots v_{a_m}
\texttt{X}
w_{b_n}^R\ldots w_{b_2}^R w_{b_1}^R
\texttt{+}
\texttt{1}\texttt{0}^{b_1}
\texttt{1}\texttt{0}^{b_2}
\ldots
\texttt{1}\texttt{0}^{b_n}
$$
for sequences
$a_1\ldots a_m$
and $b_1\ldots b_n$ with
$a_i,b_j\in \{1,\ldots,k\}$,
where the superscript~$^R$ denotes reversal,
see
\cite[p.~265, before Lemma 16]{NasuHonda1969},
  \cite[Lemma 6.13, p.~189]{paz71}.
  This language is intersected with the set $L_s
  =\{\,y\texttt{+}z\texttt{X}z^R\texttt{+}y^R \mid
  y,z\in \{\texttt{0},\texttt{1} \}^*\,\}
$
  of palindromes with central
symbol
\texttt{X} and two occurrences of the separator symbol
``\texttt{+}''
\cite[p.~265, Lemma 15]{NasuHonda1969}.
The intersection
$L(v,w)\cap L_s$
is well-known to represent the PCP solutions
\cite[p.~270, Lemma 20]{NasuHonda1969},
because the palindrome property ensures that both the index sequences
and
the produced words match between the two sides.
 Its emptiness (apart from the single word \texttt{+X+}) is therefore undecidable
 \cite[Lemma 6.16, p.~190]{paz71}.

Actually, the alphabet of these languages is
$\Sigma=\{
\texttt{0} ,\texttt{1},\texttt{+},
\texttt{X},
\bar{\texttt{0}},
\bar{\texttt{1}},
\bar{\texttt{+}}
\}$,
because, for technical reasons, every symbol $\sigma$ in the right half, after
the \texttt{X}, is replaced by another symbol, its ``complemented'' version~$\bar \sigma$. 

\paragraph{Rational automata, $P$-sets, and $E$-sets.}

\begin{definition}[{\cite[Definition 17, p.~259]{NasuHonda1969}}]

A \emph{rational PFA} is a PFA where all components of $\pi$, $f$, and
the transition matrices $M\in\mathcal{M}$ are rational numbers.
  
An \emph{$E$-set} is the set of words $u$ with $\phi(u)=\psi(u)$, where
$\phi(u)$ and $\psi(u)$ are the acceptance probabilities of two rational PFAs.

A \emph{$P$-set} is a language recognized by a rational PFA with some
rational cutpoint $\lambda$, or in other words,
 the set of words $u$ with $\phi(u)>\lambda$.

\end{definition}

We have already seen in Section~\ref{check:equality}
the very important statement
that every $E$-set is a $P$-set
\cite[Lemma 11, p.261]{NasuHonda1969},
 \cite[Corollary 6.4, p.~183]{paz71}.\footnote
{In \cite[p.~182]{paz71},
 these sets have been renamed to \emph{$E$-events} and \emph{$P$-events},
respectively.
I find this choice of terminology, which goes back to
Rabin~\cite[p.~233]{rabin1963} and
pervades much of the literature,
unfortunate, since an event is rather something whose probability is
to be measured. The terms
$E$-language and $P$-language would have been even more specific.
}

\paragraph{Deterministic linear languages.}

The goal is to show that
$L(v,w)$, $L_s$, and finally
$L(v,w)\cap L_s$ are $E$-sets (and therefore $P$-sets).

This is done by defining
a certain class of context-free
languages, the so-called \emph{deterministic linear languages}
(\cite[Definition 18, p.~263]{NasuHonda1969},
\cite[Definition 6.3, p.~187]{paz71}),
and appealing to the fact that these languages
are
$E$-sets
\cite[Lemma 14, pp.~263--265]{NasuHonda1969},
  \cite[Lemma 6.11, pp.~187--188]{paz71}.
  The proof 
  relies on $m$-ary 
automata.

  It is then easy to see that this language class includes $L_s$, and also
the language $L(v)$ of words of the form
$$
\texttt{0}^{a_m}\texttt{1}
\ldots
\texttt{0}^{a_2}\texttt{1}
\texttt{0}^{a_1}\texttt{1}
\texttt{+}
v_{a_1}
v_{a_2}
\ldots v_{a_m},
$$
which form the left half of the words of
$L(v,w)$, up to the \texttt{X} symbol. 
Similarly, the right
halves form a deterministic linear language.

Putting the two halves together
to get the language $L(v,w)$
takes more effort.
The proof in
\cite[Lemma 17, p.~265--269]{NasuHonda1969} is cumbersome and stretches
over more than three pages.
Paz formulates the argument as a separate lemma \cite[Lemma 6.12, p.~188]{paz71}:
\begin{lemma}
  \label{LXL}
If
$L_1$ and $L_2$ are
 deterministic linear languages
 over
disjoint alphabets $\Sigma_1$ and $\Sigma_2$ not containing the letter $\mathtt{X}$, then $L_1\mathtt{X}L_2$ is an $E$-set.
\end{lemma}
From the representation of
 deterministic linear languages as $E$-sets, we have four 
 $m$-ary automata: two for $L_1$ and two for $L_2$.
The trick is to mix 
two of these $m$-ary automata (one for $L_1$ and one for $L_2$) into
an
$m^2$-ary automaton:
Out of the $m^2$ digits,
one
automaton uses the digits
$m,2m, \ldots, (m-1)m $;
the other
automaton uses the digits
$1,2,\ldots, m-1$.
The acceptance probability is constructed as an $m^2$-ary number,
but when it is viewed
in the $m$-ary expansion, 
the digits alternate between the digits for $L_1$
and the digits for~$L_2$.

The construction is repeated for the two other PFAs defining the 
$E$-sets $L_1$ and $L_2$.
Equality of acceptance probabilities then means equality for both
$L_1$ and $L_2$.%
\footnote
{The requirement of disjoint alphabets is the reason for introducing
the complemented 
symbols
$\bar{\texttt{0}},
\bar{\texttt{1}},
\bar{\texttt{+}}$.
Technically, it should not be necessary; any PFA could do
this conversion implicitly after 
the letter \texttt{X} as it reads the input from left to right. 
%
}

\paragraph{Intersections of $E$-sets.}
To get from $L(v,w)$ and $L_s$ to
$L(v,w)\cap L_s$, one uses the property that
$E$-sets are closed under intersection.
This is easy to prove by comparing the 
acceptance probabilities
$\tfrac 12 - \tfrac14 (\phi_1-\psi_1)^2$ and
$\tfrac 12 + \tfrac14 (\phi_2-\psi_2)^2$,
which are constructed along the lines of
formula~\eqref{eq:equality-trick} on p.~\pageref{eq:equality-trick}
\cite[Lemma 12, p.~261]{NasuHonda1969}.\footnote
{Paz, in the corresponding passage of his proof \cite[Proof of Lemma 6.14, p.~190]{paz71},
 appeals to
``Exercise 4.a.3'', but in that exercise
\cite[
p.~172]{paz71},
closure under intersection is only proved for
$E$-sets that are defined by the condition $\phi(u)=\lambda$ for a
\emph{constant~$\lambda$}.
\label{inaccurate1}
}

Actually, this intersection property leads to a simple alternative proof of
a generalized version of
Lemma~\ref{LXL}, where 
$L_1$ and $L_2$ can be arbitrary $E$-sets:
\newcommand{\LXLcount}{}
\edef\LXLcount{\arabic{lemma}$'$}

\begingroup\let\thelemma\LXLcount
\begin{lemma}
  \label{LXL-general}
If
$L_1$ and $L_2$ are
$E$-sets
 over
some alphabet $\Sigma$ not containing the letter $\mathtt{X}$, then $L_1\mathtt{X}L_2$ is an $E$-set.
\end{lemma}
\addtocounter{lemma}{-1}
\endgroup
\begin{proof}
The languages
$L_1\texttt{X}\Sigma^*$
and $\Sigma^*\texttt{X}L_2$ are certainly $E$-sets: A~PFA can simply ignore
all symbols before or after the \texttt{X}. Their intersection is 
 $L_1\mathtt{X}L_2$. 
\end{proof}

\paragraph{Coding with 2 symbols.}
Finally, the 7-character alphabet $\Sigma=\{
\texttt{0} ,\texttt{1},\texttt{+},
\texttt{X},
\bar{\texttt{0}},
\bar{\texttt{1}},
\bar{\texttt{+}}\}$
is converted to binary by a coding function
$\tau$ that uses the codewords
$\texttt{a}^i\texttt{b}$ for $i=1,\ldots,7$.
We have already seen in 
Lemma~\ref{coding}, for a very similar code, that this can be done
while preserving the acceptance probabilities.
The step is discussed in 
detail in
Section~\ref{sec:reflection-binary}.

Paz
appeals to a general statement that
acceptance probabilities are preserved under
GSM-mappings (mappings induced by a \emph{generalized sequential machine})
\cite[Definition 6.2 and Theorem 6.10, p.~186]{paz71}.
(A generalized sequential machine is a finite automaton with output
over some alphabet $\Sigma$.
The machine is generalized in the sense
that
the output at every step is from
 $\Sigma^*$, i.e., the machine can produce several symbols or no
 output at all.)
Nasu and Honda use their even more general statement about
PGSM-mappings (mappings induced by a \emph{probabilistic GSM}
\cite[Definition 13 and Theorem 6, pp.~253--255]{NasuHonda1969}),
which
are an object of study in the paper
\cite{NasuHonda1969} and
figure 
prominently in its title.

Nasu and Honda use this to show that
$\tau(L(v,w)\cap L_s)$ is an $E$-set
\cite[Lemma 19, p.~269]{NasuHonda1969}.
Paz proves the weaker statement that
$\tau(L(v,w)\cap L_s)$ is a $P$-set
\cite[Lemma 6.15, p.~190]{paz71}
and is therefore able to shortcut the proof.\footnote
{\label{inaccurate2}%
This proof has a small technical mistake
  \cite[Proof of Lemma 6.15, p.~190, line 8]{paz71}:
  It claims that
\emph{one can easily construct a GSM mapping $\Psi^A$} to do the
decoding.
Some words $x$ are not
 decodable for the reason that they
 contain more than 7 \texttt{a}'s in a row or do not end with~\texttt{b}.
For these words, it cannot be ensured that
$\Psi^A(x)=e$ ($e$ denotes the empty word, cf.\ \cite[p.~xix, and
the footnote on p.~155]{paz71}) as Paz claims.
 The mistake is of no consequence because such ill-formed strings $x$
 are filtered out in a subsequent step, see Section~\ref{sec:reflection-binary}.
In the original proof of Nasu and Honda \cite[p.~269]{NasuHonda1969}, a GSM for decoding is described explicitly.}

Since every $E$-set is a $P$-set (Section~\ref{check:equality})
and since PFA Emptiness is about $P$-sets,
the undecidability of
 PFA Emptiness is established.

 \subsection{Deciding whether the recognized language is
    a regular  language,
    or whether it is context-free}
  \label{sec:context-free}

 It is also shown to be undecidable whether
 a $P$-set is a regular  language,
or whether it is context-free
\cite[Theorem 22, p.~270]{NasuHonda1969},
\cite[Theorem 6.17, p.~190]{paz71}.
This can be established by the following arguments,
which
are not given explicitly in
\cite{NasuHonda1969} or
\cite
{paz71}, but
only through 
unspecific references to the
theory of 
context-free languages.

If the PCP has a solution, then the language
$L(v,w)\cap L_s$ contains some word of the form
$$y\texttt{+}z\texttt{X}\hat z\bar{\texttt{+}}\hat y$$
with nonempty $y,z\in\{\texttt{0},\texttt{1}\}^*$,
where $\hat y = \bar y^R$ denotes simultaneous
complementation and reversal. 
We fix $y$ and $z$.
Since a PCP solution can be repeated arbitrarily, all words
\begin{equation}
  \label{sec:longsolution}
y^i\texttt{+}z^i\texttt{X}\hat z^i\bar{\texttt{+}}\hat y^i  
\end{equation}
for $i\ge0$ are also in the language. 
Moreover, intersection with the regular language
$\{y\}^*\texttt{+}\{z\}^*\texttt{X}\{\hat z\}^*\bar{\texttt{+}}\{\hat y\}^*$
leaves \emph{exactly} the words of the form~\eqref{sec:longsolution},
but for such a language it is known that it is not context-free.\footnote
{In fact, a language like
$\{\,\texttt{a}^i
\texttt{b}^i
\texttt{c}^i\mid i\ge0\,\}$ with just \emph{three} blocks of equal
length
is the prime
example of a language that is not context-free.}
Since the intersection of a context-free language with a regular language
is  context-free,
 the language
$L(v,w)\cap L_s$ cannot be context-free in this case.
We conclude that the recognized language is
context-free, or regular,
(namely, the one-word language $\{\texttt{+X+}\}$)
if and only if the PCP has no solution.

To get the result for a binary input alphabet, this whole chain of arguments must be transferred
from $L(v,w)\cap L_s$
to the \emph{encoded} language
$\tau(L(v,w)\cap L_s)$,
but this does not change the situation.

There is an alternative proof,
following an exercise in
Claus
\cite[p.~158, Aufgabe]{claus71}.
The exercise asks to derive the undecidability of 
testing whether
a given $P$-set 
is a regular language
as a corollary of the emptiness question.
The hint for the solution suggests to
take the language
$L$ of nonempty solution sequences
$a_1a_2\ldots a_m$ of the PCP,
and
append some
nonregular $P$-set $\tilde L$ to it.
This can be done by appealing to Lemma~\ref{LXL-general}.
We have already established that $L$ is an $E$-set
(Section~\ref{sec:PCP}).
Taking some nonregular $E$-set $\tilde L$,
for example the language
$\{\,
{\texttt{a}}^i
{\texttt{b}}^i\texttt{\#}\mid i\ge 0 \,\}$
of Section~\ref{equality-testing},
we can form
the $E$-set $L'=L\texttt{X}\tilde L$.
If the PCP has no solution, $L'$ is empty and therefore
regular.
If the PCP has a solution, $L'$ is not
regular because $\tilde L$ is not regular.

This construction can be extended for context-free languages by taking
the language
 $\tilde L=\{\,
{\texttt{a}}^i
{\texttt{b}}^i
{\texttt{c}}^i
\texttt{\#}\mid i\ge 0 \,\}$.
It is not context-free, but it is the
intersection of two $E$-sets
 $\{\,
{\texttt{a}}^i
{\texttt{b}}^i
{\texttt{c}}^n
\texttt{\#}\mid i,n\ge 0 \,\}
\cap\{\,
{\texttt{a}}^n
{\texttt{b}}^i
{\texttt{c}}^i
\texttt{\#}\mid i,n\ge 0 \,\}$
and therefore an $E$-set.


\ifnotation

\begingroup
\section{Notation, terminology, and abbreviations} 


\parindent 0pt
\everypar={\hangafter =1 \hangindent=13pt }
  
$A$ accepting computation of 2CM, appropriately encoded

$A(v,w)$, $\tilde A(v,w)$ special integer matrices that model word
pairs $(v,w)$

$A, A',B$ PFAs

$a_i$ ... $a_m$, $b_1,\ldots, b_n$ \emph{index \textbf{sequence}} for PCP solution

\texttt{a,b} symbols for equality checker, alphabet after conversion
to binary

$B(u)$ binary PFA

$B_i,C_i,D_i,E_i,F_i$ successively transformed transition matrices

$[0,c_1],
  (c_1,c_2],
  (c_2,c_3],\ldots$ interval boundaries

$d\times d$ size of matrix in joint spectral radius, size of matrix in PFA = \textbf{number of states}

$e_1,e_2,...$ (first) unit vector (local usage)

$\eps>0$. probabilities close to 0 or 1.

$\epsilon$ empty word (local usage)

$\eeta, \eeta_q $ characteristic vector (or more general, output
vector) of accepting states $q$.
$\eta^F$ is used in Nasu and Honda, where $F$ is the set
of accepting states. \cite{claus71} and most others uses $f$ instead
of $\eta$.
Paz p. 150 uses $\eta$. Blondel and Tsitsiklis 2000 use 0-1-vector $\eta$
(following Paz).

$G$ modulus

$\gamma,\gamma_1$ small constant (power of 2)

$\Gamma$ tape alphabet, including the blank symbol \blank

GSM Generalized sequential machine

$H$ halting state

$I$ unit matrix

$i,j$ lengths for equality checker

$J$ all-ones matrix, scaled by $1/d$.

$k$ number of input pairs for PCP

$K=10$ repetitions until declaring outcome ``INCORRECT''

$\lambda$ cutpoint

$l_i,r_i$ counter values in 2CM

$L,R$ movements of TM, as part of the rules (quite local usage)

MPCP Modified Post Correspondence Problem

2MPCP doubly Modified Post Correspondence Problem (not used much)

$\mathcal{M}$ set of matrices

$M,M_i,M_j$ matrices, transition matrices

$m_1,m_2,...,m_d$ columns of $M$, (local)

$m$ length of matrix product in PFA and in joint spectral radius

$m$ length of PCP solution / length of accepting computation in 2CM 

$n$ length of \emph{partial} PCP solution

$n$ length of a round, for amplification


PFA probabilistic finite automaton

PCP Post's Correspondence Problem

$p_{00},p_{01},...$ transition probabilities in 2x2 automaton $B(u)$

$\pi, \pi_q$ starting distribution

$Q$ TM states / 2CM states

$q_i$ ... $q_m$ states in accepting computation for 2CM

$q,q'$ states of TM (in rules) / state of PFA

$q_A,q_R$ extra states for PFA

RMPCP Reversed Modified Post Correspondence Problem (used only once)

$r_i,l_i$ counter values in 2CM

$x^R$ string reversal

start state, starting distribution, starting pair, starting vector

$s,s',t\in \Gamma$ tape symbols

$\sigma\in\Sigma$ (letter from) the input alphabet (of PFA). Currently
not used as alphabet of PCP

$\Sigma$ output alphabet of a GSM

$t$ number of repetitions of $A$ to boost the acceptance.

$t$ number of rounds, for amplification ($m$ in other publications)

$T$  Turing machine (local usage)

$\tau\colon \Sigma^*
\to
 \{\texttt{a},\texttt{b}\}^*$ binary encoding

$\tau^{-1}$ decoding, done by a GSM

$u$ input to TM or to PFA

$u,u'$, ($v$) binary strings

$U, U_W, U_0$ various universal Turing machines (local usage)

$0.u$ string $u$ interpreted as a binary fraction

$(v_i,w_i)$ ... $(v_k,w_k)$ PCP pairs, pairs of \textbf{words}!

$x=0.x_1x_2\ldots x_{47}$ random bits

$x=\phi(u)>1/2$, acceptance probability of particular $u$, for amplification

$y,z$ STILL MISSING THEIR ROLE!

$\Phi$, $\Psi$ states

$\phi=\phi(a)$, $\psi$ acceptance probabilities

\texttt{\#} separator symbol

\blank\ blank space

\endgroup

\fi

\end{document}
\section{Small UTMs, Overview}

\url{https://www.thi.uni-hannover.de/fileadmin/thi/abschlussarbeiten/2020/ba_strieker.pdf}
Universität Hannover,
Bachelor of Science,
Laura Strieker. 39 pp.

\url{https://link.springer.com/chapter/10.1007/978-3-642-03409-1_24}
\url{https://arxiv.org/pdf/0707.4489.pdf}
\url{https://arxiv.org/abs/0707.4489}
Small weakly universal Turing machines.
Turlough Neary, Damien Woods.
\cite{small-weakly-universal-2009}
We give small universal Turing machines with state-symbol pairs of (6,
2), (3, 3) and (2, 4). These machines are weakly universal, which
means that they have an infinitely repeated word to the left of their
input and another to the right. They simulate Rule 110 and are
currently the smallest known weakly universal Turing machines.

Recently,
Woods and Neary [29] have given 3-state, 7-symbol and 4-state,
5-symbol semiweakly universal machines that simulate cyclic tag
systems.

29. Damien Woods and Turlough Neary. Small semi-weakly universal Turing machines.
In Jérôme Durand-Lose and Maurice Margenstern, editors, Machines, Computations, and Universality (MCU), volume 4664 of LNCS, pages 306–323, Orléans,
France, September 2007. Springer.

Another in Fundamenta Informaticae 91 (2009) 179-195. 
10.3233/FI-2009-0011
\url{http://doi.org/10.3233/FI-2009-0011},
\url{https://mural.maynoothuniversity.ie/12415/}
with incorrect page numbers.
\url{https://mural.maynoothuniversity.ie/12415/1/Woods_SmallSemi_2009.pdf}
This work is an extended version of [34] and contains new results.
34 = MNU paper.

 3-state, 7-symbol machine
The cyclic tag system's list of appendants is reversed and encoded
to the left of the input. This encoded list is repeated infinitely
often to the left.

1 halting rule, 9 left, 11 right rules.

It is also known from the work of Margenstern [8], Michel [10], and Baiocchi [1] that the region between the non-universal curve and the smallest standard universal machines contains (standard) machines
that simulate the $3x + 1$ problem and other related problems.

[8] M. Margenstern. Frontier between decidability and undecidability: a survey. Theoretical Computer Science,
231(2):217–251, Jan. 2000

[10] P. Michel. Small Turing machines and generalized busy beaver competition. Theoretical Computer Science,
326:45–56, Oct. 2004.

[1] tech. report 1998, U. di Roma

\hrule

There are universal Turing machines with 4 states and 6 tape symbols.

Rogozhin's (4, 6) machine uses only 22 instructions, 
Rogozhin, Yurii (1996), "Small Universal Turing Machines", Theoretical
Computer Science, 168 (2): 215–240, doi:10.1016/S0304-3975(96)00077-1

states/symbols = 
(15, 2), (9, 3), (6, 4), (5, 5), (4, 6), (3, 9), and (2, 18)

 $U_{15,2}$ has 15 states and 2 symbols

``The initial state is u1 and the blank symbol is c.'' (Def 3.1)

 T. Neary and D. Woods / Four Small Universal Turing Machines
 p.121(wrong page number)

\vbox{
Table 16. Table of behavior for U15,2.
\begin{verbatim}
    u1   u2   u3   u4   u5   u6   u7   u8
c: cRu2 bRu3 cLu7 cLu6 bRu1 bLu4 cLu8 bLu9
b: bRu1 bRu1 cLu5 bLu5 bLu4 bLu4 bLu7 bLu7

    u9    u10   u11   u12   u13   u14   u15
c: cRu1  bLu11 cRu12 cRu13 cLu2  cLu3  cRu14
b: bLu10       bRu14 bRu12 bRu12 cRu15  bRu1
\end{verbatim}
}

 This means that the alphabet can be coded in 5 bits.
 (The halting state has to be added.)
The input words have then at most 15 bits, and
entries of the transition matrices are multiples of $1/4^{15}$.
A variable-length code might be more efficient,
each word contains at most one state, but several letters.
Using 5-letter codes of the form \texttt{0****} for the 15 states plus the halting
state leaves 3-letter codes \texttt{1**} for the 4 symbols
$\Gamma \cup \{\texttt{\#}\}$, leading to word lengths bounded
by $5+3+3=11$.

We only need to take care that the code for \texttt{\#} is not all
zeros, and then the trailing zero situation cannot arise.

e.g. (9, 3), 2 bits per symbol 2+4 bits per letter (sometimes less)

UTM(7,4) Section 8, Rogozhin, 7 states and 4 symbols

UTM(10,3) Section 8, Rogozhin, 7 states and 4 symbols

Four Small Universal Turing Machines

Neary and Woods
\cite[Section 3.5]{4-small-2009}
\footnote{see
  \url{http://mural.maynoothuniversity.ie/12416/}, 
  with 
  incorrect page numbers}

(15, 2),

(9, 3), Section 3.2 $U_{9,3}$.
13 right-moving rules, 1 halting rule,
13 left-moving rules. $|\Gamma|=3$

$9+2 + 13\times3 + 14=63$ 

(5, 5), Section 3.3 $U_{5,5}$.
11 right-moving rules, 3 halting rules,
11 left-moving rules. $|\Gamma|=5$

$15+2 + 11\times5 + 14=86$ 

Alternative representation, Section~\ref{alternative-pairs}:
$5\times(2\times 5+3)+3=68$

(6, 4), Section 3.4 $U_{6,4}$.
13 right-moving rules, 1 halting rule,
10 left-moving rules. $|\Gamma|=4$

$12+2 + 10\times4 + 14=68$ 

Section 3.5 $U_{15,2}$.
originally, 14 right-moving rules, 1 halting rule,
15 left-moving rules. $|\Gamma|=2$.
After swapping:
15 right-moving rules, 1 halting rule,
14 left-moving rules. $|\Gamma|=2$.

$2\times 3+3 + 14\times2 + 16=53$ (sic!)

Section 3.2 $U_{9,3}$.
13 right-moving rules, 1 halting rule,
13 left-moving rules. $|\Gamma|=3$

$9+2 + 13\times3 + 14=64$

Alternative representation, Section~\ref{alternative-pairs}:
$3\times(2\times 9+3)+3=66$.

We always have the freedom to swap
left-moving with right-moving rules,
by flipping the Turing machine's tape and starting over the rightmost
input character!

Universal Turing Machines
for which the Halting Problem is Undecidable.

\paragraph{Watanabe.}

Shigeru Watanabe, 4-symbol 5-state universal Turing machine,
Information Processing Society of Japan Magazine
13 (9) (1972) 588--592

also another one: 3 symbols, 7 states.

Suppose the 21 rules are balanced between left and right
Section 3.2 $U_{9,3}$.
10 right-moving rules, 1 halting rule,
10 left-moving rules. $|\Gamma|=3$

$9+3 + 10\times3 + 11=53$ ONE smaller.

\paragraph{Wolfram--Cook, Rule 110.}

\cite[Fig.~1]{cook2004}

5 symbols, 2 states, 4 L, 6 R

4 symbols, 3 states, 2 L, 6 R, 4 transitions unused

3 symbols, 4 states, 3 L, 7 R, 2 unused

2 symbols, 7 states, 3 L, 9 R, 2 unused

[8] Thanks to David Eppstein for figuring out that the two-symbol machine
can be achieved with only seven states. (personal communication, 1998)

\cite[{Note [7]}, p.~38]{cook2004}

depends (in a complicated way) on the tag system that is simulated,
and thus it is not useful for our purpose.

questions about their behavior,
such as ``Will this sequence of symbols ever appear on the tape?'', are
undecidable.

See also \cite{small-weakly-universal-2009},
before section 3.1: ``... do not halt. ...
... such that it [the glider] is generated by Rule 110 if and only if
the simulated machine halts.
...
The glider would be encoded on
our tape as a unique, constant word.''

The repeating patterns by which the empty parts of the 
tape are filled are not fixed, but depend on the tag system,
in a complicated way, see
\cite[Note on initial conditions, p.~1116]
{wolfram-NKS}\footnote
{
  \url{https://www.wolframscience.com/nks/notes-11-8--initial-conditions-for-rule-110/}}.
HOWEVER, according to
\cite{small-weakly-universal-2009},
p.5. ``Both of these repeated words are independent of the input.''
(of the cyclic tag system).
Is there perhaps something like a UNIVERSAL cyclic tag system!!

``the final initial conditions consist of an infinite repetition of b1
blocks, followed by b2, followed by an infinite repetition of b3
blocks. The b1
blocks act like `clock pulses', b2
encodes the initial conditions for the tag system and the b3
blocks encode the rules for the tag system.''

\paragraph{Margenstern and Pavlotskaya.}

Margenstern and Pavlotskaya [9] gave a
2-state, 3-symbol Turing machine that is universal when coupled with a finite
automaton. This machine uses only 5 instructions.

9. Maurice Margenstern and Liudmila Pavlotskaya. On the optimal number
of instructions for universality of Turing machines connected with a
finite automaton.
International Journal of Algebra and Computation, 13(2):133--202, April 2003.

\url{https://www.semanticscholar.org/paper/On-the-Optimal-Number-of-Instructions-for-Universal-Margenstern-Pavlotska%C3%AFa/95ab0fc3122aa196187bff9610b8e71f2a33ac6e}

\paragraph{Wolfram's 2,3 Turing Machine}

\paragraph{On universality.}

Intermediate Turing degrees.

\url{https://en.wikipedia.org/wiki/Universal_Turing_machine}


\subsection{Alternative representation}
\label{alternative-pairs}
IDEA for an alternative representation of a TM configuration as a
part of a string rewriting process:
The ``active cell'' can be \emph{either} to the left \emph{or} to the
right of the state symbol. Each state has two corresponding state
symbols
$q^L$  and $q^R$.
$q^L$ is used after a left move, and $q^R$ is used after a right move.

\begin{itemize}
\item
  For each \emph{right-moving rule} $(q,s,s',R,q')$, the pairs
  $(q_Rs,s'q'_R)$ and
  $(sq_L,s'q'_R)$.
\item
  For each \emph{left-moving rule} $(q,s,s',L,q')$, the pairs
  $(q_Rs,q'_Ls')$ and
  $(sq_L,q'_Ls')$.
\item For each \emph{halting rule} $(q,s,-)$, the pairs
  $(q_Rs,H)$
and  $(sq_L,H)$.
\end{itemize}
In total, these are $3|\Gamma|+3$ pairs,
plus 2 pairs
for every rule (no matter whether left-moving
right-moving, or halting),
plus the starting pair.

This may be advantageous if $|\Gamma|\ge 4$ (2 pairs per rule versus
1 or $|\Gamma|$ pairs per rule).
There are $|\Gamma|\cdot |Q|$ rules,
so in total
 $|\Gamma|\cdot (2|Q|+3) + 3$ pairs
plus the starting pair.
(There might be unused transitions.
This is not taken into account.)

\section{Further literature}

  Another citation, from \emph{Leviathan}, p.~115, Part I, Chapter V:
  ``For who is so stupid as both to mistake in Geometry and also to
  persist in it, when another detects his error to him?''

Stephen C.  Kleene, Mathematical Logic, 1967, pp.~3--4.
  ``To any student who is not ready to do so, we suggest that he close
  the book now, and pick up some other subject instead, such as
  acrostics or beekeeping.''
\footnote{namely, ready to keep in mind the distinction between the
  object logic and the observer's logic.}
\subsection{Turlough Neary}

Small universal Turing machines.
\url{https://www.ini.uzh.ch/~tneary/tneary_Thesis.pdf}
A thesis submitted for the
degree of Doctor of Philosophy
Department of Computer Science
National University of Ireland, Maynooth

Supervisors: Dr. Damien Woods and Dr. J. Paul Gibson

External Examiner: Prof. Maurice Margenstern


October 2008

Turlough Neary

Institute for Neuroinformatics,
University of Zürich and ETH Zürich,
Switzerland.
tneary@ini.phys.ethz.ch

\url{https://www.ini.uzh.ch/~tneary/}

I am a research scientist in the group of Matthew Cook. 

now alumnus (PostDoc)

Last entry 2019

STACS 2015: The PCP is undecidable already with as few as five
 (word) word pairs~\cite{neary:PCP5:2015}:
``Using Cook’s [6] reduction of tag systems to cyclic tag systems, it is a straightforward matter
to give a binary cyclic tag system program that is universal and contains only two 1 symbols. We also use our binary tag system construction to improve the bound for the number
of pairs of words for which the Post correspondence problem [18] is undecidable, and the
bounds for the simplest sets of matrices for which the mortality
problem [16] is undecidable.''

What is a \emph{binary cyclic tag system program}? (top of p.650)

``By applying the reductions in [8] and [3] to P in Theorem 11 we get
Corollary 12.

 Corollary 12. The matrix mortality problem is undecidable for sets with six 3 × 3 matrices
and for sets with two 18 × 18 matrices.''

\subsubsection{P-completeness of cellular automaton {Rule} 110 (2006)}
\cite{neary-woods-rule110-2006}
\footnote{see also Technical Report 04/2006, Boole Centre for Research in Informatics,
  University College Cork, Ireland, 2006, \url{https://dna.hamilton.ie/assets/dw/NearyWoodsBCRI-04-06.pdf}.
  See also Chapter~4 of Neary's thesis for a 
  more direct approach with reduced time overhead $O(|Q|t^2\log t)$.}
Some remarks and thoughts about this paper, although it is not presently relevant.
The main technical result is that Turing machines can be simulated
by cyclic tag systems with polynomial overhead.
\begin{itemize}
\item ``As with standard Turing
machines, using a binary alphabet causes at most a constant factor increase in
the time, space and number of states.'' (Referring to clockwise Turing
machines, end of Section~2).

This is a bit tricky.
One way to do it is as follows: If a larger alphabet is simulated by $b$-bit
blocks,
a consecutive block of $b-1$ bits around the current tape head is
not represented
on the binary tape, but is remembered as part of
 the state. (When the tape is extended by one letter of the original alphabet, additional bits
 are remembered and gradually output to the circular tape.)

 (The version of the Ph.D. thesis goes directly from Turing machines to cyclic tag
 systems. There, going to a binary alphabet is no problem. Apparently
 the decision whether to move left or right is made as part of the
 state? (As in Cocke \&
Minsky 1964, \url{https://doi.org/10.1145/321203.321206}).
 Or I rather believe, the subscript $L$ or $R$ as part of the encoded state only remembers
 the previous movement. This is not so clear. Probably the factor 30
 for the right-moving rules
in \cite{neary-woods-rule110-2006}
 must everywhere simply be changed to 60 to work for the
 version in the thesis.)


  
\item Halting currently leads to a tape that remains unchanged in a cyclic
  loop.
  As an alternative, it could be arranged that only empty appendants
  are used, and the cyclic tag system comes to a stop. (If the TM
  output is not interesting and only halting should be decided.)  
\end{itemize}

\subsubsection{Small PCP with 5 words, translation from the binary tag
system}
\label{sec:small-pcp}

According to  Neary~\cite{neary:PCP5:2015}.
The tag system is over the alphabet $\{\mathtt{b},\mathtt{c}\}$.

Codes for \texttt{b} and \texttt{c}:
\begin{align*}
  \langle \mathtt{b} \rangle &= \mathtt{10^\beta 1}\\
  \langle \mathtt{c} \rangle &= \mathtt{1}
\end{align*}
Every dataword ends with \texttt{b}.
A dataword $x_0x_1 \ldots x_q\mathtt{b} \in
\{\mathtt{b},\mathtt{c}\}^*\mathtt{b}$
is encoded as
\begin{itemize}
\item $(r,r\langle x_0\rangle \langle x_1\rangle
  \ldots
  \langle x_q\rangle \mathtt{10}$), before the $\beta-1$ deletions
\item $(r,r\langle x_0\rangle \langle x_1\rangle
  \ldots
  \langle x_q\rangle \mathtt{10^\beta}$), after the $\beta-1$
  deletions, ready for the next appendant (looking at $x_0$).
\end{itemize}

\def\wordpair#1#2{$\vcenter{\hbox{\strut$\mathtt{#1}$}
     \hrule\hbox{\strut$\mathtt{#2}$}}$}


A word pair $(v,w) = \hbox{\wordpair {\hbox{$v$}}{\hbox{$w$}}}$
matches the existing string with the upper component
($v$) and extends the string built so far by the lower
component~($w$). In other words, $v$ is removed from the left and $w$
is appended.

\medskip

\begin{tabular}{|l|l||l|l|}
  \hline &tag system appendant
  &appendant pair& deletion pair\\
  \hline
  $\langle \mathtt{b} \rangle =\mathtt{10^\beta 1}$
& $\texttt{b}\to \mathtt{b}$
         &\wordpair{10^\beta 1}{110}
            &\wordpair{10^\beta 1}0 \\
  \hline
  $\langle \mathtt{c} \rangle =\mathtt{1}$
& $\texttt{c}\to u_0u_1 \ldots u_l\mathtt{b}$
  &\wordpair{1}{1
                                  \hbox{$\langle u_0\rangle \langle u_1\rangle
  \ldots
  \langle u_l\rangle$}
10}
            &\wordpair{1}0 \\
  \hline
\end{tabular}

The final pair is \wordpair {10^\beta1^\ell}{1^\ell},
with
$\ell$ large enough so that $\mathtt{1^\ell}$ can never be created on
the bottom by the other pairs.

Observations:
\begin{itemize}
\item A string that wants to have a chance to be matched must be a word in $\{\mathtt{10^\beta
    1},\mathtt{1}\}^*$ or a prefix of such a word.
\item Any string that is a prefix of a word in $\{\mathtt{10^\beta
    1},\mathtt{1}\}^*$ can be uniquely factored
  into pieces
  $  
  \mathtt{10^\beta 1}$ and
  $  
  \mathtt{1}$, plus possibly a last
incomplete piece that is a prefix of
$  
\mathtt{10^\beta 1}$.
\item A \texttt{0} can be matched only as part of $\mathtt{10^\beta1}$.
\item Consequence: after creating \texttt{10}, there must be
  exactly $\beta-1$
  deletions in order to create the required \texttt{0}s.
  The next letter created at the bottom must be \texttt{1}, and therefore we
  must use an appendant pair.
\end{itemize}

In the tag system, the length is preserved modulo $\beta-1$.
This means that $(l+1)-(\beta-1)$ should be a multiple of
$\beta-1$. Consequently: $(\beta-1)|(\ell + 1)$.
This initial length is congruent to 1 mod $\beta$.
It follows that, at termination, the length must be 1.
Since the last letter is always $\mathtt{b}$, the word must be a
single $\mathtt{b}$ on termination.

\subsection{Wolfram, A New Kind of Science}

no termination mentioned (Chap. 11, around p. 690)

p. 667 tag systems of a certain kind are universal.

p.704 7 states 4 colors

``halting'' is somehow treated as an exotic feature of Turing
machines

[ Actually, the ``good'' Turing machines of Turing 1936 don't halt,
(``non-circular''), except by mistake. ]
Review in
\url
{https://www.cambridge.org/core/journals/journal-of-symbolic-logic/article/abs/a-m-turing-on-computable-numbers-with-an-application-to-the-entscheidungs-problcm-proceedings-of-the-london-mathematical-society-2-s-vol-42-19361937-pp-230265/4DFCA89035F7F7C5BF4DB5129B8BB09E}

2 states 5 colors $\to$ rule 110

p.94 tag system CAN halt: case (c): all elements are eventually removed.

So-called ``register machines'' p.97 (Chapter 3) are really 2-counter
machines.
German ``Registermaschine'' is something else, namely RAM
(random-access-machine).
But Minsky also says ``Universal
Program machines with Two Registers''.

See also notes:
p.~1115 History

p.~1116
. (The initial conditions in the note below quite soon become too
large to run explicitly on any existing computer.)

The discussion in the main text and the construction above require a
cyclic tag system with blocks that are a multiple of 6 long, and in
which at least one block is added at some point in each complete
cycle. By inserting k = 6 Ceiling[Length[subs]/6] in the definition of
TS1ToCT from page 1113 one can construct a cyclic tag system of this
kind to emulate any one-element-dependence tag system.

\subsection{On universality}

It seems that such a system can
 simulate a machine that resembles a Turing machine except that it
misses the most important state: a halting state.
 (Another, less crucial difference is
 initializing the ``empty'' parts of the ``tape'' in both sides of the
 input not with
blanks, but with some other periodic pattern.)
 
Unfortunately, the author never clearly define the concept of
universality that he uses.
.. The same vagueness and confusion ... permeates the literature in which

Wikipedia.
...

A basic fact about a universal Turing machine in the classic sense,
is that the halting problem is undecidable. Which problem
.. about

Cook
\cite[p.~11]{cook2004}
Universality in Elementary Cellular Automata

Cook states some undecidable problem

...
any program in just the tape data portion of the initial arrangement,
while always using a fixed repeating pattern on the right and on the
left.
To find a specific example of an undecidable question in such a glider
system, we can easily suppose that one of the Turing machine states
represented in the tag system has the peculiar property that all of its H,
L, and R symbols lead to empty appendants for the end of the tape.
Then, if this state is ever entered, the tag system will stop appending
symbols to the tape, and the glider system will stop sending moving data
to the end of the tape, causing ossifiers to hit tape data after all. This
can lead to a new kind of glider being produced. So one undecidable
question is, “Will the following glider ever appear?''

We could also easily suppose that one of the Turing machine states
represented in the tag system has the property that it halves both the L
and R symbols of the tape, and always stays in the same state. In this
case, the tape will eventually be zeroed out, and the tag system's behavior
will become periodic with period one, and the behavior of the glider
system will likewise become periodic, with a specific predeterminable
(space, time) period. So another example of an undecidable question is,
“Will the behavior become periodic with the following period?''

If the glider system becomes periodic at all, then the emulated tag
system must be periodic as well, meaning that the Turing machine is in
an infinite loop where it keeps entering the same state, with the same
tape to the left and right. Conversely, if the Turing machine enters
such a loop, then the glider system must become periodic. Since it is
undecidable whether a Turing machine will enter such a loop, another
example of an undecidable question for the glider system is, “Will the
behavior become periodic at all?”

%
   ``Rechnender Raum'', Elektronische Datenverarbeitung, vol. 8, pages
   336--344, 1967
\\
   \url{http://www.mathrix.org/zenil/ZuseCalculatingSpace-GermanZenil.pdf}
   \\
   \url{https://www.informationphilosopher.com/solutions/scientists/zuse/Rechnender_Raum.pdf} 9pp.
   \\
  \url{https://link.springer.com/book/10.1007/978-3-663-02723-2}\\
\url{http://www.horst-zuse.homepage.t-online.de/rechnender-raum.html}

\subsection{Tag systems and universality}

However, for our purposes, this would require
some notion of universality for an individual tag system,
and such a notion does not seem to exist?
seem to have been studied.

[It exists for cellular automata.  see Wolfram, NKS: universal in a
certain sense, it can simulate any cellular automaton, with proper
setup.]

One could look at universal tag systems, with a fixed set of rules,
but an arbitrary starting word. (e.g. a tag system simulating
a UTM? add another round of the detour!)

In particular, a 2-tag system can be constructed to emulate a
Universal Turing machine, as was done by Wang 1963 and by Cocke \&
Minsky 1964.
\url{https://en.wikipedia.org/wiki/Tag_system}

Cocke, John; Minsky, Marvin (1964). "Universality of Tag Systems with
P=2". J. Assoc. Comput. Mach. 11:
15–20. doi:10.1145/321203.321206. hdl:1721.1/6107. S2CID 2799125.
\url{https://doi.org/10.1145/321203.321206}

deletion number P = 2.

Wang, H.: "Tag Systems and Lag Systems", Math. Annalen 152, 65–74, 1963.
\url{https://doi.org/10.1007/BF01343730}

Common review of both papers in

The Journal of Symbolic Logic , Volume 36 , Issue 2 , June 1971 , pp. 344
DOI: \url{https://doi.org/10.2307/2270314}

The \emph{lag systems} of Wang would be good enough for us. They look at the
whole prefix of length $p$ (the deletion number), not just at the
first letter.
Or even \emph{Monogenic Normal Systems} ([3]). They remove
prefixes of varying length. [3] POST, E. L.: Formal reduction of the general combinatorial decision problem. Am.
J. Math. 65, 197--215 (1943).

Here is a more efficient simulation (in terms of time, not with number
of rules):
Neary \& Woods, ``{P-completeness of cellular automaton {Rule} 110}'',
\cite{neary-woods-rule110-2006}  (2006)
\footnote{see also Technical Report 04/2006, Boole Centre for Research in Informatics,
  University College Cork, Ireland, 2006, \url{https://dna.hamilton.ie/assets/dw/NearyWoodsBCRI-04-06.pdf}.
  See also Chapter~4 of Neary's thesis for a 
  more direct approach with reduced time overhead $O(|Q|t^2\log t)$.}
showed that a Turing machine can be simulated
by a cyclic tag systems with polynomial overhead.
Some other paper of Neary \& Woods showed that
 cyclic tag systems can be simulated by 2-tag systems. THERE WE GO!

\subsection{Marvin Minsky}

[ reference ???
This (WHAT?) is proved in Minsky's book (Computation, 1967, p. 255--258),
Marvin Minsky (1967). Computation: Finite and Infinite Machines (1st
ed.). Englewood Cliffs, N. J.: Prentice-Hall, Inc. In particular see
chapter 11: Models Similar to Digital Computers and chapter 14: Very
Simple Bases for Computability. In the former chapter he defines
``Program machines'' and in the later chapter he discusses ``Universal
Program machines with Two Registers'' and ``...with one register'', etc.
]

\hrule

A 2-way automaton can move forward and backward on the input tape, and
thus
process the input as often as it wants.
Such an algorithm can boost the correctness probability, and thus it
is possible to make a correct decision (accept or reject) with
probability at least $1-\eps$, for any $\eps$. (By contrast, for a (1-way)
probabilistic finite
automaton this is impossible: If there is a gap interval $(p_1,p_2)$
of positive length such that the acceptance probability
is never in this interval,
then, for every cutpoint $\lambda$ in this interval,
the accepted language is regular. Michael O. Rabin
1963,
Probabilistic Automata,
Information and Control 6 230--245. \cite{rabin1963}
\cite[Theorem~IIIB.2.3]{paz71}

\subsection{Blondel and Tsitsiklis (2000)}

Blondel and Tsitsiklis~\cite{blondel2000} 2000:
PFA EMPTINESS is undecidable cannot be found in its
entirety in the published literature. A proof (stated
with a different terminology) is given in Theorem 6:17
in p. 190 of [17]. The proof given there is a few lines
long and refers to a long cascade of lemma that appear
at various places in the book. A full proof is hard to
reconstruct.
Sketches of an alternative proof can be found in
several recent references; see, e.g., [6,14,15]. Finally,
a full proof can be found in the expanded version of
[6]; see Theorem 3:2 in [7].

[14] O. Madani, On the computability of infinite-horizon
partially observable Markov decision processes, AAAI98 Fall
Symposium on Planning with POMDPs, Orlando, FL, 1998.

[15] O. Madani, S. Hanks, A. Condon, On the undecidability of
probabilistic planning and infinite-horizon partially observable
Markov decision problems, Proceedings of the Sixteenth
National Conference on Artificial Intelligence, Orlando, FL,
July 1999.

\subsection{Madani, Hanks, Condon 2003}
\label{sec:madani-hanks-condon}

What DID Madani, Hanks, Condon \cite[Sec.~3.1 
and Appendix A
]{jair03} in 2003.
 prove?

\subsection{Others}

Rūsiņš Mārtiņš Freivalds

?? RELATION TO
sequential test (A. Wald)

\subsection{Literature about stochastic automata and languages}

Rabin~\cite[Sections~IX--XII, p.~242--245]{rabin1963}
calls PFAs
in which
all transition probabilities are positive
\emph{actual automata} and
studies their \emph{very special properties}.

WHAT ARE THOSE PROPERTIES? Well, not so special.
for ISOLATED cutpoints, only very special languages.

\hrule
\cite[p.261]{NasuHonda1969}

//Lemma 11: An E-event ($\phi(a)=\psi(a)$) is a P-event (stochastic
language).

...

``Matuura et al. and Turakainen studied linear space automata (or
generalized automata) and found independently the same fact that the
family of stochastic languages is the same as the family of languages ac-
cepted by linear space automata (or generalized automata). Matuura
et al. gave us a key idea with respect to Lemma 11 in this viewpoint.
P-sets, E-sets and D-sets can be studied generally in this viewpoint, and
E-sets and D-sets can be related to the work of Schüzenberger (sic!) (1961).
The next remark and Lemma 12 follows from Schüzenberger's results
but as for Lemma 12 we will give a straightforward proof.''

//Remark: neither E-set nor D-set
 ($\phi(a)\ne\psi(a)$)
contains the other.

LEMMA 12. ...

What are ``linear space automata''? With arbitrary matrices
$M$, working on a \emph{linear vector space} with linear operators
\cite{schutzenberger-1961}.

``(c) The unbounded part of the memory, $V_N$, is the finite dimensional
vector space of the vectors with $N$ integral coordinates; this part of the
memory plays only a passive role and all the control of the automaton
is performed by the finite part.''

Acceptance is by NON-containment in a Subspace. (or union of several subspaces)

\hrule
\cite[p.270]{NasuHonda1969}

// Thm.21 is the undecidability of PFA Emptiness.

footnote 6:
As for Theorem 21, it reduces to the statement in p.~150 of Schüzenberger
(1962) from the viewpoint that we said before Lemma 12 of this paper.

footnote $\to$ Sch\"utzenberger 1961/63:
Which of this is in Paz? only Sch 1961 (I\&C)
\cite{schutzenberger-1961}

``... grateful to Dr. Y. Inagaki and Mr. H. Matuuta (sic!) in Nagoya
University ...''

MATUURA,H., INAGAKI, Y., AND HUKUMURA, T. (1968), Generalization of finite
automata and its analysis. Papers of Technical Group on Automata. IECE,
Japan. 1968---1 (in Japanese).

MATUURA, H., INAGAKI, Y., AND HUKUMURA,T. (1968), Linear space automata
and probabilistic automata. Records of 1968 National Convention. IECE,
Japan. S 8--4. (in Japanese).

TURAKAINEN, P. (1968), On probabilistic automata and their generalizations.
Annales Academiae Scientiarum Fenicae, Series A I Mathematica 429.

SCHÜTZENBERGER, M. P. (1961), On the definitions of a family of automata, Inform.
Control 4, 245-270.
\cite{schutzenberger-1961}

SCHÜTZENBERGER, M. P. (1962), Certain elementary families of automata. In "Mathematical Theory of Automata." (J. Fox, ed.) pp. 139-153. Polytechnic Press of
the Polytechnic Institute of Brooklyn, 1963.

Marcel Schützenberger~\cite{schutzenberger63:certain}
\footnote{\url{https://monge.univ-mlv.fr/~berstel/Mps/Travaux/A/A/1963-4ElementaryFamAutomataSympThAut.pdf}}

review of Schützenberger's papers by Michael O. Rabin in
The Journal of Symbolic Logic , Volume 34 , Issue 2 , 25 July 1969 , pp. 296 - 297
DOI: \url{https://doi.org/10.2307/2271115}

``The third paper is mainly expository and is based on the first two papers.''

M. P. SCHUTZENBERGER. Certain elementary families of automata. Proceedings of the
Symposium of Mathematical Theory of Automata, New York, N.Y., Microwave Research
Symposia series vol. 12, Polytechnic Press of the Polytechnic Institute of Brooklyn, New York
1963, pp. 139-153.

\hrule
closed under complement,

This is
credited to ``(Paz)'' in \cite{NasuHonda1969}

cites 2 papers of Paz

Some aspects of probabilistic automata~\cite{paz-1966} 1966:\\
Paz, A. (1967), Fuzzy star functions, probabilistic automata and their approximation by nonprobabilistic automata. IEEE Conference Record, 1967 Eighth
    Annual Symposium on Switching and Automata Theory, 280-290.

also in the other, 1968 Nasu-Honda paper, p.288:
``THEOREM 4.1 (due to Paz).'' (closed under complement)

BUT THIS REFERS to probabilistic event == some fuzzy event.
complementation means just complementing the acceptance
probabilities. (trivial)

cites ONLY \cite{paz-1966} 1966, and gives an independent proof.

\hrule

C. V. Page~\cite{page-1966}:
 
$F$: output vector, a $n$-component column vector whose entries are
\emph{real} numbers. $F_i$ is the output from state $S_i$

arbitrary output values 
in the interval $[0,1]$
could be accommodated 
by introducing an additional accepting state $C$.
Upon reading a special termination symbol, the PFA makes one probabilistic
step to $C$ or the rejecting state $B$,
reflecting the
output values $\eeta_q$. From $C$, the PFA can only go to the rejecting state $B$.
This ensures that
the corresponding transition matrix can be used only once, at the end.

However, as this would introduce an additional state, this
cannot be used to improve
Theorem~\ref{thm-fixed-f}
for the ``classic'' setting.

Generalized automata and stochastic languages,
    by Paavo Turakainen,
    Proc. Amer. Math. Soc. 21 (1969), 303-309.
    DOI: \url{https://doi.org/10.1090/S0002-9939-1969-0242596-1}

    ``the initial vector'',  ``the final vector'' $f$, cutpoint = $\eeta$!

    Turakainen~\cite
{turakainen69}
    
    arbitrary positive or negative matrices, and $\pi$ and $\eeta$.
    Does not define a more general class of languages.
    \cite[\S 3.3.2, Der Satz von  Turakainen, p.120]{claus71}.

INFORMATION AND CONTROL 10, 215--219 (1967),
On m-Adic Probabilistic Automata,
ARTO SALOMAA,
proves regular IFF $\lambda\in\mathbb Q$,
cites Paz~\cite{paz-1966} about
$m$-adic automata.

Paz~\cite{paz-1966} 1966:\\
E X A M P L E S (p.35)

\emph{1. An ``$m$-adic two state p.a.'' is obtained by prescribing the following
conditions:} ... use all $m$ different matrices.

Rabin~\cite{rabin1963} already defined THE binary automaton, with 2
states and 2 transitions.

``The following matrices were suggested by E. F. Moore.''

Rabin~\cite{rabin1963}:
IX. ACTUAL AUTOMATA, p.242.
all transition probabilities are POSITIVE.

\paragraph{Adding a small probability to change $\ge\lambda$
into $>\lambda$.}

\medskip\hrule
The idea of adding a small probability to change $\ge\lambda$
into $>\lambda$ is can also be used to prove that
the languages recognized by
rational PFAs are closed under complement.
But: Sometimes proved with the techniques of
???Turakainen~
\cite{turakainen69}?
Paz cites other Turakainen papers! see below

 This is
crucial for for
CLOSED UNDER COMPLEMENT for rational PFAs, where is this in the literature?
[ Closure properties must have been well-known.. ? In particular for the
complement?
Turakainen~
\cite{turakainen69}?
No the complement is not mentioned
in 
\cite{turakainen69}, only the mirror.
    And moreover that appeared in the same year. Who is credited in
    Paz:
Paz uses integral pseudo-stuff for complementation. cites
Paavo Turakainen
[1969b Ann. Fenn. rational probabilistic]
\url{https://www.acadsci.fi/mathematica/1969/no429pp01-53.pdf}
51 pp.
ANNALES ACADEMIAE SCIENTIARUM FENNICAE
Series A
I. MATHEMATICA
429
On probabilistic automata and their
generalizations.
\url{https://doi.org/10.5186/aasfm.1969.429}

--53.

(53 pp.), No 429 (1968)

Claus
\cite[Hilfssatz 31.ii, p.~131]{claus71} is closed under complement
for $P$-sets. Proof as an exercise; for the \emph{rational} generalized
acceptors of Turakainen~\cite{turakainen69}.

\paragraph{The proof in Paz 1971.}

What the heck is $a$ and $b$ in 
\cite[p.~188, line~1]{paz71}?
What is $m$?
Is this erroneously copied from a different part of the proof? (Lemma
6.12)?
A similar line is
\cite[p.~188, line~$-6$]{paz71}?

disfigured, defaced, spoiled, bungled, daubed (beschmieren), botched
(Flickwerk, Pfuscherei)

defaced by occasional typos and editing errors (such as passages in a
proof that
obviously belong to another proof)

\cite[Lemma 14, pp.~264]{NasuHonda1969}
requires PFA for which
the acceptance probability is an injective mapping
from $\Sigma^*$, ``for example   .. in the proof of Proposition 13'':
These are the $m$-ary automata.

Paz (last page 193) gives precise credit for Section 6.
Section 6 up to Theorem 6.8 is from
Turakainen
[1969b Ann. Fenn. rational probabilistic]
\url{https://www.acadsci.fi/mathematica/1969/no429pp01-53.pdf}
51 pp.
(NOT
~\cite
{turakainen69}!) and 1970a (Akad. Turku, not closed under ...) and
1970b (tech.rep. ``some closure properties''), using the language of
Nasu and Honda as a starting point.
Everything from Exercise 6.9 to the end is from
Nasu and Honda \cite{NasuHonda1969}

Turakainen
[1969b Ann. Fenn. rational probabilistic]
\url{https://www.acadsci.fi/mathematica/1969/no429pp01-53.pdf}
Theorem 15 (p. 37 is about complement but ONLY when $L_=$ is regular!).

Theorem 16 on p.38 is about reducing to 0-1 $f$.
equation (9.1) on p. 39. the number of states is squared as in Claus.

Mention Bukharaev and Starke?

\paragraph{The proof in Claus 1971.}

Volker Claus from 1971 \cite[Satz 28, p.~157]{claus71},

Claus says: proved in Nasu--Honda (1969) for a fixed alphabet of size
$\ge 2$
\emph{without being able to bound the number of states}
, ``wobei nun aber die Zustandszahl der SAkz nicht beschränkt ist.''
  Nasu--Honda is a different proof after all!

Hilfssatz 33, S.~131.
closed under complement if rational.
Exercise, using generalized acceptors VAkz (S. 119).
This seems to be also how it is done in Paz.

$L_{\phi=\psi}$: p. 155, using the formula (and some lemmas
about product automata) (very nice notation!)

Using the PCP, encode with \emph{ternary} automaton, and alphabet
$\{\texttt{1},
\texttt{2}\}$, avoiding the trailing zeros issue.

\paragraph{Terminology.}
``realized'' is used only for fuzzy events.
realizable fuzzy events = probabilistic events.

\emph{represented} (Salomaa, Theory of automata, p.76) is indeed used for P-sets.

\paragraph{Re. ``events''.}
I find this choice of terminology, which pervades much of the literature,
(in particular the whole Chapter~III of the book \cite{paz71})
unfortunate since an event is rather something whose probability is
to be measured.

Turakainen~\cite{turakainen69} speaks only
of \emph{stochastic languages}. (+)

Claus
\cite[Definition 39, p.~148]{claus71}
``unbestimmtes Ereignis'' ``= translation of fuzzy set'' according to
Zadeh, actually = fuzzy \emph{language}.
Special cases:
\cite[Definition 40, p.~148]{claus71}:
``stochastisches Ereignis'' = acceptance probability of a PFA,
``reguläres Ereignis'' = (characteristic function of a) regular language.
But at least stochastic \emph{languages} = stoch.\ Sprachen:
\cite[Definition 26--27, p.~101]{claus71}

\cite[Definition 1, p.251]{NasuHonda1969}: \emph{fuzzy event}. (from
$\Sigma^*$.
``Subsets of $\Sigma^*$ will sometimes be called events or languages.''

Rabin~\cite[p.~233]{rabin1963}:
Subsets of $\Sigma^*$ (i.e., sets of tapes) will sometimes be called
\emph{events}. ...  \emph{regular event}.
Rabin~\cite[Definition 11, p.~243]{rabin1963}:
definite event.

\subsubsection{Turakainen 1975. Conversion from pseudo-stochastic automata to
stochastic automata}
\label{sec:Turakainen}

Turakainen~\cite
{Turakainen_1975}.

Theorem 1:
\begin{compactenum}[(i)]
\item
\label{tura-i}
  $n+2$ states, doubly stochastic, positive (``actual''), one
  final state
\item $n+1$ states, stochastic, positive, \emph{difference} between
  two starting distributions.
\item
\label{tura+3}
  $n+3$ states, doubly stochastic, positive, one
  final state, one deterministic initial state (equal to the final
  state).
  
(And no
  extra linear factor besides the exponential in the length of the input.)

\item (a specialization of (\ref{tura-i}) under some conditions: no
  extra linear factor besides the exponential in the length of the input)
\end{compactenum}
The empty word is always disregarded.

There are some results about what canNOT be done with $n+1$
states (Theorem 3): (\ref{tura+3}) cannot be achieved with $n+1$.

There is also something about characteristic polynomials and minimal automata.
\subsubsection{Newer literature about PFAs}
\label{sec:newer}

according to
2020 \cite 
[p.~22:2]{bell_et_al:LIPIcs.CONCUR.2020.22},
the current record for the number of states is dimension 25
from
Hirvensalo 2007,
\cite{hirvensalo07}, see also \cite{tHirvensalo06a},
but for \emph{binary} input alphabet.

Hirvensalo 2007 already uses ``Claus instance of PCP'' where first and last
word pair is fixed. (in any minal solution, thus w.l.o.g.) Contrast:
as part of the requirements on a solution.

There is a mistake in
Hirvensalo \cite[Section~3, Step~3]{hirvensalo07}.
Due to the shortcut for state $q_6$, it is not true
that ``$\mathbf{x}_3^T C_{w^R} \mathbf{y}_3=0$, if $w\in \Sigma^+$ is
not in the image of $\psi$.''
The value of
$\mathbf{x}_3^T C_{w^R} \mathbf{y}_3=0$, if $w\in \Sigma^+$ is
something not very well-defined in which some part of the vector is
ignored by 
but the contribution of state $q_6$ is counted.
To fix this, in the sense of
Section~\ref{sec:thm-variable-matrices}, Step 2,
the vector $\mathbf{x}_3 = f'$ should be
\begin{math} 
f' =
\begin{pmatrix}
  \hat f\\
  \hat f\\
  \hat f\\
  f_6
\end{pmatrix}
\text{ and not }
f' =
\begin{pmatrix}
  \hat f\\
0\\0\\
  f_6
\end{pmatrix}
\end{math}
as written in the paper: ``such that all its coordinates are zero,
except ...''
(Remember that
Hirvensalo's matrices are ``column-stochastic'' and read from right to
left, see the remark following equation~(1).)
Then the method commits the ``inconsequential error'' of
Section~\ref{sec:thm-variable-matrices}, Step 2, but in a
consistent way.

\paragraph{Falscher Ansatz, Step 1.}
We convert  $\alpha$ and $\beta$ to unit vectors,
inspired by
Blondel and Canterini \cite[Step 3a, p.~236]{blondel-canterini-03}.

For this, we need an invertible matrix $R \in \mathbb Q^{6\times 6}$
such
that
\begin{equation}
  \label{eq:specify-R}
R\beta_2=\beta_3\text{ and }R\alpha_3=\alpha_2,  
\end{equation}
where $\alpha_3$ and
$\beta_3$ are unit vectors.

QUATSCH! Was wir brauchen ist eine orthogonale transformation
verbunden mit einer Skalierung.

Skalierung unproblematisch. $>0$ or $\ge 0$ unaffected.

\subsubsection{Fijalkow-2017-SIGLOG}
\cite{Fijalkow-2017-SIGLOG},
``Undecidability Results for Probabilistic Automata''

Section 3, p.
Invalid argument to get strict inequality acceptance:
``However, since the Equality Problem is undecidable, this implies that the emptiness
problem with strict inequalities is also undecidable.''

\subsubsection{OLD: Bertoni 1974/75}
\label{sec:OLD-bertoni}

\cite{Bertoni1975},
Alberto Bertoni 1975 (Proceedings of a conference 1974 in Berlin):
It is undecidable whether 1/2 is an isolated cutpoint.

reduction from MPCP:
1/2 is an isolated cutpoint $\iff$ the Turing maching does \emph{not}
halt. ($\iff$ the PCP has \emph{no} finite solution.)

builds $\phi(x)+\frac12(1-\psi(x))= \frac12 + (\phi(x)-\psi(x))$
(with 4 states!).

Other points:

\begin{enumerate}
\item 
realizes $\phi(a)$ not by the binary automatom but a different 2x2
matrix.
Always 2x2, independent of the radix $m$. middle of p.110:
\begin{displaymath}
  A(u) :=
  \begin{pmatrix}
    \frac{1}{m^{|u|+1}} &
    1-\frac{1}{m^{|u|+1}} \\[5pt]
    \frac{(u)_m}{m^{|u|}} &
    1-\frac{(u)_m}{m^{|u|}}
  \end{pmatrix}
  \text{, for example }
  A(\texttt{234}_{10}) =
  \begin{pmatrix}
    0.0001 & 0.9999 \\[5pt]
    0.234\phantom0 &
    0.766\phantom0
  \end{pmatrix}.
\end{displaymath}
With starting and ending state both 1.
(This is applied to one part $v_i$ of the word pairs $(v_i,w_i)$.)
Is this just a mistake? Or does the automaton do some useful function?

Note that a radix $m$ (in the paper $K+1$) larger than 2 is required. See the sentence
after Th.2 (p.110): The digit $m-1$ must not appear.

\item In the reduction from PCP to MPCP refers to ``in the same way as [3]''.
  [3] Hopcroft, Ullmann, 1969.
If this is the standard trick like everywhere else, this works for
solvability!
(as stated on top of p. 109).
But not for the purpose in the paper.
Even if the TM halts, the PCP (resulting from the MPCP) will have
arbitrarily long common prefixes, simply by repeating the PCP solution
periodically and introducing a deviation at some point.
\end{enumerate}

Interesting:
Uses a TM model in which the machine either changes state or makes a
(left or right) move, but not both,
with reference to [4] Nelson: Introduction to Automata, Wiley 1968.

\paragraph{Bertoni et al. 1977}
There is also
Bertoni et al. 1977, ICALP, where two other results are treated
\begin{enumerate}
\item The problem ``is 1/2 an isolated cutpoint'' is semidecidable.
\item The problem ``is there an isolated cutpoint'' is undecidable.
\end{enumerate}

Bertoni et al. 1977, top of p.90, (3.1) write correcly
 a $k$-ary automaton, (with 2 states; this works!)

\subsection{Gimbert and Oualhadj 2010:PFA}
\cite{gimbert-oualhadj-2010:PFA}
\footnote{\url{https://hal.science/hal-00456538v3/file/gimbert_oualhadj_probabilistic_automata.pdf}}

Mistakenly credit
Alberto Bertoni 1975
\cite{Bertoni1975} for the result

``\textbf{Proposition 1 (Bertoni).} The equality problem is
undecidable'' (4th page)

The proof is given on the first page of the appendix in the
tech-report (online version).
(Forgot to exclude the empty word from the PCP solutions).

[Cf. also Paz \cite[Lemma IIIB.6.1, p.~183]{paz71}: ``Every $E$-event
can be represented in the form $\phi=\frac12$.'']

Reduction from ``Equality problem'' to Emptiness Problem is ``less
known'':
From testing $\phi=\frac12$ to $\phi(1-\phi)\ge \frac12$.

Interesting: Convex combination is done with an additional letter at
the beginning.

``Bertoni [4] proved that the Isolation Problem is undecidable in general:

\noindent
\textbf{Theorem 3 (Bertoni).} The Isolation Problem is undecidable for probabilistic
automata with five states.''

However, Bertoni
\cite[Th. 4, p.111]{Bertoni1975}:
``It is recursively unsolvable if $\frac12$ is an isolated cutpoint for
arbitrary 4-states stochastic automata.'' 5 or 4?

\subsection{Amplification (Case study 2)}
\label{sec:amplification-appendix}

idea
credited to
\cite{baier08}

Christel Baier, Nathalie Bertrand, and Marcus Größer. On decision problems for
probabilistic Büchi automata. In FoSSaCS, pages 287–301, 2008.

in the rendition of
Fijalkow
\cite{Fijalkow-2017-SIGLOG},
one begins to recognize ideas similar to the Freivalds construction.

similar to
procedure of
Freivalds discussed in Section ...

state-chart.

Fig. 1 on 7th page.

WOANDERS:
probability amplification technique
of Gimbert and Oualhadj
\cite{gimbert-oualhadj-2010:PFA} (Theorem~\ref{thm-amplify} in
Section~\ref{sec:amplification-study}),

one can transform
into dichotomy that is even stronger than
the one of
 Theorem~\ref{dichotomy} in Section~\ref{sec:boost},
 and Section~\ref{sec:further-boost}):

  reminiscent of ...

  ``success

\eqref{eq:plus}
\eqref{eq:minus}

\hrule

\paragraph{Gimbert and Oualhadj, ICALP 2010
\cite{gimbert-oualhadj-2010:PFA}}

[They construct $\mathcal{A}$ from $\mathcal{B}$.]

\texttt{sim} and \texttt{check} are called $a$ and $b$.

A single coin is flipped at the beginning.!! more complicated calculation.

The machine flips a single coin at the beginning, and with probability 1/2, it does one of the
following two things:
\begin{itemize}
\item [($+$)]
 The machine accepts if, in some round,
  the output of $A$  for \emph{all} inputs
  $u_1,\ldots,u_n$ in that round
 was ``Plus''; otherwise it rejects.
\item [($-$)]
 The machine rejects if, in some round,
  the output of $A$  for \emph{all} inputs
  $u_1,\ldots,u_n$ in that round
 was ``Minus''; otherwise it accepts.
\end{itemize}

(The change of behaviour is not mentioned in \cite{Fijalkow-2017-SIGLOG}.)

On top of p.~16, when discussing the automaton of Figure~1,
[Section 7 in the appendix],
the $b$'s come at the \emph{end} of each round.
As written this makes no sense: A word ending in $b$ is never accepted
in the left half of the automaton; thus the total acceptance probability
is at most $1/2$, and it cannot converge to 1, as claimed.

two copies of PFA B are integrated into A.
The transitions are described in detail in the text.

After everything is integrated into the new PFA A

[Proof of Theorem 4, Section 7 in the appendix, p.~16], the $b$'s come
at the beginning and not at the end. This is more reasonable.
On the other hand leaps of argument:

Factorize $w'$ in $w' = u_0v_0\sharp u_1v_1\sharp u_kv_k\cdots$
such that $u_i \in b^*$ and $v_i\in A^*$.
Such a factorization is not always possible.

There is also small oversight. If $u$ leads to the starting state
(which must by assumption be rejecting.)

``(moving to an absorbing rejecting state)''
\paragraph{Fijalkow 2017
  \cite{Fijalkow-2017-SIGLOG}}
short note, 8 pages

Fijalkow 2017 does not point out the difference to G+O.

pictorial representation of an abstraction
Each transition \texttt{sim} is supposed to represent
a run of $A$, and is supposed to have a fixed success probability $x$.

[He constructs $\mathcal{B}$ from $\mathcal{A}$.]

In Gimbert and Oualhadj
\cite{gimbert-oualhadj-2010:PFA}
(translated to Fijalkow's language)
the transitions from $\mathcal{A}$ to the absorbing states
upon reading \texttt{check}/$b$ explicitly
are performed only from the starting state. (Other transitions lead to
the rejecting sink.)

In the pictorial description of
Fijalkow 2017
\cite[
Figure 2]{Fijalkow-2017-SIGLOG},
one must understand that $L$ and $R$ are the starting states of
the respective automata.
The transition from any nonaccepting state of $\mathcal{A}$
othe than the starting state on
 reading \texttt{check} is unspecified.

 Upon reading $
 \texttt{check}\cdot w\cdot \mathtt{check}$,
where $w$ is a word that may lead back to the starting state of
$\mathcal{A}$ with positive probability, (or to a rejecting state,
with the loose interpretation of Fig.~2),
the automaton $\mathcal{B}$ has a positive acceptance probability,
 contrary to what is claimed
 on p.16.

One could easily fix this by insisting that $\mathcal{A}$ never goes
back to the start state. This has been swept under the rug.

``\textbf{The expanding automaton}

The automaton presented on the right-hand side of Figure 1 was used in
the proof of the undecidability of the value 1 problem [Gimbert and
Oualhadj 2010].'' (p.13 bottom)

``The construction is essentially the same as for the undecidability of
the value 1 problem by Gimbert and Oualhadj [Gimbert and Oualhadj
2010], the main improvements are in the correctness proof.'' [p.15--16]
\paragraph{Even simpler}
no coin at all.
One has to specify what happens for an empty round ($n=0$).
We may specify that an empty round leads to
immediate rejection (REJECT).
(According to the original automaton, it leads to
an immediate decision
ACCEPT or REJECT with probability $1/2$ each.

\paragraph{Moral}

Constructing a PFA for a particular purpose via a state diagram is
tricky.
Better stick with a high-level verbal description.

\subsection{6x6 with fractions?}

The matrix~\eqref{eq:gamma-matrix}
 also works with $10^{-x}$
 instead of $10^{x}$, perhaps with reversed order.
Experiments have confirmed this.
Then
 one is back to microscopic differences that decided about PFA acceptance.

\subsection{Wasting states}

Blondel and Canterini 2003
[Step 2, p.~235]
Undecidable problems for probabilistic automata
of fixed dimension
\cite{blondel-canterini-03}

following
V. D. Blondel and J. N. Tsitsiklis. When is a pair of matrices mortal? Inform. Process. Lett.,
63(5):283--286, 1997.

multiply the number of states by $k$.

The same is used later in 1998
Cassaigne and Karhumäki \cite{CASSAIGNE199829},
who cite
Blondel and Tsitsiklis 1997.

This is particularly striking as the main objective of their paper is
to get a small number of states.

\subsection{SOMETHING ELSE: Using small universal TM}

Blondel and Canterini 2003
\cite[Theorem 3.2, p.~241]{blondel-canterini-03}
using a particular small universal TM!
In a different context, and for a different purpose:
To get a bound for the word pairs for \emph{any} $\omega$-PCP
(with fixed pairs or not.)

``We obtain our result by combining the proof of Bertoni [B1=1975] with a
universal Turing machine encoding and sharp bounds on small universal
Turing machines taken from [R4].'' Rogozhin,
ten states and three letters.
Hier ist der Verweis of Bertoni [B1] richtig!

... infinite word W whose first letter is given and such that h(W) =
g(W) if and only if U halts on x.
(In the construction it is stated correctly: If the machine $U$ does not halt.)

No padding rules! forgotten? but \texttt{\#} (separation symbol?) is
counted as an extra letter. Is OK. $40=4\times10$ rules in total.
cases with \texttt{\#} are not treated correctly, but the
count of word pairs is fine.
Could have saved 3 by swapping left and right rules.

\subsection{Problem of representation of matrices}

the question whether a matrix $U$ can be represented by matrices
$U_1,\ldots,U_q$ using multiplication.

``partial problem'': Only $U$ is variable.

Subsequently (cf. [3]) the number of matrices in the system was
reduced to 23 and it was proved that, with an appropriate complication
in the construction of the system, the condition $n\ge6$ could be weakened
to $n\ge4$. For any $n\ge 6$ one can construct a concrete system, consisting of
12 matrices of order $n$, with unsolvable partial problem (cf. [4]). By
appropriately fixing $U$ and varying $U_1,\ldots,U_q$ the unsolvability of the
general formulation has been proved for n=3 (cf. [5]).

[3] 	A.A. Markov, ``On the problem of presenting matrices'' Z. Math. Logik und Grundl. Math. , 4 (1958) pp. 157--168 (In Russian) (German abstract)
[4] 	N.M. Nagornyi, , 6-th All-Union Congress on Math. Logic , Tbilisi (1982) pp. 124 (In Russian)
[5] 	M.S. Paterson, "Unsolvability in $3\times 3$ matrices" Stud. in Appl. Math., 49 : 1 (1970) pp. 105--107

Representation of matrices, problem of. Encyclopedia of Mathematics.
\url{http://encyclopediaofmath.org/index.php?title=Representation_of_matrices,_problem_of&oldid=32622}

\cite{halava07-Claus-instances}
Halava,  Harju, and Hirvensalo:
title = {Undecidability bounds for integer matrices using {Claus} instances},
\url{https://typeset.io/pdf/undecidability-bounds-for-integer-matrices-using-claus-1nwnni898j.pdf}
 International Journal of Foundations of Computer Science Vol. 18,
 No. 05, pp. 931-948 (2007)
 
Theorem 6. 2MPCP with 7 words in total. Everything is fixed except ONE
word from the starting pair.

\subsection{Paterson 1970}

 Paterson~\cite{paterson70} in 1970: mortality for $3\times 3$
 matrices is undecidable.

 PCP

 \begin{equation}
   \label{eq:Paterson}
   \begin{pmatrix}
   10^{|v|} &0&0\\
   0&  10^{|w|}&0\\
   (v)_{10}
&   (w)_{10} & 1   
   \end{pmatrix}
 \end{equation}

 Paterson~\cite{paterson70}
 specified this matrix only by way of example, without
 committing to a particular radix: ``The notation is quaternary, decimal, etc., according
to taste.''

The argument can be simplified by replacing $S$ and $T$ by a single
rank-1 matrix
\begin{equation}
  \label{eq:new-ST}
  \begin{pmatrix}
    1\\-1\\0
  \end{pmatrix}
  \begin{pmatrix}
    1&0&1
  \end{pmatrix}
  =
  \begin{pmatrix}
    1&0&1\\
    -1&0&-1\\
    0&0&0    
  \end{pmatrix}
\end{equation}
 Then one need not discuss 4 cases but just a single case.

\subsection{Turning a semi-Thue system with $k$ rules into a 2MPCP
  with $k+4$ pairs}
\label{sec:semi-Thue}

A = matrix([[113, 113, 1469],
[1938, 0, -7910],
[442, 113, 113]])

Conversion according to Claus \cite{claus80-Bull-EATCS}, Bull. EATCS, 1980.

simplified by converting to the (2)MPCP.

Assume rules in arbitrary (binary at this point not necessary)
alphabet $X$.
Starting word $u_0\in X^*$, ending word  $\bar u\in X^*$.

add separator symbol \texttt{\#}.

encode $X\cup \{\texttt{\#}\} \to \{\mathtt{a},\mathtt{b}\}^*$ in binary, for
example using the codewords $\mathtt{bab},\mathtt{baab},\mathtt{baaab},\ldots$ over the
alphabet $\{\mathtt{a},\mathtt{b}\}$. No codeword is a substring of
another codeword, and the strings that are composed from such codewords
can always be uniquely decoded.

Claus \cite[p.57]{claus80-Bull-EATCS},
applies a homomorphism $h$ with three codewords
\texttt{01}, \texttt{011}, \texttt{0111}.
Although this is injective, it can lead to mistakes.
The first two bits of
 \texttt{0111} can be ``mistaken'' for
\texttt{01}, and a rule for
\texttt{01} can be applied. Afterwards we can copy the two remaing bits \texttt{11} just so.

$w \in (X\cup \{\texttt{\#}\})^* \mapsto \langle w\rangle \in \{\mathtt{a},\mathtt{b}\}^*$ in binary.

Word pairs:
\begin{compactitem}
\item For every rule $v\to w$: $(\langle v\rangle,\langle w\rangle)$
\item $(a,a)$ and $(b,b)$. [First explain without binary coding: too many copying rules]
\item Starting rule $(v_1,w_1)= (\langle \texttt{\#}\rangle, \langle\texttt{\#}u_0\texttt{\#}\rangle)$
\item Ending rule $(v_\infty,w_\infty)= 
  (  \langle\texttt{\#}\bar u\texttt{\#}\rangle,
  \langle \texttt{\#}\rangle)$  
\end{compactitem}

\cite{halava07-Claus-instances}
Halava,  Harju, and Hirvensalo:
title = {Undecidability bounds for integer matrices using {Claus} instances},
\url{https://typeset.io/pdf/undecidability-bounds-for-integer-matrices-using-claus-1nwnni898j.pdf}

\subsection{Decision algorithm for 2x2 matrices, according to Claus 1981}
\label{sec:decide-2x2}

Claus \cite[Theorem 7 and Corollary, p.~155]{claus81} show that
the emptiness question can be decided for
two-state PFAs.

Transformation of the problem. One variable $x$ = the probability of being
in the accepting state, is sufficient.

Equivalent formulation: Given a family of linear functions $f_i\colon[0,1]\to[0,1]$,
for $i\in I$, a starting
value $a \in [0,1]$, and a threshold $\lambda$,
is there a sequence $i_1\ldots i_k\in I^*$ such that
$(
f_{i_1}\circ
f_{i_2}\circ
\cdots \circ
f_{i_k})(a)
=f_{i_k}(
f_{i_{k-1}}(\dots
(f_{i_2}(
f_{i_1}(a)))\dots))  
> \lambda$?\footnote
{Easy approach: 
  iterate the functions $f_i$ and
enlarge the max/min of reachable values, starting with~$a$.}

For a word $w=i_1\ldots i_k\in I^*$ we write
$f_w$ for $
f_{i_1}\circ
f_{i_2}\circ
\cdots \circ
f_{i_k}$.

Exclude slopes $\pm 1$ by preprocessing.\footnote
{Claus \cite
  {claus81} is a bit overzealous: Since two starting vectors $(1-a,a)$ and
  $(a,1-a)$ are tried,
  the matrix $\left(
    \begin{smallmatrix}
      0&1\\1&0\\
    \end{smallmatrix}\right)$ only needs to be postmultiplied with the
  other matrices. The premultiplied versions are unnecessary.}
Assume slopes are in the interval $(-1,+1)$.

Distinguish
\emph{negative functions} $f_i, i\in I_-$ with negative slope,
and
\emph{positive functions} $f_i, i\in I_+$ with nonnegative slope.
(Or \emph{rising} and \emph{falling}?)

Negative composed with negative gives positive, etc. (The slopes are
multiplied under composition.)

Every function $f_w$ has a unique fixpoint, denoted by $s_w$: $f_w(s_w)=s_w$.
Every function is a contraction towards its fixpoint.
See Figure~\ref{fig:iteration-towards-fixpoint}: 
\begin{figure}[htb]
  \centering
  \includegraphics[scale=1]{iteration-towards-fixpoint}
  \caption{Convergence towards $s_w$ for positive and negative functions $f_w$.}
  \label{fig:iteration-towards-fixpoint}
\end{figure}

Define ``attraction intervals'':
The interval
$A_0$ is spanned by the fixpoints for certain words of length 1 and 2.
The interval
$A$ is spanned by what can be reached in one step from $A_0$.
\begin{align*}
A_0 &:= \mathrm{conv}(\{\,s_i\mid i\in I\,\}
      \cup \{\,s_{ij}\mid i,j\in I_-\,\}) =: [\underline s, \overline s]\\
A  &:= \mathrm{conv}( A_0 \cup \{\, f_i(s) \mid s \in A_0, i \in I\,\}
     ) \\
&= \mathrm{conv}( A_0 \cup \{\, f_i(s) \mid s \in \{\underline s, \overline s\}, i \in I\,\}
     ) \\  
&= \mathrm{conv}( A_0 \cup \{\, f_i(s) \mid s \in \{\underline s, \overline s\}, i \in I_-\,\}
     ) =: [\underline t, \overline t]
\end{align*}
The last equation holds because a positive function $f_i$ contracts
towards $s_i$ and remains on the same side of $s_i$ and hence cannot
lead out of $A_0$.\footnote
{If $f_v$ and $f_w$ are positive, $s_{vw}$ lies in the interval
$\mathrm{conv}\{s_v,s_w\}$ (Corollary of Lemma 8/9 for the subsystem
$\{f_v,f_w\}$.) Useful?}

Any fixpoint $s_w$, as well as
$\underline t$ and $\overline t$ can be reached arbitrarily closely
from any starting point $a$.

\noindent\textbf{Lemma 8 and 9.} Cannot leave the interval
$A=[\underline t, \overline t]$:
$\underline t \le f_j(\overline t) \le \overline t$
and
$\underline t \le f_j(\underline t) \le \overline t$.

Proof: The cases $j\in I_+$ and $\underline t = \underline s$ (by
definition) are easy.
Interesting case:  $j\in I_-$, and $\underline t = f_i(\overline s)$ for $i\in I_-$.
Then  $f_j(\underline t) = f_{ij}(\overline s)$, and since $s_{ij}\in
A_0$ and $f_{ij}$ is positive,
 $f_{ij}$ contracts towards a point in $A_0$,
and therefore
$f_j(\underline t) = f_{ij}(\overline s)
\le \overline s\le \overline t$.\qed

Corollary: $\underline t \le s_w \le \overline t$ for all $w$.

Observation: If $a\notin [\underline t, \overline t]$, then the extreme values that
can be reached (at least arbitrarily closely) are in the set
$\{a\} \cup
\{\,f_i(a)\mid i\in I_-\,\}
 \cup\{ \underline t, \overline t\}
$.
(The same is trivially true
if $a\in [\underline t, \overline t]$.)

\subsection{Leibniz an Tschirnhaus, 1678}
\label{sec:leibniz}

LEIBNIZ AN E. W. VON TSCHIRNHAUS, Ende V./Anfang VI. 1678

In signis spectanda est commoditas ad inveniendum, quae maxima est
quoties rei naturam intimam paucis exprimunt, et velut pingunt, ita
enim mirifice imminuitur cogitandi labor.
\cite[p.~444--445]{Leibniz}

In signs, one has to pay attention to the
suitability
for 
discovery,
which is greatest 
if 
the signs
express
the inner nature 
of things concisely, as though
painting it, 
for 
the labor of thinking is
thus 
wonderfully reduced.

G. W. Leibniz in a letter to E. W. von Tschirnhaus, 1678

Notio und Notatio.
Aufzeichnung zu einem Vortrag in Heidelberg im Juni 1941.
Von HELLMUTH KNESER in Tübingen.
Jahresber. Deutsch. Math.-Verein. 53 (1943), 9--21
[Zbl 0028.10202].
\url{https://doi.org/10.1515/9783110894516.377}

about the use of (not: subscript!) indices
indicating the row and column
when writing a system of equations.

Haec observatio omnium, quae de calculo fieri possunt, maximi momenta
est: ita facile et velut sponte sua se detegunt praeclara theoremata
et progressiones, ita ut pleraque sine calculo scribi possint, initiis
tantum praelibatis. (Praelibare heißt vorwegnehmen). Leibnizens
mathematische Schriften, hrsg. von C. J. Gerhardt (Dritte Folge in:
Leibnizens gesammelte Werke, aus den Handschriften der Königlichen
Bibliotek zu Hannover, hrsg. von Georg Heinrich Pertz), 2. Abt. 3
(1863), S. 7. Im folgenden nenne ich diese Ausgabe kurz als „Leibniz,
Math. Schr.''

Bei den Zeichen muß man auf ihre Eignung zum Erfinden achten. Diese
ist am größten, wenn die Zeichen das Wesen der Sache in Kürze ausdrücken
und gleichsam abbilden; denn so verringert sich die Denkarbeit in
wundersamem Maße.

The latin citation in note [11] misses one word:

``In signis spectanda est commoditas ad inveniendum, quae maxima est
quoties rei naturam paucis exprimunt et velut pingunt, ita enim
mirifice imminuitur cogitandi labor.'' Leibniz, Math. Schr., 1. Abt. 4, S. 455.

\paragraph{Old proof of weaker version of
  Eliminating the output vector, proof of
  Theorem
  ~\ref{thm-fixed-f}a}

We will use an adaptation of this procedure to prove part~(a) of the
following theorem.

 (a) We first explain a simple version, where the number of states
 is
multiplied by 8.

We start with the set $\mathcal{M}'''$ of 52 $9\times 9$ matrices of
Theorem~\ref{thm-fixed-pi}a.
We simulate the corresponding automaton, but simultaneously, with every
input symbol, we throw an octahedral die. In other words, we
determine a random number $r$  between 1 and~8.
The acceptance criterion is as follows:
If we are in state $q$ and the output value is $f_q=\frac i8$, we
accept if the input if $r\in \{1,\ldots,i\}$.
In other words, the randomness for making the acceptance decision is
already
generated \emph{before} the last symbol is read, and we simply need to
use it by looking up the value $r$.
It is clear that the input is accepted with the correct probability
 $f_q=\frac i8$.
 The number of states is multiplied by 8, and each transition matrix
  $M\in\mathcal{M}'''$ is converted to the matrix
\begin{displaymath}
  \begin{pmatrix}
 1/8\cdot M&1/8\cdot M&1/8\cdot M&1/8\cdot M&1/8\cdot M&1/8\cdot M&1/8\cdot M&1/8\cdot M\\
 1/8\cdot M&1/8\cdot M&1/8\cdot M&1/8\cdot M&1/8\cdot M&1/8\cdot M&1/8\cdot M&1/8\cdot M\\
 1/8\cdot M&1/8\cdot M&1/8\cdot M&1/8\cdot M&1/8\cdot M&1/8\cdot M&1/8\cdot M&1/8\cdot M\\
 1/8\cdot M&1/8\cdot M&1/8\cdot M&1/8\cdot M&1/8\cdot M&1/8\cdot M&1/8\cdot M&1/8\cdot M\\
 1/8\cdot M&1/8\cdot M&1/8\cdot M&1/8\cdot M&1/8\cdot M&1/8\cdot M&1/8\cdot M&1/8\cdot M\\
 1/8\cdot M&1/8\cdot M&1/8\cdot M&1/8\cdot M&1/8\cdot M&1/8\cdot M&1/8\cdot M&1/8\cdot M\\
 1/8\cdot M&1/8\cdot M&1/8\cdot M&1/8\cdot M&1/8\cdot M&1/8\cdot M&1/8\cdot M&1/8\cdot M\\
 1/8\cdot M&1/8\cdot M&1/8\cdot M&1/8\cdot M&1/8\cdot M&1/8\cdot M&1/8\cdot M&1/8\cdot M\\
  \end{pmatrix}.
\end{displaymath}
(A more straightforward procedure would be to throw the die once at the
  beginning and remember the result $r$, but then the transition matrix
  would not be positive.)
  
We can save states by making a cleverer, irregular die. All that is
needed is
that every output value $\eeta_q \in\{\frac14,\frac12,\frac58\}$ can
be combined from the face probabilities of the die. ...
\qed

\end{document}